\documentclass[11pt]{article}
\usepackage[margin=0.85in]{geometry}
\usepackage{import}
\usepackage[utf8]{inputenc}
\usepackage[T1]{fontenc}
\usepackage{lmodern}
\usepackage{amsmath,amssymb,amsthm,mathtools,braket}
\usepackage{array,booktabs,tabularx,microtype}
\usepackage{enumitem}
\usepackage{needspace}
\usepackage[hidelinks]{hyperref}
\hypersetup{pdftitle={Quantum Complexity of Solving Linear Equations on Higher-Order Networks},pdfauthor={Caesnan M. G. Leditto}}
\usepackage{cleveref}
\newcommand{\BQP}{\mathsf{BQP}}
\newcommand{\poly}{\mathrm{poly}}
\DeclareMathOperator{\im}{im}
\newcommand{\norm}[1]{\left\lVert #1\right\rVert}

\newcommand{\VerifyProblemSA}{\mathrm{Verify}^{\rm SA}}
\newcommand{\VerifyProblemCirc}{\mathrm{Verify}^{\rm circ}}

\newcommand{\QLSP}{\mathsf{QLSP}}

\newcommand{\BoundedOccurrenceDA}{\mathsf{BO\text{-}DA}}
\newcommand{\BoundaryQLSP}{\QLSP(\partial_2)}
\newcommand{\UpperHodgeQLSP}{\QLSP^{\uparrow}}
\newcommand{\LowerHodgeQLSP}{\QLSP^{\downarrow}}
\newcommand{\FullHodgeQLSP}{\QLSP^{\rm full}}

\newcommand{\FKClockLS}{\QLSP(G_C)}

\newcommand{\Id}{\mathrm{I}}
\newcommand{\Span}{\mathrm{span}}

\newcommand{\nnz}{\mathrm{nnz}}
\newcommand{\spec}{\mathrm{Spec}}
\newcommand{\normmax}[1]{\left\lVert #1\right\rVert_{\max}}
\newcommand{\cond}{\kappa}
\newcommand{\alphabe}{\alpha_{\rm be}}
\newcommand{\dist}{\operatorname{dist}}

\newcommand{\ketbraop}[2]{\ket{#1}\!\bra{#2}}

\theoremstyle{definition}
\newtheorem{definition}{Definition}[section]
\newtheorem{remark}{Remark}[section]

\theoremstyle{plain}
\newtheorem{lemma}{Lemma}[section]
\newtheorem{proposition}{Proposition}[section]
\newtheorem{theorem}{Theorem}
\newtheorem{corollary}{Corollary}

\begin{document}
\title{Quantum Complexity of Solving Linear Equations on Higher-Order Networks}
\author{%
  Caesnan M. G. Leditto$^{1,2}$\thanks{Email: \href{mailto:caesnan.leditto@monash.edu}{\texttt{caesnan.leditto@monash.edu}}}\\[0.8em]
  \small $^{1}$School of Physics and Astronomy, Monash University,\\
  \small Clayton, VIC 3168, Australia\\[0.3em]
  \small $^{2}$Faculty of Science and Computer, Universitas Kristen Immanuel,\\
  \small Yogyakarta, Indonesia
}
\date{}
\maketitle
\begin{abstract}
Higher-order networks (HONs) represent interactions among groups of objects, supporting the study of collective behaviour that pairwise models can miss. In simplicial models of these networks, Hodge Laplacian linear systems provide a common mathematical framework for analysis, for example, for solving statistical ranking problems and investigating long-term stability in coupled oscillator systems. The large number of variables associated with edges, triangles, and higher-dimensional simplices can make these systems costly to solve. Recent quantum algorithms address this bottleneck. However, comparisons with particular classical algorithms do not decision output whether solving these linear system problems establishes a quantum complexity foundation for provable quantum advantage. We prove that preparing a quantum state encoding the normalized minimum-norm solution of Hodge Laplacian linear systems is $\mathsf{BQP}$-hard under the specified sparse oracle model and parameter promises. Our sequence of reductions maps an arbitrary polynomial-time bounded-error quantum computation to such linear systems. Each reduction allows efficient recovery of the normalized minimum-norm solution of the preceding linear system. Together with an efficient quantum algorithm for the associated decision problem from the approximate solution state, we show that this problem is $\BQP$-complete. These results provide a worst-case complexity foundation for evaluating quantum advantage in HON analysis and show that the $\mathsf{BQP}$-hardness of the quantum linear system problem (QLSP) persists even when the matrices are restricted to Hodge Laplacians.
\end{abstract}

\section{Introduction}\label{sec:main-introduction}
Many complex systems involve interactions among groups rather than only
between pairs~\cite{BattistonReview,Battiston2021}. Examples arise in
social networks~\cite{IacopiniKarsaiBarrat2024}, ecological
communities~\cite{Grilli2017}, and coupled oscillator
models~\cite{SkardalArenas2019}. Higher-order networks
(HONs) model these group interactions~\cite{Torres2021,Bianconi2021,Bick2023}.
Simplicial complexes are one such model: vertices represent individual
objects, edges represent pairs, filled triangles represent triples,
and higher-dimensional simplices represent larger groups. Each simplex
includes all its faces. These models help us study how group interactions
affect collective behaviour.

This richer description also creates a computational challenge. Even
at a fixed group size, a network can contain far more interactions than
individual objects. Storing and examining all of them can be costly.
We therefore study tasks that use local queries to the network without
listing all its simplices.

Linear equations arise in several of these tasks. HodgeRank uses
least squares and discrete exterior calculus to recover a global ranking
from pairwise comparisons and identify cyclic inconsistencies, including
when some comparisons are missing~\cite{JiangHodgeRank,DesbrunDEC2005}.
In some higher-order oscillator models, phases are assigned to simplices,
and the analysis of stable states leads to least-squares
problems~\cite{ArnaudonKuramoto,Nurisso2024}. These equations involve
Hodge Laplacians, which extend graph Laplacians to values on edges,
triangles, and higher-dimensional simplices~\cite{HorakJost2013}.
We call them \emph{Hodge Laplacian linear systems}. Their matrices
can be large even when each row has few nonzero entries. In HodgeRank,
the global ranking scores are obtained from a graph Laplacian; our
hardness results concern the higher-dimensional operators specified
in \Cref{thm:main_hardness}.

Quantum linear-system algorithms prepare states whose amplitudes encode
normalized solutions~\cite{HHL2009,CKS2017,GSLW2019}. Measurements can
then estimate selected properties of the solution. Recent algorithms
for HodgeRank and coupled oscillators improve on the costs of specified
classical methods for certain input families, under assumptions on
input access and network structure~\cite{LedittoHodgeRank,LedittoOscillators}.
These comparisons leave open whether another classical algorithm could
perform the same task efficiently.

One way to address this question is to prove $\BQP$-hardness: encode
any efficient quantum decision computation into the task and recover
its answer from the output~\cite{WatrousSurvey}. Harrow, Hassidim,
and Lloyd established this connection for sparse linear systems with
constant sparsity and condition numbers polynomial in the logarithm
of the matrix dimension~\cite{HHL2009}. The hardness persists for
positive-definite matrices~\cite{OrsucciDunjko2021}. Hodge Laplacians
have further restrictions: their entries must come from the faces of
a simplicial complex. Hardness for general positive-definite matrices
does not establish hardness for this family.

Ding, Kyng, Probst Gutenberg, and Zhang (DKGZ) reduce general sparse
linear equations to boundary linear systems and Hodge Laplacian linear systems on
two-dimensional simplicial complexes~\cite{DingFull}. Their
condition-number bounds depend on the number of nonzero entries in
the original matrix. A matrix with a short oracle description can
have exponentially many such entries, so these bounds do not by
themselves ensure the polynomial cost needed here. We need bounds
polynomial in the input description length, efficient sparse access
to the constructed matrices, and efficient recovery from their
normalized minimum-norm solution states. Table~\ref{tab:prior}
summarizes the relation to these results.

\paragraph{Main result.}
We show that Hodge Laplacian linear systems can encode any efficient
quantum decision computation, even on unweighted simplicial complexes
of fixed dimension. We efficiently construct circuits that provide
sparse access to the matrices and prepare the right-hand side. An
efficient quantum procedure then recovers the encoded computation's answer from
an approximate normalized minimum-norm solution state.
The resulting decision problems are $\BQP$-complete when their input
oracles are supplied as circuits. Preparing the solution states is
therefore $\BQP$-hard: repeated preparation, recovery, and measurement
allow the decision problem to be solved with bounded error.
\Cref{thm:main_hardness} states the precise dimension
ranges and input promises for upper, lower, and full Hodge Laplacians.

\paragraph{Proof idea.}
We first encode a quantum circuit in sparse linear equations whose
solution records the computation, following circuit-to-Hamiltonian
constructions~\cite{KitaevShenVyalyi2002,AharonovEtAl2007,BabbushOscillators}.
We then represent these equations by a simplicial boundary matrix.
The main difficulty is that adding variables and copying coordinates
can change the normalized minimum-norm solution. We adapt the DKGZ
construction so that representing variables by triangulated surfaces
preserves the norm of the encoded vector. We also bound the condition
numbers and provide efficient matrix queries and recovery operations
at every step. Finally, we pass from boundary linear systems to Hodge
Laplacian linear systems and extend the construction to higher fixed
dimensions. Sections~\ref{sec:main-boundary}--\ref{sec:main-decision}
give the reductions and quantum algorithms, with detailed proofs in
the appendices.

\paragraph{Significance.}
The result shows that the structure of a simplicial complex and a
fixed dimension do not remove the quantum computational hardness of
linear systems. A classical randomized polynomial-time algorithm for
our circuit-encoded decision problems would give such an algorithm
for every problem in $\BQP$. This is a worst-case result. Practical
advantage still depends on input preparation, the properties of the
instances, and the information extracted from the solution state.

\begin{table}[t]
\caption{Related results and the problem addressed here. Each result
uses its stated input model and promises.}
\label{tab:prior}
\small
\begin{tabularx}{\linewidth}{
@{}>{\raggedright\arraybackslash}p{0.20\linewidth}
>{\raggedright\arraybackslash}p{0.29\linewidth}
>{\raggedright\arraybackslash}X@{}}
\toprule
Work & Matrices and task & Relation to this work\\
\midrule
Harrow--Hassidim--Lloyd~\cite{HHL2009}
& General sparse quantum linear systems
& Proves $\BQP$-hardness for sparse linear systems without requiring
the matrices to come from a simplicial complex.\\
\addlinespace
Orsucci--Dunjko~\cite{OrsucciDunjko2021}
& Positive-definite quantum linear systems
& Proves $\BQP$-hardness with polynomial condition numbers, including
sums of positive-definite terms acting on logarithmically many qubits;
the matrices need not come from a simplicial complex.\\
\addlinespace
Ding et al.~\cite{DingFull}
& Classical reductions to boundary linear systems and Hodge Laplacian linear systems
& Represents sparse linear equations by simplicial complexes; the
condition-number bounds depend on the number of nonzero entries
in the original matrix.\\
\addlinespace
This work
& Boundary linear systems and Hodge Laplacian linear systems on unweighted simplicial
complexes of fixed dimension
& Gives quantum hardness reductions with efficient solution recovery
and condition numbers polynomial in the input description length.\\
\bottomrule
\end{tabularx}
\end{table}

\Needspace{12\baselineskip}
\section{Input models and main results}\label{sec:main-input}
To state our complexity results, we specify the solution state and the
access to its defining data. We then explain the promises that allow
efficient state preparation and recovery of a decision from that state.

\subsection{Hodge Laplacians and quantum solution states}
Choose one orientation for each simplex of an unweighted simplicial
complex $K$. The boundary matrix $B_j:=\partial_j$ sends each oriented
$j$-simplex to the signed sum of its $(j-1)$-dimensional faces. Its
entries are therefore $0,\pm1$. Variables indexed by $k$-simplices
carry values on edges when $k=1$, on triangles when $k=2$, and so on.
The lower and upper Hodge Laplacians and their sum, the full Hodge
Laplacian, act on these values as follows:
\begin{equation}
 L_k^\downarrow=B_k^\top B_k,\qquad
 L_k^\uparrow=B_{k+1}B_{k+1}^\top,\qquad
 \Delta_k=L_k^\downarrow+L_k^\uparrow.
 \label{eq:hodge-three}
\end{equation}
The lower Hodge Laplacian couples simplices that share a face; the upper
Hodge Laplacian couples simplices that are faces of a common higher-dimensional
simplex. The full Hodge Laplacian includes both contributions from the
same complex. An absent boundary matrix contributes zero. At vertices, the
upper Hodge Laplacian is the usual graph Laplacian. Appendix~\ref{sec:model}
gives the orientation conventions, and \Cref{sec:main-hodge} explains
how the three operators act on the Hodge decomposition.

These matrices are positive semidefinite and may have zero eigenvalues. We
therefore use the Moore--Penrose pseudoinverse, which inverts the
nonzero eigenvalues and sets the zero-eigenvalue components to zero.
For any of the three matrices $M$ and a normalized right-hand side
$\ket\eta$, the quantum linear-system problem (QLSP) asks us to prepare
\begin{equation}
 \ket{z_M}=\frac{M^+\ket\eta}{\norm{M^+\ket\eta}}.
 \label{eq:target-state}
\end{equation}
Its amplitudes encode the normalized minimum-norm least-squares
solution. When $\eta$ lies in the range of $M$, this is the normalized
minimum-norm solution of $Mz=\eta$. We write $\LowerHodgeQLSP_k$,
$\UpperHodgeQLSP_k$, and $\FullHodgeQLSP_k$ for the corresponding
tasks. Measurements give access to selected properties of the solution.
For vectors, $\norm{\cdot}$ is
the Euclidean norm; for matrices, it is the spectral norm.

\subsection{Access to the input and computational cost}
The number of simplices can be too large to list. We use sparse-matrix
oracle access, as in quantum algorithms for Hamiltonian simulation
and linear systems~\cite{BCK2015,CKS2017}. For each boundary matrix
appearing in $M$, a query specifies a row or column and a position
in its list of nonzero entries, and returns the corresponding index
and value. These queries can be applied to superpositions. The input gives the matrix dimensions, positive integer
row and column sparsity bounds, and the precision of the returned
values. The row and column lists describe the same entries and values. A
unitary $\mathcal O_\eta$ prepares
$\ket\eta=\mathcal O_\eta\ket0$, and its inverse is also supplied.
Both matrix queries and state preparation return their workspace
qubits to zero. \Cref{def:sparse-access,def:rhs-preparation-error}
specify these interfaces. The exact state prepared by
$\mathcal O_\eta$ determines the target in Eq.~\eqref{eq:target-state}.

We distinguish two ways of supplying this access. In the
\emph{sparse-oracle model}, the oracles are supplied together with a
classical description and a size parameter $\nu$, written in unary.
The description length, index-register widths, and numerical data
obey polynomial bounds in $\nu$; the matrices themselves may have
exponentially many rows and columns. Computational cost counts oracle
queries and additional gates, including all calls used for preparation
and recovery. The description contains only the dimensions, indexing
rules, preparation rules, and parameters permitted by
\Cref{def:instance-encoding}. The precise size requirements are in
\Cref{def:sparse-access-relativized-model}.

In the \emph{circuit-encoded model}, the input additionally supplies
descriptions, of size polynomial in $\nu$, of the circuits implementing
every oracle.
It is thus a finite bit string $x$, and cost is measured in its full
length $|x|$. Executing the supplied circuits converts query bounds to
gate bounds. \Cref{def:circuit-encoded-input} specifies the encoding,
with the fixed gate convention of \Cref{rem:implementation-gate-set}.
This distinction lets us state query guarantees for oracle inputs and
ordinary complexity-class completeness for finite inputs.

\subsection{Rescaling, conditioning, and solution accuracy}
The quantum algorithms use matrices rescaled to have norm at most one.
For a matrix $A$, write $\alpha(A)>0$ for the known factor satisfying
$\norm A\le\alpha(A)$. For an unweighted boundary matrix $B$ with supplied
row and column sparsity bounds $s_{\rm row},s_{\rm col}$, our convention
fixes $\alpha(B)=\sqrt{s_{\rm row}s_{\rm col}}$.
For a lower or upper Hodge Laplacian, the rescaling factor is the
square of the factor for its boundary matrix. For the full Hodge
Laplacian, it is the sum of these squared factors, with absent
boundary matrices omitted. \Cref{def:be-normalization} gives the
formulas, and \Cref{sec:main-decision} explains their quantum
implementation. The equations use the unweighted matrices themselves.
Since $(M/\alpha(M))^+=\alpha(M)M^+$, this rescaling leaves the
normalized solution state unchanged.

Small positive eigenvalues matter because solving the system inverts
them. The usual condition number on the nonzero spectrum is
$\kappa(M)=\norm M/\lambda_{\min}^+(M)$, where
$\lambda_{\min}^+(M)$ is the smallest positive eigenvalue. For the
algorithm, we need a bound relative to the known rescaling factor.
The input therefore supplies $\kappa\ge1$ satisfying
\[
 \kappa(M)\le\frac{\alpha(M)}{\lambda_{\min}^+(M)}\le\kappa.
\]
Equivalently, every nonzero eigenvalue of $M/\alpha(M)$ lies in
$[1/\kappa,1]$. The supplied parameter can exceed the usual condition
number because $\alpha(M)$ can exceed $\norm M$.

We also need enough of the right-hand side to survive removal of the
kernel. Let $\Pi_M$ project onto the range of $M$. The input supplies
$\zeta_\eta\ge1$ and promises
\[
 \norm{\Pi_M\eta}\ge\frac1{\zeta_\eta}>0.
\]
This bound ensures that $M^+\eta$ is nonzero and controls the cost
of preparing its normalized state. It allows a component of $\eta$
in the kernel. If the normalized right-hand side lies in the range
of $M$, we may take $\zeta_\eta=1$. \Cref{def:rhs-preparation-error} gives the details.

The requested state error is $0<\varepsilon\le1/2$, measured by the
Euclidean distance between normalized states after ignoring a global
phase. A measurement outcome indicates whether state preparation
succeeded. The accuracy requirement applies to the output conditioned
on the algorithm reporting success~\cite[Problem~1]{CKS2017}.
Sparsity bounds, $\kappa$, $\zeta_\eta$, and $\varepsilon^{-1}$ must
be polynomially bounded in $\nu$, or in $|x|$ for finite inputs.
The rescaling factors are then polynomially bounded as well.
Spectral promises concern the exact matrices; errors in returned
entries and implemented gates are included separately in the
algorithm's error budget. \Cref{def:quphls,def:lower-full-hodge}
give the complete task definitions. The probability that the quantum solver succeeds and its accuracy
are bounded in \Cref{thm:qsvt-final}.

\subsection{Recovery and the associated decision problem}
State preparation has a quantum output. To state a completeness
result for a decision problem, we specify how that output is used.
A fixed recovery and measurement procedure, defined in
\Cref{def:measure-wrapper}, extracts the encoded computation on the
instances constructed by the reductions. The recovered state includes
a \emph{clock register}, whose basis values record how many circuit
gates have been applied. A specified final interval of clock values
corresponds to a completed computation; \Cref{sec:main-boundary}
explains how the circuit encoding gives this interval.

The input label $\tau$ selects the Hodge Laplacian type and the
permitted row, column, clock, and output-register indices. It contains
no separate source circuit or arbitrary verification circuit. Recovery
is generated efficiently from the final instance's description and
uses only its supplied oracles. Queries to matrices from earlier
reduction steps must also be implemented through these final oracles.
\Cref{lem:reverse-typed-oracle-simulation} establishes this
implementation on the constructed instances.

The procedure measures the clock register and then the designated
output qubit if the clock value lies in the final interval. If recovery
fails or the clock value is outside this interval, it returns an
independent fair random bit. Let $p_\star$ be the resulting acceptance
probability on the exact state in Eq.~\eqref{eq:target-state}, including
all these outcomes. The decision task distinguishes
\[
 p_\star\ge\tfrac12+\Delta
 \quad\text{from}\quad
 p_\star\le\tfrac12-\Delta,
\]
where the supplied $\Delta>0$ has polynomially bounded inverse.
As a promise problem, it requires a correct answer only when the
input satisfies the stated assumptions and one of these alternatives.

The probability that recovery succeeds is different from the
probability of answering the decision problem correctly. On the
constructed instances, recovery succeeds with probability at least
$1/\poly(\nu)$. The circuit encoding also gives constant probability
of reaching the final clock interval after recovery. Returning a fair
random bit on the other outcomes therefore leaves an inverse-polynomial
acceptance gap. Repeating state preparation,
recovery, and measurement then gives a decision with error at most $1/3$.
\Cref{lem:composed-heralded-accounting} establishes this connection.
Errors in the prepared solution state and in the quantum operations,
together with the probability that solution preparation fails, are
bounded so that they change the acceptance probability by at most
$\Delta/2$. \Cref{lem:state-to-verifier-final,cor:complete-final}
show that the required accuracy and repetitions have polynomial cost.

In the proofs for lower and full Hodge Laplacian linear systems,
we compare the prepared right-hand side with a normalized reference
vector in the range of the matrix. \Cref{lem:rhs-robustness} bounds the resulting
change in the normalized solution. This reference vector is used only
in the proof; the actual prepared state defines the problem.

We use $\BQP$ for bounded-error quantum promise problems, with the
convention in Appendix~\ref{sec:model}. For oracle inputs, the hardness
reductions are classical polynomial-time algorithms that construct the
instance description and circuits implementing every required oracle,
as specified in \Cref{def:state-red}. An algorithm that prepares an
accurate solution state can be used to solve the decision problem above
by repeating the prescribed recovery and measurement on newly prepared
states. Decision hardness therefore establishes state-preparation
hardness in the sense of
\Cref{def:state-red}. These are consequences of the same reduction.
Completeness concerns the decision problem with a finite circuit-encoded
input.

\begin{theorem}[Hardness and quantum algorithms]\label{thm:main_hardness}
Under the preceding promises, the prescribed circuit-encoded decision
problems are $\BQP$-complete under classical polynomial-time many-one
reductions for upper Hodge Laplacians in every fixed dimension $k\ge1$,
and for lower and full Hodge Laplacians in every fixed dimension $k\ge2$.
In the sparse-oracle model, these decision problems are $\BQP$-hard
under the reductions described above and have quantum algorithms with
polynomially many queries and additional gates and error at most $1/3$.

Consequently, solution-state preparation is $\BQP$-hard at the
prescribed inverse-polynomial accuracy. For every promised input, a
quantum algorithm reports success with probability at least $2/3$ and,
conditioned on success, outputs a state within error $\varepsilon$ of
the target, using polynomially many queries and additional gates.
\end{theorem}

The proofs appear in \Cref{cor:complete-final,cor:lower-full-hodge-hardness}.
At fixed $k$, the hard instances have constant row and column sparsity,
constant rescaling factors, and classical reversible matrix-query
circuits. Their right-hand sides use only constantly many basis
states and a fixed finite set of amplitude patterns, as specified by
the gate convention above. The computation is encoded in the boundary
matrices. In the basic construction for full Hodge Laplacian linear
systems, the upper Hodge Laplacian is zero; \Cref{cor:full-both-nonzero} also gives hardness with both
terms nonzero by adding a fixed simplex as a separate component.

\section{From quantum circuits to difference--average equations}\label{sec:main-boundary}
To prove hardness, we start with a quantum circuit and construct a linear
system whose normalized solution lets us recover the circuit's answer.
The matrix may be exponentially large, so the construction must provide
efficient sparse access and keep its condition number polynomial in the
circuit size. We first achieve these requirements with a Dirichlet
propagation linear system, then convert its equations to
bounded-occurrence difference--average equations. This prepares the
input for the simplicial construction in \Cref{sec:main-transfer}.

\subsection{The Dirichlet propagation linear system}
We may start with a circuit $C=U_L\cdots U_1$ on $n_q$ qubits over the fixed real rational
gate set of \Cref{def:rational-gate-set}: circuit acceptance over this
set is $\BQP$-complete by \Cref{lem:bqpcirc-q-complete}. We use real entries
because the later reduction produces real boundary matrices, and
rational entries allow us to obtain integer equations by multiplying
by a fixed denominator.

We introduce a clock register that records the number of gates already
applied. To make the final output accessible, we append $L$ identity
gates and use $m=2L+1$ clock values. Clock value $\ell$ corresponds to
the state after $\ell-1$ gates. Following the Feynman--Kitaev
construction, we couple adjacent clock values through the gate between
them. The resulting system is
\[
\begin{aligned}
 K_C&=2I-\sum_{\ell=1}^{m-1}
 \left(\ket{\ell+1}\!\bra\ell\otimes U_\ell
       +\ket\ell\!\bra{\ell+1}\otimes U_\ell^\top\right),\\
 K_Cx_C&=b_C,\qquad b_C=\ket1\ket{0^{n_q}},
\end{aligned}
\]
where $U_\ell=I$ for $L<\ell\le2L$.
The right-hand side places the initial computational state at the first
clock value. Each row of $K_C$ involves only its own clock value and
at most two neighbouring values, with entries determined by individual
gates. Because each gate acts on at most three qubits, the matrix has
constant row and column sparsity, and its entries can be queried
efficiently from the circuit description
(\Cref{thm:shifted-clock-hard}).

The Dirichlet conditions make this matrix invertible and give a direct
bound on its condition number. After an orthogonal change of coordinates
that removes the gates, $K_C$ becomes the scalar path matrix
$D_m=\operatorname{tridiag}(-1,2,-1)$ tensored with the identity.
In these coordinates, the Dirichlet conditions set the values at the
two exterior clock positions to zero. The eigenvalues therefore depend
only on the clock length, and
\Cref{lem:clock-conjugation,lem:path-green} show that it is positive
definite with condition number $\Theta(L^2)$. This change of coordinates
is used to analyse the matrix; sparse queries require only the
individual gates described above.

The same path calculation explains what the solution contains.
Let $\ket{\psi_0}=\ket{0^{n_q}}$ and
$\ket{\psi_\ell}=U_\ell\cdots U_1\ket{0^{n_q}}$ for $\ell\ge1$.
By \Cref{lem:weighted-history},
\begin{equation}
 x_C=K_C^{-1}b_C=\sum_{\ell=1}^m
 \frac{m+1-\ell}{m+1}\ket\ell\ket{\psi_{\ell-1}}.
 \label{eq:history}
\end{equation}
After normalization, this is a weighted history state: a superposition
of intermediate computational states, each paired with its clock value.
The identity gates leave the completed computation unchanged, so every
clock value in $\mathcal T=\{L+1,\ldots,2L+1\}$ carries the final
computational state. \Cref{lem:weighted-history} shows that measuring
the clock gives a value in this interval with probability at least
$1/8$, independently of circuit length.

We can therefore decide circuit acceptance by measuring the output
qubit when the clock lies in $\mathcal T$ and returning an independent
fair random bit otherwise. This retains a constant fraction of the
circuit's acceptance bias, so a constant number of repetitions gives a
bounded-error decision. \Cref{thm:shifted-clock-hard} proves that
the distinction survives a sufficiently small constant error in the
prepared solution state. This establishes hardness of the decision
obtained from this measurement and hence of preparing the solution state.
For the next reduction, we write the same equations as $G_Cx_C=c_C$,
where $G_C=5K_C$ and $c_C=5b_C$. The factor $5$ clears the gate
denominators without changing the solution or the condition number;
the resulting matrix entries belong to a fixed finite set of integers
(\Cref{def:shifted-clock}).

\subsection{Bounded-occurrence difference--average equations}
The propagation equations now encode the computation, but their
coefficients are not yet in the form used by the simplicial
construction. We convert them to equations of the forms
\[
 x_u-x_v=\beta,\qquad x_u+x_v-2x_w=0,
\]
while keeping the sum of absolute coefficients in each column bounded
by a constant. Repeated variables are allowed; their
coefficients are collected before constructing the boundary matrix.

Both equation types have coefficients summing to zero. Introducing one
shared offset variable would give this property to the propagation
equations, but that variable would occur in too many rows. We instead
give each row its own offset and impose equality constraints along
a constant-degree expander. Connectivity makes the offsets equal,
while the expander's spectral gap bounds differences between the
offsets when the equality constraints hold only approximately.
Each offset then appears in only
constantly many equations. The construction in \Cref{sec:lg-da}
replaces the resulting integer rows by fixed-size collections of
difference and averaging equations.

Let $M_Cz=d_C$ be the resulting system. Subtracting the corresponding
offset from each original variable defines a recovery map
$D_{\rm aux}$, after discarding the other auxiliary variables.
\Cref{thm:shifted-to-da} proves
\[
 D_{\rm aux}M_C^+d_C=G_C^{-1}c_C,
\]
together with polynomial conditioning and efficient sparse access.
\Cref{lem:lg-mp-decoder,lem:lg-decoder-stability} show that this map
recovers the normalized propagation solution with constant success
probability and controls the error from an approximate input state.
The right-hand side remains supported on a constant number of basis
states by \Cref{lem:lg-rhs-state}. Thus the conversion retains the
circuit information needed for the output measurement above.

\section{Adapting the simplicial construction to quantum linear systems}\label{sec:main-transfer}
The difference--average linear system from \Cref{sec:main-boundary} still carries
the circuit's answer. We now encode its equations in a simplicial
boundary matrix, preparing the reduction to Hodge Laplacian linear systems in
\Cref{sec:main-hodge}. DKGZ provide the
geometric starting point: triangulated surfaces represent variables,
and their boundary cycles are joined to express difference and averaging
equations~\cite{DingFull}. Their classical reduction controls sparsity,
condition numbers, and errors in satisfying the original equations. For our QLSP
reduction, we must also recover the prescribed normalized minimum-norm
state by an efficient quantum operation.

This changes what the construction must establish. Repeating a variable
on several triangles changes its contribution to the solution norm,
which can change the minimum-norm solution. The number of nonzero matrix entries
can be exponential in the input length, so constructing the entire
complex is too costly. For the same reason, the dependence on this
number in DKGZ's published condition-number bound does not give the
polynomial bound we require~\cite[full version, Lemma~5.2]{DingFull}.
We address these issues by balancing the copies of each surface,
bounding their size through the occurrence bound, and assigning indices
that support local queries and quantum recovery.

\subsection{Preserving the normalized minimum-norm solution}
Write the consistent difference--average linear system as $Ax=b$. Following
the surface construction, each signed occurrence of a variable is
assigned a triangular boundary cycle. A coefficient $-2$, for example,
uses two cycles attached with reversed orientation. Cycles assigned to the same
equation are identified, and the right-hand side $\gamma$ equals $b_r$
on each edge of the oriented cycle for row $r$, and zero elsewhere.
\Cref{def:bc-construction-short,lem:bc-simplicial-legality} specify
these identifications and show that they produce an abstract simplicial
complex with boundary matrix $\partial_2$.

The equations at interior edges force the triangle coefficients on
each connected surface to be equal. With just one surface of $t_i$
triangles for variable $x_i$, those coefficients would contribute
$t_i x_i^2$ to the squared norm. Minimizing this weighted sum need not
give the minimum-norm solution of $Ax=b$. We therefore use $t_i$ copies
of the $t_i$-triangle surface and let the sum of their coefficients
represent $x_i$. Here $t_i$ is determined by the coefficient
multiplicity $u_i=\sum_r|A_{ri}|$ and is bounded by a constant because
each variable occurs only a bounded number of times in the source equations.

More precisely, let $a_{i,s}$ be the common triangle coefficient on
copy $s$. \Cref{lem:bc-kernel-short} shows that the equations on the cycles assigned to constraints are exactly $Ax=b$ with
$x_i=\sum_{s=1}^{t_i}a_{i,s}$. If $f_i$ contains all triangle
coefficients representing $x_i$, then
\begin{equation}
 \norm{f_i}^2=t_i\sum_{s=1}^{t_i}a_{i,s}^2
 \ \ge\ \left(\sum_{s=1}^{t_i}a_{i,s}\right)^2=x_i^2.
 \label{eq:equal-copy}
\end{equation}
For a fixed $x_i$, the minimum occurs when every copy has coefficient
$x_i/t_i$. This assignment puts $x_i/t_i$ on $t_i^2$ triangles and
preserves the norm. It defines an isometry $Q$, and minimizing over
all feasible $x$ gives
\begin{equation}
 (\partial_2)^+\gamma=QA^+b,\qquad Q^\top Q=I,
 \label{eq:main-isometry}
\end{equation}
by \Cref{thm:bc-mp-main}. The vectors $b$ and $\gamma$ here are the
right-hand sides before normalization for state preparation. Thus the
normalized solution of the boundary linear system is exactly $Q$ applied to the normalized
minimum-norm solution of the difference--average linear system.

This identity also specifies quantum recovery. Each variable occupies
a constant-size block of triangle coordinates, so a fixed finite
family of unitaries implements $Q^\top$ with a success flag. Success
is certain on the exact encoded state. For an approximate state,
\Cref{lem:bc-decoder-stability} bounds the error in the recovered state and the
probability that recovery succeeds. We therefore recover amplitudes coherently,
without reading the triangle coefficients as classical data.

\subsection{Conditioning and access without constructing the complex}
Preserving the solution does not by itself control the smallest
nonzero singular value. For this, the proof separates variations
within each surface from its average coefficient. Equations indexed by interior edges
control the variations; \Cref{lem:bc-port-control} bounds
their effect on the constraint cycles. On vectors orthogonal to the
kernel, the remaining averages are controlled by the nonzero singular
values of $A$. Combining these bounds
gives the spectral estimate in \Cref{lem:bc-spectral-main}, with
constants depending only on the fixed occurrence bound, rather than
the total number of variables.

The same occurrence bound limits the number of triangles incident to
each edge. Together with the three edges of every triangle, this gives
constant row and column sparsity
(\Cref{lem:bc-port-sparsity}). The rescaling factor is consequently
constant, and \Cref{thm:bc-da-to-b2-wrapper-explicit} gives
nonzero singular values of $B_2/\alpha(B_2)$ in $[1/\kappa_B,1]$,
where $B_2=\partial_2$ and $\kappa_B\le\poly(\nu)$ on the constructed
instances. No additional factor depending on the number of nonzero
entries of $A$ is needed.

We must also compute simplex indices without listing the complex.
If variable blocks had different lengths, locating a block could
require summing the sizes of all preceding blocks. We give every block
the same size by adding disjoint triangles and isolated edges, with
zero right-hand side on these additions.
\Cref{lem:bc-fixed-capacity-indexing} shows that this preserves the
minimum-norm solution and permits direct reversible index calculation.
Each query to the boundary matrix uses only a constant number of
queries listing a variable's occurrences in the source equations. Conversely, \Cref{lem:bc-template-identification} identifies
the surface for a variable using only the oracle for the boundary matrix, allowing
recovery to select the unitary for $Q^\top$. The reduction therefore
supports both construction and recovery on superpositions of indices.

The resulting boundary linear system meets the requirements on solution
accuracy, spectra, and matrix access. In \Cref{sec:main-hodge}, we use Hodge
decomposition and suspension to obtain the corresponding Hodge
Laplacian linear systems in every stated fixed dimension.

\section{Hodge decomposition and suspension}\label{sec:main-hodge}
The preceding construction gives a boundary linear system $Bf=\gamma$ whose
normalized minimum-norm solution retains the circuit's answer. We now
use it to obtain upper, lower, and full Hodge Laplacian linear systems in the
dimensions of \Cref{thm:main_hardness}. The Hodge decomposition explains
how the upper and lower Hodge Laplacians act on different parts of a vector.
Pseudoinverse identities then relate the solutions of these linear systems, and
suspension raises the dimension while preserving the encoded solution.

\subsection{The orthogonal Hodge decomposition}
Let $C_k(K)$ be the vector space with the oriented $k$-simplices of $K$
as an orthonormal basis. Its vectors assign coefficients to these
simplices. Applying two consecutive boundary matrices gives zero, so
$B_kB_{k+1}=0$. Consequently, the spaces $\im(B_k^\top)$ and
$\im(B_{k+1})$ are orthogonal. The first consists of vectors obtained
from $(k-1)$-simplex coefficients by applying $B_k^\top$; the second
consists of vectors obtained by applying $B_{k+1}$ to coefficients on
$(k+1)$-simplices.

The remaining component is the harmonic space $\ker\Delta_k$. Indeed,
\[
 \langle x,\Delta_kx\rangle
 =\norm{B_kx}^2+\norm{B_{k+1}^\top x}^2,
\]
so a vector is harmonic exactly when both terms vanish. These
observations give the orthogonal decomposition
\[
 \begin{aligned}
 C_k(K)&=\im(B_k^\top)\oplus\ker\Delta_k\oplus\im(B_{k+1}),\\
 \ker\Delta_k&=\ker B_k\cap\ker B_{k+1}^\top,
 \end{aligned}
\]
stated in \Cref{lem:standard-hodge-decomposition}. Every vector has a
unique decomposition into these three components.

The lower Hodge Laplacian acts only on the first component, and the
upper Hodge Laplacian only on the third; both annihilate the harmonic
component. In particular,
$L_k^\downarrow L_k^\uparrow=L_k^\uparrow L_k^\downarrow=0$.
The full Hodge Laplacian therefore acts independently on the two
nonharmonic components, and its positive spectrum is the union of
their positive spectra. Its pseudoinverse inverts these nonzero
eigenvalues and sets the harmonic component to zero. This is the
minimum-norm choice: adding a nonzero kernel component leaves the
image under $\Delta_k$ unchanged and increases the norm. We use this decomposition to
understand the solution spaces, while implementing the reductions
through sparse access to the boundary matrices.

\subsection{From boundary linear systems to Hodge Laplacian linear systems}
For the initial boundary linear system, $B=\partial_2$ and
$\gamma\in\im B$. The matrices $BB^\top$ and $B^\top B$ are the upper
Hodge Laplacian on edges and the lower Hodge Laplacian on triangles.
Their solution vectors are related by
\begin{align}
 (B^\top B)^+B^\top\gamma&=B^+\gamma,
 \label{eq:boundary-lower}\\
 B^\top(BB^\top)^+\gamma&=B^+\gamma.
 \label{eq:boundary-upper}
\end{align}
These are the standard identities in
\Cref{lem:standard-mp-identities}. On a nonzero singular direction
with singular value $\sigma$, the Hodge Laplacian pseudoinverse contributes
$1/\sigma^2$ and $B^\top$ contributes $\sigma$, leaving the factor
$1/\sigma$ of $B^+$. Zero singular directions contribute zero.
The identities thus relate the prescribed minimum-norm solutions,
including when the matrices have kernels.

For the lower Hodge Laplacian, Eq.~\eqref{eq:boundary-lower} specifies
the reference right-hand side
$\eta_{\rm ref}=B^\top\gamma/\norm{B^\top\gamma}$. Its normalized
solution is already the desired solution state of the boundary linear system. On the
constructed instances, $\gamma$ has constant support and each row of
$B$ has constantly many nonzero entries, so $\eta_{\rm ref}$ also has
constant support. \Cref{lem:lower-rhs-preparation} provides an efficient
preparation circuit with all workspace returned to zero. The state
actually prepared defines the input; the approximation bounds in
\Cref{lem:rhs-robustness,lem:state-to-verifier-final} show that
inverse-polynomial preparation accuracy controls both the error in the solution state
and the change in acceptance probability.

For the upper Hodge Laplacian, the right-hand side remains $\gamma$.
Equation~\eqref{eq:boundary-upper} recovers the solution of the boundary linear system by
applying $B^\top$ to the solution of the upper Hodge Laplacian linear
system and normalizing. This is a
quantum operation with a success flag. The ideal successful branch
applies $B^\top/\alpha(B)$. If the nonzero singular values
of $B/\alpha(B)$ lie in $[1/\kappa_B,1]$, its ideal success probability
is at least $1/\kappa_B^2$. \Cref{lem:bc-b2-to-quphls1-explicit} proves
this bound and controls the error after recovery. The two Hodge Laplacians,
rescaled by $\alpha(B)^2$, have nonzero eigenvalues in
$[1/\kappa_B^2,1]$. Thus polynomial spectral bounds for the boundary
linear system give polynomial bounds for both the Hodge Laplacian solver and the
recovery operation.

\subsection{Raising the dimension by suspension}
The preceding reductions start on edges for the upper Hodge
Laplacian and on triangles for the lower Hodge Laplacian. We raise these dimensions
by simplicial suspension. Add two vertices $a,b$ and join each one
separately to every original simplex; no simplex contains both new
vertices. For an $i$-simplex $\sigma$, the two resulting
$(i+1)$-simplices are denoted $a*\sigma$ and $b*\sigma$.
With the orientations of \Cref{def:suspension-explicit}, encode its
basis state as
\[
 \widehat J_i\ket\sigma
 =\frac{\ket{a*\sigma}-\ket{b*\sigma}}{\sqrt2}.
\]
The two simplices share the original face $\sigma$. Their opposite
coefficients cancel this face when the boundary matrix is applied, while the
remaining faces carry two copies of the original simplex's boundary with
opposite signs. This cancellation is expressed by the block matrix
in \Cref{lem:suspension-block-explicit}.

For the upper Hodge Laplacian, the same cancellation makes the
encoded subspace invariant: the upper Hodge Laplacian of the suspended complex acts there as the
original upper Hodge Laplacian. Writing $L_{i+1}^{\uparrow,\Sigma}$ for the upper
Hodge Laplacian of the suspended complex,
\Cref{lem:dim-raise-explicit} gives
\[
 (L_{i+1}^{\uparrow,\Sigma})^+\widehat J_i\gamma
 =\widehat J_i(L_i^\uparrow)^+\gamma.
\]
Since $\widehat J_i$ preserves the norm, this identity also preserves
the normalized solution. A circuit acting on the register that records
the choice of $a$ or $b$ reverses the encoding; recovery succeeds with
probability one on the exact encoded state. The identity uses the
unweighted operators. Their rescaling factors may differ, but these
scalar factors disappear when the solution vectors are normalized.

For the lower Hodge Laplacian, we suspend the boundary linear system
before forming the lower Hodge Laplacian. The new right-hand side has zero
coefficients on the original simplices and equal and opposite copies
of $\gamma/\sqrt2$ on the two cones.
\Cref{lem:explicit-boundary-desuspension} shows that the new
minimum-norm solution of the boundary linear system is the corresponding pair of opposite
copies of the old solution, up to a global sign. Reversing the cone
encoding recovers the old state. Applying
Eq.~\eqref{eq:boundary-lower} to the boundary matrix of the suspended complex then gives
the reduction to lower Hodge Laplacian linear systems in the next dimension.

Preserving the encoded solution is only part of the argument: the
new matrix also acts on vectors outside that subspace. For the
unweighted upper Hodge Laplacians, \Cref{lem:suspension-spectrum-explicit}
shows that each old positive eigenvalue $\lambda$ is retained and
has a companion $\lambda+2$; any other positive eigenvalues equal
$2$. The corresponding lower Hodge Laplacians have the same nonzero
spectra by \Cref{lem:standard-mp-identities}. Together with the sparse-access construction in
\Cref{lem:suspension-locality-explicit}, this preserves the required
spectral bounds after rescaling. On the constructed instances, a
fixed number of suspensions keeps sparsity and rescaling factors
constant and all inverse spectral bounds polynomial. Hence $k-1$
suspensions give the result for upper Hodge Laplacian linear systems
in every fixed dimension $k\ge1$, and $k-2$ give the result for lower
Hodge Laplacian linear systems for every fixed $k\ge2$
(\Cref{cor:dim-raise-explicit,cor:lower-full-hodge-hardness}).

\subsection{Full Hodge Laplacians and recovery of the circuit output}
The Hodge decomposition shows how the full Hodge Laplacian combines
the lower and upper Hodge Laplacians. In the construction for lower
Hodge Laplacian linear systems, the complex has
dimension $k$, so there are no $(k+1)$-simplices. Consequently
$B_{k+1}=0$ and $\Delta_k=L_k^\downarrow$: the lower and full Hodge Laplacian linear systems
have identical solutions. This gives the result for full Hodge Laplacian linear systems
for every fixed $k\ge2$ in \Cref{cor:lower-full-hodge-hardness}.
If both summands must be nonzero, adjoin a disjoint $(k+1)$-simplex
and give that component zero right-hand-side amplitude.
\Cref{cor:full-both-nonzero} verifies that its positive spectrum
preserves the required bound and that the reference solution extends
by zero to the added component.

These identities complete the reduction from circuit acceptance to
solution states of Hodge Laplacian linear systems. Recovery reverses the suspension
encodings, applies the transpose of the boundary matrix when needed for the
upper Hodge Laplacian linear system, and then reverses the surface and auxiliary-variable
constructions. \Cref{lem:reverse-typed-oracle-simulation} shows how
every query to an earlier boundary matrix is answered using only the final
instance's oracles. The recovered propagation state can then be
measured as in \Cref{sec:main-boundary}. The next section bounds the
overall success probability and approximation error, and gives the
quantum algorithms for the stated input models.

\section{Decision hardness and quantum algorithms}\label{sec:main-decision}
The reductions above explain how the solution state of a Hodge
Laplacian linear system contains the answer to a quantum computation. We now show that the
prescribed recovery and measurement retain a detectable acceptance
gap. We then give a quantum algorithm for every input satisfying the
promises, and combine these two directions to obtain completeness.

\subsection{The decision problem and recovery of the circuit output}
For a Hodge Laplacian $M$, the decision problem concerns the state
$\ket{z_M}=M^+\ket\eta/\norm{M^+\ket\eta}$ from
\Cref{sec:main-input}. The supplied preparation oracle determines
$\ket\eta$. A fixed quantum procedure recovers the propagation state
on the constructed instances and measures its clock and output qubit.
The input label selects the Hodge Laplacian type and the permitted indices;
it supplies no separate source circuit or arbitrary verification
circuit. The measurement follows recovery, whose cost is included.

The procedure returns a fair random bit whenever recovery fails or
the clock falls outside the final interval. Its acceptance probability
$p_\star$ includes these outcomes. The promise distinguishes
$p_\star\ge1/2+\Delta$ from $p_\star\le1/2-\Delta$, where
$\Delta^{-1}$ is polynomially bounded. This is the decision problem of
\Cref{def:measure-wrapper}.

For the reference solution of an upper Hodge Laplacian linear system,
recovery first reverses the suspensions and then applies the transpose
of the boundary matrix, recording whether this step succeeds, as in
\Cref{sec:main-hodge}. For a lower Hodge Laplacian linear system,
reversing the suspensions already gives the reference solution of the
original boundary linear system. The construction for full Hodge
Laplacian linear systems uses the same recovery on the component that
encodes the computation.
In each case, $Q^\top$ recovers the difference--average solution, and
$D_{\rm aux}$ recovers the propagation solution by removing the
auxiliary offsets. These maps act coherently on amplitudes and record
success in flag qubits. The last step measures the clock and the
designated output qubit. \Cref{lem:reverse-typed-oracle-simulation}
shows that every query needed along this sequence can be answered
using the final instance's oracles with polynomial cost.

\subsection{Why the decision problem is hard}
First consider the exact reference solution and ideal recovery
operations. Write $p_\star^{\rm id}$ for their acceptance probability,
$p_{\rm succ}$ for the probability that all recovery operations
succeed, and $p_{\rm tail}$ for the probability of the final clock
interval in the recovered history state. On that interval, the
output qubit has the source circuit's acceptance probability
$p_{\rm acc}(C)$. All other outcomes contribute acceptance probability
$1/2$, so
\begin{equation}
 p_\star^{\rm id}-\frac12
 =p_{\rm succ}p_{\rm tail}
   \left(p_{\rm acc}(C)-\frac12\right).
 \label{eq:bias}
\end{equation}
Thus recovery preserves the sign of the source circuit's acceptance
bias. The clock construction gives $p_{\rm tail}\ge1/8$
(\Cref{lem:weighted-history}). Reversing suspension and applying
$Q^\top$ succeed with probability one on their exact images, and
subtracting the auxiliary offsets succeeds with constant probability. Applying the transpose of the boundary
matrix for upper Hodge Laplacian linear systems succeeds with inverse-polynomial
probability. \Cref{lem:herald-bounds,lem:lower-full-recovery-gap}
therefore give $p_{\rm succ}\ge1/\poly(\nu)$. Since the source
circuit has a constant acceptance gap, Eq.~\eqref{eq:bias} gives an
inverse-polynomial gap for the ideal experiment.

The implemented experiment must retain this gap. For the constructions of lower and
full Hodge Laplacian linear systems, \Cref{lem:lower-rhs-preparation,lem:rhs-robustness}
control the change from the reference right-hand side to the state
actually prepared. \Cref{lem:composed-heralded-accounting} also bounds
the effects of approximate recovery operations on the complete
experiment, including failed runs. Choosing these errors small enough
bounds the difference between $p_\star$ and $p_\star^{\rm id}$ by half
the proved lower bound on the ideal gap, establishing the promise for
the actual target.
The same lemma shows that preparing the solution to the prescribed
inverse-polynomial accuracy preserves a detectable gap. Both bounds
include failure outcomes.

The classical reduction constructs the instance and its oracle
circuits in polynomial time. It therefore reduces quantum circuit
acceptance to the prescribed Hodge Laplacian decision problem.
\Cref{thm:hardness-final} proves this for upper Hodge Laplacians in
every fixed dimension $k\ge1$;
\Cref{cor:lower-full-hodge-hardness} proves it for lower and full
Hodge Laplacians in every fixed dimension $k\ge2$. Consequently,
preparing the solution state is $\BQP$-hard: repeated preparation to
the specified accuracy, followed by the efficient recovery and
measurement above, decides the source computation with bounded error.
This uses the same reduction as the decision-hardness result.

\subsection{Preparing the solution with block encoding and QSVT}
For the quantum upper bound, we must solve every promised input,
including inputs outside the constructed hard families. We use sparse
access to the boundary matrices to apply a matrix transformation without listing the
matrix entries.

A block encoding implements a rescaled matrix as part of a unitary.
With its auxiliary qubits initially zero, the branch in which they
return to zero applies $A/\alphabe(A)$ to the input vector, up to a
controlled error. We now identify this normalization with the factor
introduced in \Cref{sec:main-input}:
\[
 \alpha(A)=\alphabe(A).
\]
They denote the same factor. For an unweighted boundary matrix $B$,
the sparse-access construction fixes
$\alphabe(B)=\sqrt{s_{\rm row}s_{\rm col}}$ from the supplied sparsity
bounds. \Cref{lem:sparse-access-block-encoding} implements it with
polynomially many gates and $O(s_{\rm row}+s_{\rm col})$ sparse queries
under our oracle convention. The equations still use the unweighted
matrix; its rescaling does not change the normalized solution.

Composing a block encoding of a boundary matrix with its adjoint,
using separate auxiliary registers, gives the upper or lower Hodge
Laplacian, depending on the order of multiplication. Their
normalizations are $\alphabe(B_{k+1})^2$ and $\alphabe(B_k)^2$,
respectively. For the full Hodge Laplacian, a control qubit combines the two
Gram encodings with weights proportional to these squared factors.
The normalization is their sum, omitting any absent boundary matrix.
\Cref{def:be-normalization} records these constructions and factors.
They give a block encoding of $M$ with normalization $\alphabe(M)$
for each of the three Hodge Laplacian types.

Quantum singular value transformation (QSVT) uses this unitary and its
inverse to apply a polynomial to the normalized eigenvalues of the
positive-semidefinite matrix $M$~\cite{GSLW2019}. The spectral promise
places every nonzero normalized eigenvalue in $[1/\kappa,1]$. We choose an odd
polynomial that approximates
\[
 \lambda\longmapsto\frac{1}{2\kappa\lambda}
 \qquad (1/\kappa\le\lambda\le1)
\]
and remains bounded by one on $[-1,1]$, as required for the quantum
implementation. Its value at zero is zero. It therefore removes the
kernel component and approximately applies
$\alphabe(M)M^+/(2\kappa)$. Applied to $\ket\eta$, its successful
branch approximates a vector proportional to the desired minimum-norm
solution.

The range-overlap promise prevents this branch from being too small:
the exact reciprocal applied to $\ket\eta$ gives a vector of norm at least
$1/(2\kappa\zeta_\eta)$. Choosing the polynomial and implementation
errors sufficiently small preserves an inverse-polynomial success
probability. Fixed-point amplitude amplification increases this
probability using the supplied overlap bound, the preparation oracle,
and its inverse.
\Cref{thm:qsvt-final} gives the polynomial approximation, the required
implementation precision, and the amplified preparation cost. Under
the promises of \Cref{sec:main-input}, these costs are polynomial.
\Cref{thm:containment-final,thm:lower-full-containment} apply the
algorithm to the upper, lower, and full Hodge Laplacians, respectively.

\subsection{From the quantum algorithm to completeness}
To decide the promised alternatives, prepare the solution to accuracy
$O(\Delta)$ and run the prescribed recovery and measurement.
\Cref{lem:state-to-verifier-final} bounds the change in acceptance
probability by the error in the prepared solution state. Each preparation attempt has
a prescribed polynomial cost; its small failure probability is also
included by returning a fair random bit on failure. Choosing
the state and implementation errors to change acceptance by at most
$\Delta/2$ leaves a detectable gap. Then $O(\Delta^{-2})$ independent
samples distinguish the two cases with constant bounded error.
All preparation and recovery queries are counted, giving polynomial
query complexity and polynomially many additional gates.

For circuit-encoded inputs, executing the supplied oracle circuits
turns this algorithm into a uniform quantum computation of size
polynomial in the input length
(\Cref{lem:circuit-input-simulation}). Hence the decision problems
belong to $\BQP$. Conversely,
\Cref{lem:finite-oracle-encoding-size} ensures that the hardness
constructions output these finite circuit descriptions in classical
polynomial time. Combining membership with hardness makes the
circuit-encoded decision problems $\BQP$-complete under classical
polynomial-time many-one reductions, in the dimension ranges above
(\Cref{cor:complete-final,cor:lower-full-hodge-hardness}). The
sparse-oracle versions retain the corresponding hardness and
polynomial-query quantum algorithms.

\section{Discussion}\label{sec:main-discussion}
We have shown that the prescribed decision problems for Hodge
Laplacian linear systems are $\BQP$-complete for circuit-encoded
inputs under the stated promises. This holds already on edges for
upper Hodge Laplacians and on triangles for lower and full Hodge
Laplacians, and extends to every higher fixed dimension in the
respective ranges of \Cref{thm:main_hardness}. Preparing the normalized
minimum-norm solution states is therefore $\BQP$-hard: repeated
preparation, recovery, and measurement solve the decision problem
with bounded error. Thus unweighted simplicial structure and fixed
interaction order do not remove this hardness.

The proof brings together circuit encoding, simplicial constructions,
and quantum linear algebra. A Dirichlet propagation linear system first
stores a quantum computation in its solution. Bounded-occurrence
difference--average equations then allow us to adapt the triangulated
surfaces of DKGZ~\cite{DingFull}. Balancing the surface copies preserves
the minimum-norm solution under a norm-preserving encoding, while the
occurrence bound controls sparsity and conditioning. Pseudoinverse
identities relate the solution of the boundary linear system to the
solutions of the upper and lower Hodge Laplacian linear systems, and suspension gives the required dimensions. Throughout, we provide local access and efficient
quantum recovery, with controlled error and success probability.
Block encoding and QSVT supply the quantum algorithms for all promised
inputs. The resulting reductions show what must be preserved when
using a classical linear-system construction to prove hardness of a
decision problem defined by measurements of a recovered quantum state.

The result connects QLSP complexity to the analysis of large
simplicial networks accessed through local queries. A randomized classical polynomial-time algorithm for
the prescribed circuit-encoded decision problems would decide every
problem in $\BQP$. This consequence concerns the specified input
access and recovered decision output. It goes beyond comparisons with
particular classical solvers, while leaving room for efficient
classical methods on more restricted network families or for other
outputs. Assessing an application therefore requires matching the
input access, accuracy, and information returned by the classical and
quantum algorithms.

Higher-dimensional HodgeRank provides a natural direction for this
analysis. Its least-squares formulation fits data on simplices by
scores on their faces, using an upper Hodge Laplacian in the normal
equations~\cite{JiangHodgeRank,LedittoHodgeRank}. This connects the
ranking problem to the operators studied here. A hardness reduction
for that task must also construct the comparison data efficiently:
the right-hand side is obtained by applying a boundary matrix to
those data, rather than being an independently supplied vector.
It must also preserve the input representation and state-preparation
costs. Where the application requires a clique complex, whose simplices
are the fully connected vertex sets of a graph, the reduction must
produce one. A further step is to encode the answer in a ranking
quantity of interest, such as a relative score or a consistency
measure. These questions would connect hardness of solution-state
preparation to useful ranking outputs. The present dimension ranges
do not establish hardness for ordinary graph-based ranking.

Coupled oscillators offer a related direction through the simplicial
Kuramoto model. Pseudoinverses of boundary matrices and Hodge Laplacians applied to
the natural frequencies give bounds below which the projected phase
variables cannot become stationary. These bounds give sufficient
conditions for no phase locking~\cite{Nurisso2024,LedittoOscillators}. These expressions provide a
concrete starting point for adapting our constructions to dynamical
questions. Such an adaptation would need to realize the encoded
right-hand side through the model's frequency data, preserve the
required spectral bounds, and retain a detectable difference in the
chosen diagnostic. For example, a no-phase-locking certificate can
depend on the norm of a pseudoinverse vector, whereas our decision
procedure uses information in its normalized state after recovery.
Proving hardness of that norm threshold requires an additional
argument. Establishing this link would give a complexity explanation
for a specific dynamical diagnostic; it would not by itself classify
the cost of simulating the full nonlinear evolution.

These applications also motivate studying which further restrictions
preserve hardness. The construction for full Hodge Laplacian linear systems can make both
summands nonzero by adding a disjoint simplex, so obtaining comparable
guarantees for connected complexes is a separate question
(\Cref{rem:lower-full-scope}). More broadly, identifying natural
ranking and dynamical observables that retain an inverse-polynomial
decision gap would connect the present worst-case classification to
specific scientific questions. The reductions developed here provide
a starting point for that work, together with explicit bounds on the
matrix-access costs, spectral parameters, and quantum recovery costs
that such a connection must preserve.

\section{AI Statement}\label{sec:ai-statement}
The results of this paper were developed through iterative interaction between the author and successive ChatGPT models. These models were used for brainstorming, proof development, drafting, and verification. The author proposed the original idea and carried out the early proofs
development for the Dirichlet propagation linear system and DKGZ gadget adaptation. GPT-5.5 was used to expand this idea and develop the subsequent arguments at a stage where the author had been unable to complete the adaptation. It then helped identify and formulate the requirements for correct reductions and dimension raising. Most parts of the current draft were generated using GPT-6 Astra for technical presentation only. Its writing contains language characteristic of large language models and does not yet meet the standard of a publication-ready manuscript. Further revision is necessary to improve the clarity, precision, and quality of the writing. The author has independently checked and verified the results and takes full responsibility for their correctness.

\clearpage
\appendix
\section{Models, reductions, and solution states}\label{sec:model}

Our reductions describe large matrices through local queries. They
recover the answer to a quantum computation from an approximate solution
state. This section specifies the matrix and right-hand-side inputs,
the cost of using them, and the error allowed in a prepared state.
Section~\ref{sec:promise-problems} defines the tasks and their decision
versions. Section~\ref{app:audit-addendum} gives the identities used
to relate their solutions.

\paragraph{Indices, norms, and error between states.}
Sparse-oracle addresses use $[N]=\{0,\ldots,N-1\}$.
Matrix labels $1,\ldots,N$ correspond to addresses $0,\ldots,N-1$.
Thus $A_{\rm addr}(i,j)=A_{i+1,j+1}$.
Clock and simplex labels use the ranking maps given in each construction.
For an $m\times n$ matrix, the valid index sets are $[m],[n]$.
The registers that store these indices may have larger dimensions,
rounded up to powers of two. For an index outside the valid ranges,
a sparse oracle returns zero and a fixed dummy index outside the valid
output range. Output index registers include space for this dummy index.
Fixed-point values bounded by $V$ with entry error
$0<\delta_{\rm val}\le1$ use
$O(\log(2+V)+\log(1/\delta_{\rm val}))$ bits.
For vectors, $\norm{\cdot}$ is the $\ell_2$-norm; for matrices, it is the
spectral norm. Matrix $1$- and $\infty$-norms are the maximum absolute
column and row sums, and $\normmax A=\max_{i,j}|A_{ij}|$.
For nonzero $M$, the condition number is
$\cond(M)=\sigma_{\max}(M)/\sigma_{\min}^+(M)$, where
$\sigma_{\max}(M)$ and $\sigma_{\min}^+(M)$ are its largest and
smallest positive singular values, respectively. For a positive-semidefinite matrix $H$, $\lambda_{\min}^+(H)$ is its
smallest positive eigenvalue. We write $\nnz(A)$ for the number of
nonzero entries of $A$, $e_i$ for a coordinate basis vector, and
$\mathbf1$ for an all-ones vector of the indicated dimension.
For boundary linear systems and Hodge Laplacian linear systems, the
supplied parameter $\cond$ bounds the reciprocal of the smallest
positive singular value or eigenvalue after division by the rescaling
factor $\alpha(A)\ge\norm A$. It need not equal the condition number.
For normalized pure states we use
\[
\dist_{\rm ph}(\ket{\widetilde\psi},\ket\psi)
=\min_{\theta\in\mathbb R}
\norm{\ket{\widetilde\psi}-e^{i\theta}\ket\psi}.
\]
All bounds on the error between normalized states use this
phase-invariant error: we choose the common phase of all amplitudes
that makes the two states closest.

\paragraph{Complexity conventions.}
A promise problem specifies disjoint sets of yes and no inputs, with
no requirement outside their union. Throughout, $\BQP$ denotes the class
of promise problems decided by uniform quantum circuits of polynomial
size, with error at most $1/3$ on every promised
input~\cite[Section~IV.1]{WatrousSurvey}.
Here uniformity requires a deterministic classical algorithm to
generate the circuit for inputs of length $n$ in time
$\poly(n)$. Repeating a computation can turn an inverse-polynomial acceptance gap
into constant bounded error at polynomial cost~\cite[Proposition~3]{WatrousSurvey}.
Hardness and completeness of finite-input promise problems refer to
classical polynomial-time many-one reductions; completeness also
requires membership in the class. Below we explain how hardness of the
prescribed decision problem
establishes hardness of preparing the solution state.

\paragraph{Chain spaces and boundary matrices.}
Choose one orientation for each simplex of an unweighted simplicial
complex $K$. The real chain space $C_k(K)$ has these oriented
$k$-simplices as an orthonormal basis; reversing an orientation changes
the sign of its basis vector. For $k\ge1$, the boundary matrix
$\partial_k:C_k(K)\to C_{k-1}(K)$ sends an oriented simplex
$[v_0,\ldots,v_k]$ to
$\sum_{r=0}^k(-1)^r[v_0,\ldots,\widehat v_r,\ldots,v_k]$,
where the hat denotes omission and each face is expressed in its
chosen orientation. These incidence signs give
$\partial_k\partial_{k+1}=0$.
We write $B_k:=\partial_k$ for the unweighted boundary matrix.
If $K$ has no $j$-simplices, set $C_j(K)=\{0\}$ and treat each
boundary matrix entering or leaving this space and the corresponding
Gram-matrix term as zero.

\begin{remark}[Gate set and exact preparation]\label{rem:implementation-gate-set}
For each fixed simplicial dimension, finite circuit encodings use a
fixed finite universal gate set. We add reversible Boolean gates and
the algebraic unitaries of constant dimension needed by the reductions.
These unitaries prepare uniform superpositions of three
computational-basis states, antisymmetric superpositions of two basis
states, normalized constant vectors on fixed-size coordinate blocks,
and the finitely many right-hand-side states with prescribed amplitudes
on fixed supports. We include their inverses and controlled versions.
Once the construction constants are fixed, the list of added gates is
finite. Each added gate has efficiently computable algebraic entries
and does not depend on the length of the source circuit.

A circuit is \emph{clean} if every workspace qubit returns to $\ket0$
while its designated output registers retain the result. Exact
preparation uses this fixed gate convention. Definition~\ref{def:rational-gate-set}
specifies the source gate set separately. Another fixed universal gate
set with efficient gate approximation, such as Clifford+$T$, can
implement the circuits with the simulation error described after
Lemma~\ref{lem:circuit-input-simulation}. The exact prepared right-hand
side remains the one defined by the original circuit.
\end{remark}

\begin{definition}[Sparse access]\label{def:sparse-access}
For $A\in\mathbb R^{m\times n}$, let $s_{\rm row},s_{\rm col}$
be known positive integer bounds on the number of nonzero entries in
any row and column.
The following reversible oracles use bitwise exclusive-OR, denoted
$\oplus$, to write the index and value into output registers:
\[
\begin{aligned}
\mathcal O_A^{\rm row}:\ket{i,\ell,j,z}&\mapsto
\ket{i,\ell,j\oplus j(i,\ell),z\oplus\widetilde A_{i,j(i,\ell)}},
&& i\in[m],\ \ell\in[s_{\rm row}],\\
\mathcal O_A^{\rm col}:\ket{j,\ell,i,z}&\mapsto
\ket{j,\ell,i\oplus i(j,\ell),z\oplus\widetilde A_{i(j,\ell),j}},
&& j\in[n],\ \ell\in[s_{\rm col}].
\end{aligned}
\]
These transformations act coherently on superpositions and return
all workspace qubits to $\ket0$. Each nonzero entry appears exactly
once in its row list and its column list. The two lists agree on the index and fixed-point value. Unused
positions in either list are called dummy slots; they return the dummy
index and zero. At active entries, the returned values satisfy
$|\widetilde A_{uv}-A_{uv}|\le\delta_{\rm val}$. The precision
$\delta_{\rm val}$ must meet the error bounds required by
Lemma~\ref{lem:sparse-access-block-encoding} and Theorem~\ref{thm:qsvt-final}.

The numerical and register conventions above apply. Right-hand-side
access is specified in Definition~\ref{def:rhs-preparation-error}.
\end{definition}

\begin{definition}[Uniform sparse-oracle inputs]\label{def:sparse-access-relativized-model}
An input is $(1^\nu,\mathsf{Enc}_{\rm SA}(x),\mathcal O_x)$, where
$\nu$ is a unary size parameter, $\mathsf{Enc}_{\rm SA}(x)$ is the
finite description of length $\poly(\nu)$ specified below, and
$\mathcal O_x$ contains the
row, column, scalar, and preparation oracles specified by $\mathsf{Enc}_{\rm SA}(x)$.
Sparsities, address widths, numerical precision and magnitudes,
inverse spectral bounds, condition bounds, and the rescaling factors
are polynomially bounded in $\nu$.
An efficient oracle algorithm uses a polynomial-time uniform circuit
family with polynomially many queries and additional elementary gates.
Its prescribed accuracy, decision gap, and lower bound on the probability
that recovery succeeds all have polynomially bounded reciprocals.

Reductions and query accounting follow
Definition~\ref{def:state-red} and its accompanying hardness convention.
For oracle inputs, we prove hardness and bound the number of queries;
$\BQP$-completeness under classical many-one reductions is
asserted only for the finite circuit inputs of
Definition~\ref{def:circuit-encoded-input}, whose computational costs
are measured in their full bit length $|x|$.
\end{definition}

For hardness with finite inputs, the reduction must produce oracle
circuits with descriptions of polynomial length. The next lemma proves
this for our constructions. Lemma~\ref{lem:circuit-input-simulation}
shows that the supplied circuits can be executed with polynomial cost.
These two facts are used for hardness and membership, respectively,
in Corollary~\ref{cor:complete-final}.

\begin{lemma}[Finite oracle encodings]\label{lem:finite-oracle-encoding-size}
Let $C$ have $L$ gates of arity at most three and one designated
output qubit. Delete untouched non-output wires and relabel; let
$n_C$ be the bit length of the resulting gate-list description.
Every row, column, scalar, and preparation oracle constructed from a
source circuit after removing unused wires has a clean circuit of size
$\poly(n_C)$. Classical reversible circuits supply the matrix entries.
Unitary circuits prepare the right-hand-side states. All remaining
input fields have total length $\poly(n_C)$. Set $\nu=n_C$.
The unary field and the designated oracle circuits then give
\[
n_C\le|x_C^{\rm circ}|\le\poly(n_C).
\]
Thus every $\poly(L)$ parameter is polynomial in both $\nu$ and the
finite input length. These encodings let us apply the oracle reductions
and algorithms below to finite inputs.
\end{lemma}
\noindent\emph{Implementation check.}
Removing untouched non-output wires gives $n_q\le3L+1$ and
$L\le n_C\le\poly(L)$. The specified matrix queries combine
lookups in the source gate list with polynomial-time arithmetic and
procedures for computing indices in the fixed constructions, expander
rotations, and suspensions.
Bennett's reversible simulation gives clean implementations with
polynomial overhead~\cite[Theorem]{Bennett1973}.
The right-hand-side circuits use the fixed preparations of
Remark~\ref{rem:implementation-gate-set} and
Lemma~\ref{lem:lower-rhs-preparation}.
These circuits and all remaining fields are generated in polynomial
time from $C$; setting $\nu=n_C$ gives the stated length bounds.

\begin{remark}[Fixed sparse ordering]\label{rem:canonical-produced-ordering}
Each construction fixes the row and column labels and the order of
nonzero entries in every sparse list. Recovery may use these labels
and orders. It need not work after an arbitrary reordering of the same
matrix. The stated oracle costs include the computations needed to
find these indices.
\end{remark}

\begin{definition}[Square and rectangular block encodings]\label{def:block-encoding}
A unitary $U$ is an $(\alphabe(M),a,\varepsilon)$ block encoding of a
square matrix $M$ on $\mathcal H$ if
\[
\norm{M-\alphabe(M)(\bra0^{\otimes a}\otimes I_{\mathcal H})
U(\ket0^{\otimes a}\otimes I_{\mathcal H})}\le\varepsilon.
\]
Here $\alphabe(M)>0$ is the known factor that scales the indicated
block of the unitary to represent the matrix. In the constructions
below, the rescaling
factor is $\alpha(M)=\alphabe(M)$; we use $\alphabe(M)$ in the
implementation and error bounds.
For rectangular $A:\mathcal H_{\rm col}\to\mathcal H_{\rm row}$,
let $J_{\rm col},J_{\rm row}$ embed the logical spaces in a common
power-of-two register. A projected encoding has
\[
\norm{A-\alphabe(A) J_{\rm row}^{\dagger}
(\bra0^{\otimes a}\otimes I)U_A
(\ket0^{\otimes a}\otimes I)J_{\rm col}}\le\varepsilon.
\]
The error bound constrains only this block on the logical spaces.
\end{definition}

We next construct block encodings from the supplied sparse oracles.
Registers may contain unused coordinates because their dimensions are
powers of two. We project onto the logical spaces so that these
coordinates do not contribute when we multiply block encodings of
boundary matrices.
The construction below and its product bounds give the operations,
costs, and error bounds used in the quantum algorithms of
Theorems~\ref{thm:containment-final} and~\ref{thm:lower-full-containment}.

\paragraph{Projection onto the logical spaces.}
Put $P_b=J_bJ_b^\dagger$ for
$b\in\{\mathrm{row},\mathrm{col}\}$ and introduce separate flags:
\[
F_b=P_b\otimes I_{f_b}+(I-P_b)\otimes X_{f_b},\qquad
\widetilde U_A=F_{\rm row}
(U_A\otimes I_{f_{\rm row}f_{\rm col}})F_{\rm col}.
\]
Index-range tests implement $F_b$.
The zero-ancilla block of $\widetilde U_A$ is the corresponding block
of $U_A$ projected on both sides, so $\widetilde U_A$ is an ordinary
$(\alphabe(A),a+2,\varepsilon)$ encoding of
$\widetilde A=J_{\rm row}AJ_{\rm col}^\dagger$.
We keep these projections whenever we multiply block encodings to
form Gram matrices of boundary matrices. Reversing the unitary encodes
the adjoint. Products with separate ancillas encode
$\widetilde A^\dagger\widetilde A$ or
$\widetilde A\widetilde A^\dagger$, restricting to $A^\dagger A$ or
$AA^\dagger$ on the logical spaces.

\begin{lemma}[Block encoding from row and column access]\label{lem:sparse-access-block-encoding}
Let $A$ have sparsities $s_{\rm row},s_{\rm col}$ and a known entry
bound $\beta_A>0$ satisfying $\|A\|_{\max}:=\max_{i,j}|A_{ij}|\le\beta_A$.
Suppose its row and column sparse oracles agree, return zero on dummy
slots, and satisfy $\delta_{\rm val}\sqrt{s_{\rm row}s_{\rm col}}\le\varepsilon_A/2$.
For every $\varepsilon_A>0$, there is a projected
$(\alphabe(A),a_A,\varepsilon_A)$ encoding with
\[
\alphabe(A)=\beta_A\sqrt{s_{\rm row}s_{\rm col}},\qquad
 a_A=O(\log(m+n)+\log(s_{\rm row}s_{\rm col})+p_{\rm arith}),
\]
using $O(s_{\rm row}+s_{\rm col})$ calls to the oracles and their inverses.  The arithmetic register widths $p_{\rm arith}$ are polynomial in the
bit lengths of the matrix entries and scalar parameters and in
$\log(1+\alphabe(A)/\varepsilon_A)$. The number of elementary gates is
polynomial in these quantities, the sparsity bounds, and the address widths.
\end{lemma}
\begin{proof}[Using the standard sparse-access construction]
Clip the returned values to $[-\beta_A,\beta_A]$, obtaining $\widetilde A$.
Apply the state-preparation construction of
\cite[full version, Lemmas~47--48]{GSLW2019} to
$\widetilde A/\beta_A$, padded with zeros to a common register.
To obtain the required location and value access, we reversibly scan
the sparse lists with the XOR oracles. Each active index has a unique
position in its list. During this conversion, dummy slots receive
distinct labels outside the valid index range.
The scans and their inverses use $O(s_{\rm row}+s_{\rm col})$ queries,
with polynomial reversible arithmetic. The resulting normalization is
$\alphabe(A)=\beta_A\sqrt{s_{\rm row}s_{\rm col}}$.
Clipping does not increase the entry error, and
\[
\|\widetilde A-A\|
\le\sqrt{\|\widetilde A-A\|_1\|\widetilde A-A\|_\infty}
\le\delta_{\rm val}\sqrt{s_{\rm row}s_{\rm col}}
\le\varepsilon_A/2.
\]
Use the remaining error allowance to implement arithmetic and state
preparation.
Use precision $O(\min\{1,\varepsilon_A/\alphabe(A)\})$, with a sufficiently
small constant. The logical projections above give the rectangular
encoding.
\end{proof}

\paragraph{Block encodings of the Hodge Laplacians.}
The standard product and linear-combination constructions~\cite[full version, Lemmas~52--53]{GSLW2019} give the following bounds. For a boundary matrix $B_j$, an
$(\alphabe(B_j),a_j,\varepsilon_{B_j})$ block encoding yields either
Gram matrix with normalization $\alphabe(B_j)^2$ and error
\[
2\norm{B_j}\varepsilon_{B_j}+\varepsilon_{B_j}^2.
\]
For the full Hodge Laplacian, we use a linear combination of unitaries
(LCU). Its normalization is the sum of these squared normalizations.
Its error is at most the sum of the product errors. Boundary matrices
entering or leaving a zero chain space and their Gram-matrix terms
are omitted. These constructions are also used in
Theorem~\ref{thm:lower-full-containment}.

\begin{definition}[Block-encoding normalizations]\label{def:be-normalization}
For each supplied boundary matrix $B_j=\partial_j$, let
$s_{{\rm row},j},s_{{\rm col},j}\ge1$ be its supplied integer sparsity
bounds. Since its entries lie in $\{-1,0,1\}$, use entry bound one
in Lemma~\ref{lem:sparse-access-block-encoding}, giving
\[
\alphabe(B_j)=\sqrt{s_{{\rm row},j}s_{{\rm col},j}}\ge1.
\]
The product and sum constructions fix
\[
\begin{aligned}
\alphabe(L_k^\uparrow)&=\alphabe(B_{k+1})^2,\\
\alphabe(L_k^\downarrow)&=\alphabe(B_k)^2,\\
\alphabe(\Delta_k)&=
\sum_{\substack{j\in\{k,k+1\}\\B_j\text{ present}}}
\alphabe(B_j)^2.
\end{aligned}
\]
Here a boundary matrix is present when both its chain spaces are
nonzero. A boundary matrix entering or leaving a zero chain space is
omitted and has no assigned normalization.
The promise that the right-hand side has a nonzero projection onto the
range ensures that at least one boundary matrix is present.
We compute these normalizations from the supplied sparsity bounds.
They are not separate input fields. For every promised boundary matrix or
Hodge Laplacian $A$, the construction gives
$\norm A\le\alphabe(A)$. Spectral promises refer to
$A/\alphabe(A)$, while the linear systems use $A$ itself.
Polynomial sparsity bounds give polynomial normalizations; on the
fixed-dimension hard families, both are constant.
\end{definition}

A reduction must also recover the source solution from an approximate
target solution. The following definitions specify the recovery procedure,
how it reports failure, and the prepared right-hand side that determines
the target solution. Lemma~\ref{lem:composed-heralded-accounting} uses
these definitions to bound the probability of success and the error
in the recovered state when the reductions are applied in sequence.

\paragraph{Reporting successful recovery and computing matrix entries.}
An operation is \emph{heralded} when a measured flag identifies its
success. Every such recovery operation succeeds with probability at least an
inverse polynomial. A reduction is local with constant query overhead
if each target query uses $O(1)$ source queries and computes addresses reversibly in
polynomial time. Recovery costs may depend polynomially
on sparsity, normalization, inverse spectral bounds, and the prescribed
precision.

\begin{definition}[Reductions between state-preparation tasks]\label{def:state-red}\label{def:bqp-hard-complete}
A state-preparation reduction $\mathsf P\le_T\mathsf Q$ has two parts.
First, a classical polynomial-time procedure constructs circuits for
the target oracles from the source oracles. Second, a polynomial-time
quantum procedure recovers a valid source state from every valid
target state. All queries used to construct or simulate target oracles,
or to recover the source solution, are counted. The reduction specifies
the required target accuracy.
\end{definition}

\paragraph{Decision hardness and state preparation.}
For oracle inputs, the hardness reduction runs in classical polynomial
time and outputs $1^\nu$, the finite description, and polynomial-size
circuits implementing every target oracle. It maps each bounded-error
quantum circuit-acceptance instance to a promised instance of the
decision problem in Definition~\ref{def:measure-wrapper}, with the same
yes or no answer. Preparing sufficiently accurate solution states and
applying the prescribed recovery and measurement to each copy decides
that problem with bounded error after polynomially many repetitions.
The same reduction therefore establishes $\BQP$-hardness of preparing
the solution state. Recovery is generated from the target description
and its permitted label, queries only the supplied target oracles, and
receives no separate source circuit. In this paper it consists of the
specified linear recovery maps and
the output measurement on the final clock values described in
Definition~\ref{def:measure-wrapper}.

\begin{definition}[Prepared right-hand sides and range overlap]\label{def:rhs-preparation-error}
The exact prepared state is $\ket{b_x}=\mathcal O_b\ket0$, and the
input also supplies $\mathcal O_b^\dagger$. This state determines the
normalized solution state and the acceptance probability in the
decision problem. For rectangular $A$ or positive-semidefinite $H$,
$\Pi_A$ or $\Pi_H$ projects onto the corresponding matrix range.
The input promises
\[
\rho_b=\norm{\Pi_A b_x}\ge1/\zeta_b>0
\quad\text{or}\quad
\rho_b=\norm{\Pi_H b_x}\ge1/\zeta_b>0,\qquad\zeta_b\ge1.
\]
We call this projection norm the range overlap. The bound uses the
norm, not its square, and ensures that the target solution is nonzero. If the prepared state lies in
the range, we may set $\zeta_b=1$. A reduction may prove
$\dist_{\rm ph}(\ket{b_x},\ket{b_{\rm ref}})\le\delta_{\rm rhs}$
for a normalized reference vector in the matrix range. This reference
vector and error bound are used only in the proof; neither is part of
the input. Lemma~\ref{lem:rhs-robustness} bounds how this error affects
the normalized solution and the projection of the prepared right-hand
side onto the matrix range.
\end{definition}

\begin{remark}[Preparation unitary on other inputs]\label{rem:rhs-oracle-extension}
The target state and the decision problem's acceptance probability
depend only on $\mathcal O_b\ket0$. Algorithms use the unitary and
its inverse only to prepare or unprepare that state and to implement
$\mathcal O_b(2\ket0\!\bra0-I)\mathcal O_b^\dagger$. This reflection
also depends only on the prepared state. Reductions obtain no
information from the unitary's action on states orthogonal to its
designated initial state.
\end{remark}

\begin{definition}[Access to an unweighted boundary matrix]\label{def:scaled-boundary-access}
An input for a boundary linear system supplies sparse access to the unweighted oriented
simplicial boundary matrix
$B_{k+1}:=\partial_{k+1}:C_{k+1}(K)\to C_k(K)$ and its positive
integer row and column sparsity bounds. Exact nonzero entries are
$\pm1$; finite-precision values satisfy Definition~\ref{def:sparse-access}.
The block-encoding normalization is computed from these sparsity bounds
as in Definition~\ref{def:be-normalization}, using $\beta_{B_{k+1}}=1$
and clipping rounded values to $[-1,1]$ if needed.
Spectral promises concern the exact matrix divided by this normalization;
errors in matrix entries and in approximating circuits are counted
separately.
\end{definition}

\section{Promise problems}\label{sec:promise-problems}

The reduction first produces difference--average linear systems, then
boundary linear systems, and finally Hodge Laplacian linear systems.
The definitions below specify the input, promises, and output
at each step. The input for a difference--average linear system also
provides oracle access to the occurrences
of each variable. These queries let the simplicial construction locate
the relevant equations without scanning the entire matrix, as required
by Theorem~\ref{thm:bc-da-to-b2-wrapper-explicit}.

\begin{definition}[Difference--average matrices and bounded occurrence]\label{def:literal-da}
A difference--average (DA) linear system has equations of the form
$x_u-x_v=\beta$ (Diff) or $x_u+x_v-2x_w=0$ (Avg).
We call its coefficient matrix a difference--average matrix.
A difference equation has coefficient vector $e_u-e_v$. Its
right-hand-side entry $\beta$ is rational with a polynomial-bit
description and need not belong to $\{-1,0,1\}$. An average equation
has right-hand side zero.

Variables may repeat in an equation. We collect their coefficients
before applying the reduction. The nonzero collected coefficients lie
in $\{-2,-1,1,2\}$, and each row has $\ell_1$-norm at most four.
We represent a coefficient $a_{q,i}$ by $|a_{q,i}|$ unit occurrences,
each with the coefficient's sign. The coefficient multiplicity of a
variable is $u_i=\sum_q|A_{q,i}|$, and the largest multiplicity is
$\operatorname{eocc}(A)=\max_i u_i$. Thus coefficients $\pm2$ count
twice. Bounded occurrence means $\operatorname{eocc}(A)\le b_\ast$.
\end{definition}

This extends the class of difference--average matrices in Ding et al.~\cite{DingFull}, whose
nonzero rows have two distinct entries $+1,-1$ or three distinct entries
$+1,+1,-2$. An average row with repeated variables may become a
difference, a row with coefficients $+2,-2$, or a zero row
(Remark~\ref{rem:zero-columns-da}). The conversion to difference--average equations may produce
right-hand-side entries $c_r/P_\ast$ for a fixed power of two $P_\ast$.

\begin{definition}[$\BoundedOccurrenceDA_{b_\ast,\cond,\varepsilon}$]\label{def:dabo}
\emph{Input.} An instance supplies sparse access to a difference--average matrix, a
preparation oracle for a real unit vector $b\in\mathbb R^m$, and mutually inverse
reversible occurrence oracles with interfaces
\[
\begin{aligned}
(q,a)&\longmapsto (i,\sigma,t)\ \text{or a dummy},\\
(i,t)&\longmapsto (q,a,\sigma)\ \text{or a dummy}.
\end{aligned}
\]
Here $q\in[m]$, $a\in[4]$, $i\in[n]$, $\sigma\in\{\pm1\}$,
and $t\in[b_\ast]$ indexes the occurrence in column $i$. The two maps
agree on active slots, and each active slot has a unique
inverse. They represent $\pm2$ by two slots with the same sign.
Unused slots return a fixed dummy symbol that differs from every
valid occurrence record.

\emph{Promise.} The instance satisfies
$\operatorname{eocc}(A)\le b_\ast$,
$b\in\im(A)\setminus\{0\}$, and $\cond(A)\le\cond$.
Either there are no zero columns, or supplied indexing oracles allow
restriction to the nonzero columns and recovery by inserting zero
amplitudes at the removed indices. The constructed hard family has no zero columns.
The problem does not require $\norm A\le1$.

\emph{Output.} Prepare $\ket{x_A}=A^+\ket b/\norm{A^+b}$ within
$\varepsilon$; feasibility ensures it is nonzero.
\end{definition}

\paragraph{Occurrence records.}
Each row requires at most four records of signed unit occurrences,
using the address width specified in the input. These records are computed by the reduction on constructed
instances and are supplied by the occurrence oracles for arbitrary difference--average inputs.

\begin{remark}[Unnormalized representatives]\label{rem:demand-representatives}
In algebraic identities, we use nonzero, unnormalized vectors to
represent the prepared right-hand-side states. Thus $DM^+d=A^+b$
means that applying $D$ and normalizing its nonzero output recovers
the same normalized state, up to phase. The preparation oracle supplies
$d/\norm d$. Rescaling by a nonzero scalar rescales the pseudoinverse
vector by the same amount and preserves its state. The scalars used
in our constructions are positive rational numbers or fixed algebraic
normalization constants with polynomial-bit descriptions.
\end{remark}

\begin{remark}[Zero columns, zero rows, and locality]\label{rem:zero-columns-da}
Finding all zero columns may require querying the entire matrix. Thus
Definition~\ref{def:dabo} must either exclude zero columns or supply
oracles that index the nonzero columns. Deleting zero columns preserves
feasibility and all nonzero singular values. A zero column adds an independent coordinate in the kernel;
the minimum-norm solution sets that coordinate to zero. Recovery
inserts the zero amplitudes at their original indices. Every retained
variable therefore has coefficient multiplicity $u_i\ge1$, with
occurrence access as in Definition~\ref{def:dabo}.

Rows that become zero after collecting coefficients, including
$x_u-x_u=0$ and $x_u+x_u-2x_u=0$, keep their labels and zero
right-hand sides. We need not remove these rows or reindex the matrix.
A zero row with a nonzero right-hand side is infeasible. The promises
exclude $A^+b=0$, which cannot be normalized to a state. Every nonzero
difference--average matrix has an entry of magnitude at least one,
which gives the bound $\norm A\ge1$ used later.
\end{remark}

\begin{definition}[$\BoundaryQLSP_{(s_{\rm row},s_{\rm col}),\cond,\varepsilon}$]\label{def:boundary-QLSP}
\emph{Input.} An instance supplies sparse access to
$B_2=\partial_2:C_2(K)\to C_1(K)$ of a two-dimensional simplicial
complex, its sparsity bounds, a preparation oracle
$\mathcal O_\gamma$ and its inverse $\mathcal O_\gamma^\dagger$,
and $\zeta_\gamma$. The prepared state is
$\ket{\gamma_x}=\mathcal O_\gamma\ket0$.

\emph{Promise.} With
$\rho_\gamma=\norm{\Pi_{B_2}\gamma_x}$, the instance satisfies
\[
\frac{\sigma_{\min}^+(B_2)}{\alphabe(B_2)}\ge\frac1\cond,
\qquad \rho_\gamma\ge1/\zeta_\gamma>0.
\]
The normalization is computed as in Definition~\ref{def:be-normalization}.
The nonzero singular values of $B_2/\alphabe(B_2)$ therefore lie in
$[1/\cond,1]$. The parameter $\cond$ bounds the inverse of the
smallest positive singular value after normalization; a bound on
$\cond(B_2)$ alone does not imply this promise.

\emph{Output.} Prepare
$\ket f=B_2^+\ket{\gamma_x}/\norm{B_2^+\gamma_x}$ within
$\varepsilon$. The range-overlap promise ensures a nonzero target.
\end{definition}

The constructed hard instances have exact range membership and
$\zeta_\gamma=1$.

\begin{definition}[$\UpperHodgeQLSP_{k,\cond,\varepsilon}$]\label{def:quphls}
\emph{Input.} For fixed $k\ge1$, an instance supplies sparse access
to $B_{k+1}=\partial_{k+1}:C_{k+1}(K)\to C_k(K)$ of an oriented
simplicial complex $K$, its sparsity bounds, a preparation oracle
$\mathcal O_\gamma$ and its inverse, and $\zeta_\gamma$. Set
$L_k^\uparrow=B_{k+1}B_{k+1}^\top$ and
$\ket{\gamma_x}=\mathcal O_\gamma\ket0$.

\emph{Promise.} With
$\rho_\gamma=\norm{\Pi_{L_k^\uparrow}\gamma_x}$, the instance satisfies
\[
\frac{\lambda_{\min}^+(L_k^\uparrow)}
{\alphabe(L_k^\uparrow)}\ge\frac1\cond,
\qquad\rho_\gamma\ge1/\zeta_\gamma>0.
\]
Here $\alphabe(L_k^\uparrow)=\alphabe(B_{k+1})^2$
is computed as in Definition~\ref{def:be-normalization}.
The nonzero eigenvalues of
$L_k^\uparrow/\alphabe(L_k^\uparrow)$ lie in $[1/\cond,1]$.
The parameter $\cond$ bounds the inverse of the smallest positive
eigenvalue after normalization; it need not equal the condition number. The spectral promise concerns the exact matrix; entry
approximation contributes separately to the block-encoding error.

\emph{Output.} Prepare
$\ket z=(L_k^\uparrow)^+\ket{\gamma_x}/\norm{(L_k^\uparrow)^+\gamma_x}$
within $\varepsilon$. The range-overlap promise ensures a nonzero target.
\end{definition}

The quantum upper bound requires
$\cond,\zeta_\gamma,\varepsilon^{-1}$ and the sparsity bounds to be
polynomial in $\nu$, or in $|x|$ for finite encodings.
The constructed hard instances have $\rho_\gamma=1$ and constant
normalization.

\begin{definition}[$\LowerHodgeQLSP_{k,\cond,\varepsilon}$ and $\FullHodgeQLSP_{k,\cond,\varepsilon}$]\label{def:lower-full-hodge}
\emph{Input.} For fixed $k\ge1$, an instance supplies
sparse access to $B_k=\partial_k:C_k(K)\to C_{k-1}(K)$ of an
oriented simplicial complex $K$ and its sparsity bounds. The task for full Hodge Laplacian linear systems also requires
$B_{k+1}=\partial_{k+1}:C_{k+1}(K)\to C_k(K)$ of the same complex
and its sparsity bounds. Boundary matrices entering or leaving a zero chain space are omitted.
The matrices are
\[
L_k^\downarrow=B_k^\top B_k,\qquad
\Delta_k=B_k^\top B_k+B_{k+1}B_{k+1}^\top.
\]
Each input also supplies a preparation oracle $\mathcal O_\eta$,
its inverse, and $\zeta_\eta$. The prepared state is
$\ket{\eta_x}=\mathcal O_\eta\ket0$.
The task for full Hodge Laplacian linear systems allows $C_{k+1}=\{0\}$, hence $B_{k+1}=0$; in this
case its Gram-matrix term and block encoding are omitted.

\emph{Promise.} For $M=L_k^\downarrow$ or $\Delta_k$, respectively,
the instance satisfies
\[
\frac{\lambda_{\min}^+(M)}{\alphabe(M)}\ge\frac1\cond,
\qquad\rho_\eta=\norm{\Pi_M\eta_x}\ge1/\zeta_\eta>0.
\]
The normalization is computed as in Definition~\ref{def:be-normalization}.
The nonzero eigenvalues of $M/\alphabe(M)$ lie in $[1/\cond,1]$.
The parameter $\cond$ bounds the inverse of the smallest positive
eigenvalue after normalization; it need not equal the condition number.

\emph{Output.} Prepare, respectively,
$\ket{z_\downarrow}=(L_k^\downarrow)^+\ket{\eta_x}/\norm{(L_k^\downarrow)^+\eta_x}$
and $\ket{z_{\rm full}}=\Delta_k^+\ket{\eta_x}/\norm{\Delta_k^+\eta_x}$
within $\varepsilon$, after minimizing over global phase. The promise
on the projection onto the range ensures that the prepared state
determines a nonzero target.
\end{definition}

The quantum upper bounds require
$\cond,\zeta_\eta,\varepsilon^{-1}$ and the sparsity bounds to be
polynomial in $\nu$, or in $|x|$ for finite encodings.
Hardness with both summands nonzero is stated separately in
Corollary~\ref{cor:full-both-nonzero}. In the constructed hard instances,
the clean preparation circuits produce states whose errors relative
to reference vectors in the range are inverse-polynomial. These
reference vectors are used only in the hardness proof, through Lemma~\ref{lem:rhs-robustness}.

For boundary linear systems and Hodge Laplacian linear systems, the
prepared right-hand side may have a component orthogonal to the matrix
range. The pseudoinverse removes that component. The promise of a
nonzero projection onto the range ensures that the remaining solution
can be normalized.

The remaining definitions give the decision problems for these
state-preparation tasks, using either oracle access or finite circuit
descriptions. They specify the supplied information, the promises,
and the recovery and measurement procedure. Its acceptance probability
includes failed recovery attempts. This is the probability tested in
both Theorem~\ref{thm:hardness-final} and
Corollary~\ref{cor:complete-final}.

\begin{definition}[Sparse instance encoding]\label{def:instance-encoding}\label{def:produced-language}
The finite description is
\[
\mathsf{Enc}_{\rm SA}(x)=
(\mathsf{dims},\mathsf{addr},\mathsf{types},\mathsf{slots},
\mathsf{scalars},\mathsf{templates},\mathsf{rhs},\mathsf{bounds}).
\]
These fields give the dimensions, address lengths, and row and column
sparsity bounds. They specify how to interpret indices and order
sparse slots, as well as the numerical precision and scalar parameters
needed by the constructions. They also identify the fixed
difference--average row constructions, surface types, and expander
rotation maps. The remaining fields give the preparation rules and
the spectral, range-overlap, recovery, accuracy, and decision bounds.
The block-encoding normalizations for boundary matrices and Hodge
Laplacians are computed from the sparsity bounds as in Definition~\ref{def:be-normalization}.
No separate verifier, unrestricted program, oracle-value table, or
additional arbitrary data are supplied. Other instance-dependent
values are accessed only through the supplied oracles, with every query counted.

Let $\mathrm{WF}^{\rm SA}_{\mathsf S}$ denote descriptions with the
required fields, formats, and identifiers for task $\mathsf S$.
These requirements are checked from the finite description.
Correct oracle action, sparsity, spectral bounds, and range overlap
are separate input promises and need not be verified.
Write $\mathrm{WF}_{\mathsf S}$ when the access model is clear.
\end{definition}

\begin{definition}[Finite circuit inputs]\label{def:circuit-encoded-input}
For a finite input, the circuits implementing the oracles are included
in the following designated fields:
\[
\mathsf{Enc}_{\rm circ}(x)=(1^\nu,\mathsf{Enc}_{\rm SA}(x),
\langle C_{{\rm row},1}\rangle,\langle C_{{\rm col},1}\rangle,\ldots,
\langle C_{\rm rhs}\rangle).
\]
Each present boundary matrix has one pair of row and column oracle
circuits. A fixed list gives the circuits for any required scalar or
auxiliary oracles. Each supplied circuit must meet the specified
polynomial length bound in $\nu$ and use the required input and output
registers. We check these requirements from its description. The
input cannot contain a separate verifier or free-form program.

Each lookup applies the prescribed XOR on every basis state of its
input registers, using rounded numerical values and the prescribed
dummy outputs. It returns all workspace qubits to $\ket0$ and exposes
no other output. The clean unitary $C_{\rm rhs}$ defines the exact
target right-hand side $C_{\rm rhs}\ket0$; reversing the circuit
supplies its inverse. Entry errors must satisfy the fixed-point
bounds and the operator estimate in Lemma~\ref{lem:sparse-access-block-encoding}.
Reference vectors are used only for comparison in the proof. We denote
well-formed finite inputs by $\mathrm{WF}^{\rm circ}_{\mathsf S}$.
\end{definition}

\begin{lemma}[Circuit simulation]\label{lem:circuit-input-simulation}
An efficient uniform oracle algorithm on
$\mathrm{WF}^{\rm circ}_{\mathsf S}$ can be simulated with polynomial
gate complexity in the finite input length $|x|$.
In particular, its bounded-error decision problem is in
$\BQP$. Lemma~\ref{lem:finite-oracle-encoding-size}
supplies the corresponding finite-string hardness reductions.
\end{lemma}

This is standard circuit substitution: execute each supplied oracle
circuit, or its inverse, using its description as classical input.
For the standard universal-circuit formulation,
see~\cite[Definition~1 and Theorem~9]{BeraEtAl2010}.
Here each gate acts on at most a fixed number of qubits. Controlled gates and swaps,
selected by the input gate labels and wire indices, give a uniform
implementation with polynomial overhead.
The bounds on the number of queries and additional gates, together
with $\nu\le|x|$, give polynomial total cost. Since each circuit
returns its workspace to zero, the circuits can be applied in sequence.
If the resulting $M$ gates are approximated over another
fixed universal gate set, operator-norm error $\epsilon/M$ per gate
gives total error at most $\epsilon$. This contributes to the
acceptance-probability error; the target state is still exactly
$C_{\rm rhs}\ket0$.

\begin{remark}[Source-circuit dependence]\label{rem:circuit-dependent-oracles}
In the hardness reductions, the source gate list is not an additional
field of $\mathsf{Enc}_{\rm SA}$. Values that depend on the source
circuit are available only through the supplied oracles, and every
query is counted. For finite inputs, these data occur inside the
designated row, column, and preparation circuits. Neither model
supplies an independent decision circuit.
\end{remark}

\paragraph{Prescribed recovery and measurement.}
For a task $\mathsf S$, a label
$\tau=(\mathsf{mode},\mathsf{indices})$ selects its prescribed recovery
followed by the history-state output test. The allowed modes are
$\mathsf{UHQLSP\mbox{-}OUTPUT}$, $\mathsf{LHQLSP\mbox{-}OUTPUT}$,
$\mathsf{FHQLSP\mbox{-}OUTPUT}$, and $\mathsf{DA\mbox{-}OUTPUT}$.
The label supplies only the required row, column, clock, and
output-register indices. Their number is fixed, and we check them
against the dimensions in the input. The label contains no circuit
or other program.

A deterministic classical algorithm constructs, in polynomial time
from $(\mathsf{Enc}_{\rm SA}(x),\tau)$, an oracle circuit
$W_{x,\tau}^{\mathcal O_x}$. It applies the specified recovery maps,
measures the clock, and measures the designated output qubit when the
clock lies in the final clock interval. If recovery reports failure
or the clock lies outside that interval, it outputs an independently
sampled uniformly random bit. Every oracle query is counted.

\begin{definition}[Decision problems with prescribed recovery]\label{def:measure-wrapper}\label{def:canonical-readout-tags}
Use the preceding recovery convention for task $\mathsf S$ and label $\tau$.
For the target state $\ket{\psi_x}$ determined by the actual prepared
right-hand side, define
\[
p_\star=\Pr[W_{x,\tau}^{\mathcal O_x}(\ket{\psi_x})=1].
\]
The promise is
\[
p_\star\ge\tfrac12+\Delta
\quad\text{or}\quad
p_\star\le\tfrac12-\Delta.
\]
The sparse-oracle input to $\VerifyProblemSA\text{-}\mathsf S$ is
$(1^\nu,\mathsf{Enc}_{\rm SA}(x),\mathcal O_x,\tau,\Delta)$, with
$x\in\mathrm{WF}^{\rm SA}_{\mathsf S}$ and
$\Delta^{-1}\le\poly(\nu)$.
The finite input to $\VerifyProblemCirc\text{-}\mathsf S$ is
$(\mathsf{Enc}_{\rm circ}(x),\tau,\Delta)$, with
$x\in\mathrm{WF}^{\rm circ}_{\mathsf S}$ and
$\Delta^{-1}\le\poly(|x|)$.
Malformed descriptions are outside the promise. Only the supplied
instance and label determine $p_\star$; neither a reference state used
only in the proof nor an independent verifier enters its definition.
\end{definition}

\begin{remark}[Recovery from the solution state]\label{rem:measure-nonvacuity}
An arbitrary verification circuit supplied as input could decide the
source problem without using the solution state. Definition~\ref{def:measure-wrapper}
fixes the recovery and measurement procedure and allows it to query
only the supplied instance oracles. Removing a separate verifier
field does not settle the issue: the hardness proof must also show
that the prescribed procedure recovers the answer from the constructed
solution states. Section~\ref{sec:hodge-classification} shows that
the operations work in sequence as required. It bounds both the
probability of success and the error in the recovered state.
\end{remark}

\begin{remark}[Intermediate oracle access]\label{rem:measure-oracle-access}
Recovery may need to query an instance from an earlier stage of the
reduction. It must implement those queries using the supplied oracles
for the final instance and the classical parameters in the input.
Every query used in this implementation is counted. The input supplies
no additional oracle for an intermediate instance.
\end{remark}

\section{Linear algebra and Hodge decomposition}\label{app:audit-addendum}
The first lemma relates the minimum-norm solutions and spectra of a
boundary matrix and its Gram matrices. We use these identities to pass
to linear systems with an upper Hodge Laplacian or a lower Hodge
Laplacian in
Lemma~\ref{lem:bc-b2-to-quphls1-explicit} and
Theorem~\ref{cor:lower-full-hodge-hardness}.

The second lemma bounds how much normalization can increase vector
error. It is needed when recovery returns a vector with small norm and
when we compare solutions obtained from nearby right-hand sides; see
Lemmas~\ref{lem:bc-decoder-stability} and~\ref{lem:rhs-robustness}.
The last lemma shows that the ranges of the upper Hodge Laplacian and
the lower Hodge Laplacian are orthogonal.
Remark~\ref{rem:lower-full-scope} uses this fact to
explain a restriction of the hardness result for full Hodge Laplacian
linear systems.

\begin{lemma}[Pseudoinverse and Gram-matrix identities]\label{lem:standard-mp-identities}
For every real rectangular matrix $B$,
\[
(B^\top B)^+B^\top=B^+=B^\top(BB^\top)^+.
\]
These operator identities hold whether or not the right-hand side
lies in the matrix range.
The matrices $B^\top B$ and $BB^\top$ have the same nonzero
eigenvalues, namely the squares of the nonzero singular values of $B$.
Moreover, $\im(BB^\top)=\im B$, $\im(B^\top B)=\im B^\top$, and
\[
\|B^\top z\|\ge\sigma_{\min}^+(B)\|z\|\quad(z\in\im B),\qquad
\|B^+\|=1/\sigma_{\min}^+(B)
\]
when $B\ne0$.  These are standard pseudoinverse and singular-value
identities~\cite[Lemma, relation~1.5]{Penrose1955}
\cite[Eqs.~(1.153) and~(1.160)]{WatrousTQI}.
\end{lemma}

\begin{lemma}[Error after normalization]\label{lem:standard-normalization}
For nonzero vectors $u,v$ in a Hilbert space, the standard normalization
bound~\cite[Corollary~4, Eq.~(3.5)]{Dragomir2007Norm} gives
\[
\left\|\frac{v}{\|v\|}-\frac{u}{\|u\|}\right\|
\le\frac{2\|v-u\|}{\|u\|}.
\]
In particular, if $\|u\|\ge\sqrt p>0$ and $\|v-u\|\le\delta\le\sqrt p/2$,
then $\|v\|^2\ge p/4$ and the error between the normalized vectors is at most
$2\delta/\sqrt p$. For an operation with a flag that reports success,
$u$ and $v$ are
the unnormalized outputs for the successful outcome. The same bounds
hold after choosing the phase that minimizes the distance between
the input states. We need no bound on the output when the operation fails.
\end{lemma}

\begin{lemma}[Orthogonal Hodge decomposition]\label{lem:standard-hodge-decomposition}
For the real simplicial chain groups with the oriented simplices as an
orthonormal basis and unweighted boundary matrices $B_j=\partial_j$,
\[
C_k=\im(B_k^\top)\oplus\ker\Delta_k\oplus\im(B_{k+1})
\]
orthogonally, where $\ker\Delta_k=\ker B_k\cap\ker B_{k+1}^\top$.
Consequently $L_k^\downarrow L_k^\uparrow=L_k^\uparrow L_k^\downarrow=0$,
and the positive spectrum of $\Delta_k$ is the union, with multiplicity,
of the positive spectra of the lower Hodge Laplacian and the upper
Hodge Laplacian.
See~\cite[Theorem~2.2 and Eqs.~(2.3)--(2.6)]{HorakJost2013}.
\end{lemma}
\section{Dirichlet propagation linear system}
\label{sec:shifted-input}

We reduce circuit acceptance to a measurement of the normalized
solution state of a Dirichlet propagation linear system. This decision
problem is BQP-hard. Repeated preparation and measurement of a sufficiently
accurate solution state solve this decision problem, so state preparation
is also BQP-hard. A clock register records the circuit step, and the matrix
couples adjacent clock values through the circuit gates. The solution
records the intermediate computational states. The propagation operator
is related to the one used by Babbush et al. for coupled classical
oscillators~\cite{BabbushOscillators}. Their results concern
oscillator dynamics and kinetic-energy estimation. We prove the theorem
for linear systems below independently.

We use real gates because the boundary matrices in the later reductions are real.
We also require rational gate entries. Clearing their denominators gives a
propagation matrix whose entries lie in a fixed finite set of integers.
The results below show that these circuits can still simulate arbitrary
BQP computations. Theorem~\ref{thm:shifted-clock-hard} reduces their
acceptance to a measurement of a solution state. The integer entries meet
the coefficient assumptions in Section~\ref{sec:lg-da}.

\begin{definition}[Rational real gate set]
\label{def:rational-gate-set}
Set
\[
R:=\begin{pmatrix}3/5&4/5\\4/5&-3/5\end{pmatrix},
\qquad
\mathcal G_{\mathbb Q}:=\{R,X,Z,\mathrm{Toffoli}\}.
\]
These gates are real, orthogonal, symmetric, and rational.
The signed-permutation gate
\(Z=\operatorname{diag}(1,-1)\) is included to supply coefficient
patterns. It is not needed for universality.
\end{definition}

\begin{lemma}[Uniform approximation using \(\mathcal G_{\mathbb Q}\)]
\label{lem:uniform-gq-approx}
Consider a polynomial-time uniform family of polynomial-size real
circuits, indexed by input length \(n\). Assume its gate set is fixed
and finite, and its gate entries can be computed to \(b\) bits in
\(\poly(b)\) time. The family has a polynomial-time uniform
approximation using \(\mathcal G_{\mathbb Q}\), with size
\(\poly(n,\log(1/\eta))\) and operator-norm error \(\eta\).
This approximation uses a fixed number of additional rebits, which are
two-level registers with real amplitudes. The target circuit acts as the
identity on these rebits.
After amplifying the acceptance gap, choose inverse-polynomial \(\eta\)
below half the promised magnitude of the acceptance bias
about \(1/2\). Acceptance probabilities change by at most
\(\eta\), so the promise gap remains constant.
\end{lemma}
\noindent\emph{Why the chosen gate set suffices.}
Let $d$ be the largest number of wires on which any target gate acts, and set $c_{\rm gate}=\max\{4,d+1\}$.
Extending each gate by the identity on the additional rebits puts every
target gate and every generator acting on its chosen wires in $SO(2^{c_{\rm gate}})$.
The rotation $XR$ has $\cos\theta=4/5$ and $\theta/\pi$ irrational:
otherwise $2\cos\theta=8/5$ would be an algebraic integer,
whereas rational algebraic integers are integers.
Its powers are therefore dense in $SO(2)$.
For distinct wires $a,b,c$, let $T_{ab;c}$ be Toffoli with controls $a,b$
and target $c$. Then
\[
T_{ab;c}X_bT_{ab;c}X_b=\mathrm{CNOT}_{a;c}\otimes I_b.
\]
Thus CNOT is implemented for every state of the extra wire. Local real
rotations and CNOT generate $SO(2^{c_{\rm gate}})$~\cite[real-theory argument following Theorem~1, Eqs.~(8)--(11)]{Belenchia2013}.
There are finitely many ways to place generators on the wires. Each resulting
gate has rational entries and is its own inverse, so this set is closed under inverses.
The Solovay--Kitaev theorem therefore applies to this fixed compact,
connected, semisimple group~\cite[v2, Theorems~1.2 and~1.4 and Corollary~1.3]{KuperbergSK}.
For an $M$-gate circuit, approximate each gate, including its added identity
registers, to error $\eta/M$. Summing the errors across the circuit gives
operator-norm error at most $\eta$ on every input, including the added
rebits. The assumption on computing target entries gives the stated
uniform polynomial cost. Measurement probabilities change by at most $\eta$.

\begin{lemma}[Simulation of BQP circuits with real amplitudes]
\label{lem:real-bqp-reduction}
BQP circuits with complex amplitudes can be simulated using real
amplitudes, one additional rebit, and polynomial overhead. This preserves
every computational-basis output probability on the original wires.
Together with Lemma~\ref{lem:uniform-gq-approx}, this shows that a
hardness reduction from circuit acceptance over \(\mathcal G_{\mathbb Q}\)
establishes BQP-hardness. This is the standard real
simulation~\cite[Section~2.1, Definition~3 and Lemma~1]{McKague2011}.
\end{lemma}

\begin{definition}[Circuit acceptance over the rational real gate set]
\label{def:bqpcirc-q}
An instance of \(\mathsf{BQPCirc}_{\mathbb Q}\) describes
a circuit \(C\) with unitary \(U_C:=U_L\cdots U_1\) and length \(L\ge1\)
on \(n_q\) qubits, with gates in \(\mathcal G_{\mathbb Q}\),
initial state \(\ket{0^{n_q}}\), and a designated output
qubit. With
\[
\Pi_{\rm out}:=\ket1\!\bra1_{\rm out}\otimes I_{\rm rest},
\]
define
\[
p_{\rm acc}(C):=
\norm{\Pi_{\rm out}U_C\ket{0^{n_q}}}^2.
\]
Decide whether \(p_{\rm acc}(C)\ge2/3\) or
\(p_{\rm acc}(C)\le1/3\), promised one of these cases.
Equivalently, the thresholds are
\(1/2\pm\Delta_{\rm circ}\), where
\(\Delta_{\rm circ}:=1/6\).
\end{definition}

\begin{lemma}[\(\mathsf{BQPCirc}_{\mathbb Q}\) is \(\BQP\)-complete]
\label{lem:bqpcirc-q-complete}
\(\mathsf{BQPCirc}_{\mathbb Q}\) is \(\BQP\)-complete under classical polynomial-time many-one reductions.
\end{lemma}
This follows from the standard completeness result for circuit acceptance~\cite[Section~IV.1 and Proposition~3]{WatrousSurvey}, applied using Lemmas~\ref{lem:uniform-gq-approx} and~\ref{lem:real-bqp-reduction}.
Membership follows because the fixed computable gate set can be simulated uniformly.
For hardness, amplify the circuit so that its acceptance probability is at
least $5/6$ in the accepting case and at most $1/6$ in the rejecting case.
Then simulate it with real amplitudes and approximate its gates, with
acceptance-probability error at most $1/6$. The resulting probabilities
satisfy the required $2/3$ and $1/3$ thresholds.

The propagation matrix must have a condition number polynomial in circuit
length. Its solution state must also give the final circuit output with
constant probability. The Dirichlet boundary conditions give the first
property. Lemmas~\ref{lem:clock-conjugation} and~\ref{lem:path-green}
compute the eigenvalues from a one-dimensional path with these boundary
conditions. Appending identity gates gives the second property, as
Lemma~\ref{lem:weighted-history} shows.
Theorem~\ref{thm:shifted-clock-hard} uses both properties to read the
circuit output from a solution state prepared to constant accuracy.
The later reductions recover this state and use the same measurement.

\begin{definition}[Dirichlet propagation linear system with identity gates]
\label{def:shifted-clock}
For a circuit \(C=U_L\cdots U_1\) over
\(\mathcal G_{\mathbb Q}\) on \(n_q\) qubits, append
\(L\) identity gates. Set $T:=2L,\; m:=T+1,\;
U_\ell=\Id\;\text{for }L<\ell\le T$,
and define $V_0:=\Id,\; V_\ell:=U_\ell\cdots U_1
\;\text{for }1\le\ell\le T$. The clock register has basis states
\(\ket1,\ldots,\ket m\), where clock label \(\ell\)
corresponds to the computation after \(\ell-1\) gates.
On \(\mathbb R^m\otimes\mathbb R^{2^{n_q}}\), define
the \emph{Dirichlet propagation matrix}
\begin{equation}\label{eq:dirichlet-propagation-matrix}
K_C:=2\Id-\sum_{\ell=1}^{T}
\left(
\ketbraop{\ell}{\ell+1}\otimes U_\ell^\top
+
\ketbraop{\ell+1}{\ell}\otimes U_\ell
\right).
\end{equation}
Its diagonal blocks are \(2\Id\), and its adjacent
off-diagonal blocks couple successive clock values through
the circuit gates. Both blocks associated with step \(\ell\)
are \(-U_\ell\), since every gate is symmetric. Set $b_C:=\ket1_{\rm clock}\ket{0^{n_q}}_{\rm comp},\;G_C:=5K_C,\;c_C:=5b_C$. The associated \emph{Dirichlet propagation linear system} is
\begin{equation}\label{eq:dirichlet-propagation-system}
K_Cx_C=b_C,
\qquad\text{equivalently}\qquad
G_Cx_C=c_C.
\end{equation}
\end{definition}

Writing the unknown vector as
\[
x_C=\sum_{\ell=1}^{m}\ket\ell\otimes x_\ell,
\qquad x_\ell\in\mathbb R^{2^{n_q}},
\]
the system has the block equations
\begin{equation}\label{eq:dirichlet-propagation-blocks}
\begin{aligned}
2x_1-U_1^\top x_2
  &=\ket{0^{n_q}},\\
2x_\ell-U_{\ell-1}x_{\ell-1}-U_\ell^\top x_{\ell+1}
  &=0,
  &&2\le\ell\le m-1,\\
2x_m-U_{m-1}x_{m-1}
  &=0.
\end{aligned}
\end{equation}
The factor \(5\) clears all gate denominators, making
\(G_C\) and \(c_C\) integer-valued without changing the
solution vector.

Define the orthogonal change of coordinates and the Dirichlet path Laplacian by
\[
S:=\sum_{\ell=1}^{m}
\ketbraop{\ell}{\ell}\otimes V_{\ell-1},
\qquad
D_m:=2\Id_m-\sum_{\ell=1}^{m-1}
\left(
\ketbraop{\ell}{\ell+1}
+
\ketbraop{\ell+1}{\ell}
\right).
\]

\begin{lemma}[Change of coordinates to the Dirichlet path Laplacian]
\label{lem:clock-conjugation}
The Dirichlet propagation matrix satisfies
\[
K_C=S(D_m\otimes\Id)S^\top.
\]
It is positive definite, with each eigenvalue of \(D_m\)
repeated \(2^{n_q}\) times.
\end{lemma}

\begin{lemma}[Inverse and eigenvalues of the path matrix]
\label{lem:path-green}
Standard formulas for the inverse~\cite[Theorem~1, Eqs.~(12) and~(15)]{YuehInverse2006}
and eigenvalues~\cite[Section~3, Eq.~(12)]{YuehSpectrum2005} of a tridiagonal matrix give
\[
\begin{gathered}
D_m^{-1}\ket1
=\sum_{\ell=1}^{m}
\frac{m+1-\ell}{m+1}\ket\ell,\\
\lambda_j(D_m)
=2-2\cos\frac{j\pi}{m+1}
=4\sin^2\frac{j\pi}{2(m+1)},
\qquad j=1,\ldots,m.
\end{gathered}
\]
Together with Lemma~\ref{lem:clock-conjugation}, these formulas give
\[
\cond(G_C)=\cond(K_C)=\cond(D_m)
=\cot^2\!\left(\frac{\pi}{2(m+1)}\right)=\Theta(m^2).
\]
Since this is less than $(m+1)^2\le4m^2$ and $m=2L+1$,
$\cond(G_C)\le4(2L+1)^2$; thus $C_{\rm clock}=4$.
\end{lemma}

\begin{lemma}[The solution as a weighted history state]
\label{lem:weighted-history}
Writing
\(\ket{\psi_{\ell-1}}:=V_{\ell-1}\ket{0^{n_q}}\),
the normalized solution of
Eq.~\eqref{eq:dirichlet-propagation-system} is the
\emph{weighted history state}
\[
\ket{x_C}:=\frac{x_C}{\norm{x_C}}
=
\frac{1}{\sqrt{\sum_{j=1}^{m}j^2}}
\sum_{\ell=1}^{m}
(m+1-\ell)\ket\ell\ket{\psi_{\ell-1}}.
\]
In \(\ket{x_C}\), the probability of a clock value in
\(\mathcal T:=\{L+1,\ldots,2L+1\}\) is
\[
p_{\rm tail}
=\frac{\sum_{j=1}^{L+1}j^2}
       {\sum_{j=1}^{2L+1}j^2}
\ge\frac18.
\]
For every \(\ell\in\mathcal T\),
\(\ket{\psi_{\ell-1}}=U_C\ket{0^{n_q}}\).
Thus measuring the clock yields the circuit's output
state in the computational register with probability
at least \(1/8\).
\end{lemma}

\begin{definition}[Preparing the Dirichlet propagation solution state]
\label{def:propagation-qlsp}
The problem \(\FKClockLS_\varepsilon\) takes as input a classical
description of a circuit \(C\) over
\(\mathcal G_{\mathbb Q}\) and an accuracy
\(\varepsilon\in(0,1)\). The task is to prepare
\[
\ket{x_C}=\frac{G_C^{-1}c_C}{\norm{G_C^{-1}c_C}}
\]
to phase-invariant state error at most \(\varepsilon\),
where \((G_C,c_C)\) is constructed from \(C\) as in
Definition~\ref{def:shifted-clock}.
\end{definition}

\begin{theorem}[Hardness of preparing the Dirichlet propagation solution state]
\label{thm:shifted-clock-hard}
From an instance \(C\) of \(\mathsf{BQPCirc}_{\mathbb Q}\),
with description length \(n_C\), one can construct in
\(\poly(n_C)\) time a sparse-access description of
\((G_C,c_C)\). The matrix is real, symmetric,
positive definite, integer-valued, has constant row
and column sparsity, and has entries in a fixed finite
integer set independent of circuit length.
Moreover, \(\cond(G_C)=\Theta(L^2)\), and for a
universal constant \(\varepsilon_0>0\), preparing
the normalized solution state of
Eq.~\eqref{eq:dirichlet-propagation-system},
\[
\ket{x_C}
=
\frac{G_C^{-1}c_C}{\norm{G_C^{-1}c_C}}
=
\frac{K_C^{-1}b_C}{\norm{K_C^{-1}b_C}},
\]
to phase-invariant state error at most
\(\varepsilon_0\) is BQP-hard.
\end{theorem}

\subsection{Proofs for the propagation linear system}

\begin{proof}[Proof of Lemma~\ref{lem:clock-conjugation}]
$S$ is orthogonal because each $V_\ell$ is orthogonal. The diagonal blocks
of $S(D_m\otimes\Id)S^\top$ are $2\Id$, and its $(\ell,\ell+1)$ block is
$-V_{\ell-1}V_\ell^\top=-V_{\ell-1}(U_\ell V_{\ell-1})^\top=-U_\ell^\top$.
The block in the transposed position is $-U_\ell$. These are exactly the
blocks in Definition~\ref{def:shifted-clock}.
Since \(D_m\) is positive definite, this orthogonal change of coordinates
also gives the stated eigenvalues and their multiplicities. Under the
change of coordinates \(y_\ell:=V_{\ell-1}^\top x_\ell\), the Dirichlet
propagation linear system becomes
\begin{equation}\label{eq:dirichlet-clock-recurrence}
2y_\ell-y_{\ell-1}-y_{\ell+1}
=\delta_{\ell,1}\ket{0^{n_q}},
\qquad 1\le\ell\le m,
\end{equation}
with zero Dirichlet boundary conditions
\[
y_0=y_{m+1}=0.
\]
These conditions set the values at the two exterior clock positions \(0\) and \(m+1\) to zero;
they do not require the first and last solution components to vanish.
Positive definiteness gives a unique solution. Scaling the matrix and
right-hand side by the same factor leaves it unchanged:
\[
x_C=K_C^{-1}b_C=G_C^{-1}c_C.
\]
\end{proof}

\begin{proof}[Proof of Lemma~\ref{lem:weighted-history}]
The clock change of coordinates gives $K_C^{-1}b_C=S(D_m^{-1}\ket1\otimes\ket{0^{n_q}})$. Substituting the path inverse gives
\[
x_C=K_C^{-1}b_C
=\sum_{\ell=1}^{m}
\frac{m+1-\ell}{m+1}
\ket\ell\ket{\psi_{\ell-1}}.
\]
Normalizing gives the stated weighted history state. For $\ell\ge L+1$, the computational register is in the state $U_C\ket{0^{n_q}}$.
The weights for the final clock interval are proportional to $L+1,L,\ldots,1$,
and the weights for all clock values are proportional to $2L+1,2L,\ldots,1$.
Summing their squares gives
\[
p_{\rm tail}=\frac{(L+2)(2L+3)}{2(2L+1)(4L+3)},\qquad
p_{\rm tail}-\frac18=\frac{18L+21}{8(2L+1)(4L+3)}>0.
\]
Independent repetitions increase the probability of recovering the
circuit's output state.
\end{proof}

\begin{proof}[Proof of Theorem~\ref{thm:shifted-clock-hard}]
The matrix properties follow from Definition~\ref{def:shifted-clock}.
The diagonal entries of \(G_C\) are \(10\), and the off-diagonal gate
entries come from \(5X\), \(5Z\), \(5\mathrm{Toffoli}\), or \(5R\).
Thus all entries lie in a fixed finite set of integers. Each row of a
one-qubit \(R\) block has two nonzero entries, and each row of \(X\),
Toffoli, or \(\Id\) has one. Hence \(G_C\) has constant row and column
sparsity. The condition-number bound follows from
Lemmas~\ref{lem:clock-conjugation} and \ref{lem:path-green}, since scaling
by \(5\) does not change the condition number.

We reduce from \(\mathsf{BQPCirc}_{\mathbb Q}\), which is \(\BQP\)-complete
by Lemma~\ref{lem:bqpcirc-q-complete}. Let \(\Delta_{\rm circ}=1/6\) be
the lower bound on the magnitude of the circuit's acceptance bias relative
to \(1/2\), as in Definition~\ref{def:bqpcirc-q}.
Given the ideal normalized solution state, first measure the clock.
If its value lies in the final clock interval \(\mathcal T\), measure the
designated output qubit in the computational register. If the clock is
outside \(\mathcal T\), return an independently sampled fair random bit.
When the clock lies in this interval, Lemma~\ref{lem:weighted-history}
shows that the computational register is in the final circuit state
\(U_C\ket{0^{n_q}}\). Thus, if \(p_C\) is the probability of measuring $1$
on the circuit output qubit, this measurement returns $1$ with probability
\[
(1-p_{\rm tail})\frac12+p_{\rm tail} p_C
=\frac12+p_{\rm tail}\left(p_C-\frac12\right).
\]
Since \(p_{\rm tail}\ge1/8\), the magnitude of the acceptance bias is at least
\[
\Delta_{\rm FK}:=p_{\rm tail}\Delta_{\rm circ}\ge \frac{1}{48}.
\]
If a prepared state differs from the ideal solution state by at most
\(\varepsilon\) in the phase-invariant distance,
Lemma~\ref{lem:state-to-verifier-final} bounds the change in any verifier
acceptance probability by at most \(\varepsilon\). Thus the constant choice
\[
\varepsilon_0:=\frac1{192}
\]
leaves an acceptance bias of at least \(1/96\). A constant number of
independent repetitions therefore decides circuit acceptance with bounded
error, so this measurement defines a BQP-hard decision problem. Repeated
preparation of the solution state to the stated constant accuracy, followed
by this measurement, solves that problem. Hence preparing solution states
for these Dirichlet propagation linear systems is BQP-hard.
\end{proof}

\section{Reduction to bounded-occurrence difference--average equations}\label{sec:lg-da}

The surface construction in Section~\ref{sec:da-to-b2} requires difference--average equations with a bounded sum of absolute coefficients in each column. Both types of equation have coefficients summing to zero. One extra variable could make all row sums zero, but it could occur in too many equations. We instead add one variable for each row and require these variables to be equal along the edges of a constant-degree expander. We can query its edges without constructing the full graph. Its spectral gap also bounds differences between the added variables when the equality constraints hold only approximately. Lemmas~\ref{lem:H-exact} and~\ref{lem:H-condition} show how to recover the original solution and bound the new condition number under the stated coefficient, sparsity, and norm assumptions. We then replace the resulting rows by difference--average equations.

\begin{theorem}[Strongly explicit loopless expanders]\label{thm:certified-expander-family}
There are fixed constants \(d_E\in\mathbb N\) and rational
\(\underline\sigma_E,\overline\sigma_E>0\) such that, for every
\(r\ge1\), a connected loopless \(d_E\)-regular multigraph
\(E_{2^r}\) has a rotation map
\[
\mathsf{Rot}_E(r,v,\ell)=(w,\ell'),\qquad
v,w\in[2^r],\quad\ell,\ell'\in[d_E].
\]
The map is computable in \(\poly(r)\) time and satisfies \(w\ne v\)
and \(\mathsf{Rot}_E(r,w,\ell')=(v,\ell)\). A port label selects one
incident edge, and parallel edges have distinct port labels.
The incidence matrix \(B_E\in\mathbb R^{(2^rd_E)\times2^r}\)
includes both orientations of each edge, with
\((B_Eg)_{(v,\ell)}:=g_v-g_w\), and satisfies
\[
\begin{gathered}
B_E^\top B_E=2L(E_{2^r}),\qquad
\norm{B_E}\le\overline\sigma_E,\\
\norm{B_E\delta}\ge\underline\sigma_E\norm\delta
\quad\text{for }\delta\perp\mathbf1.
\end{gathered}
\]
Here \(L(E_{2^r})\) is the combinatorial graph Laplacian, counting
parallel edges. Both row and column sparse access are computable
in \(\poly(r)\) time.
\end{theorem}

\begin{remark}[Padding auxiliary variables]\label{rem:dummy-gauges}
For \(N\ge1\) original rows, choose
\(N_E:=2^{\max\{1,\lceil\log_2N\rceil\}}\), so
\(N\le N_E\le2N\), and use \(E_{N_E}\) from
Theorem~\ref{thm:certified-expander-family}. The same spectral bounds apply at every size, and each sparse query
takes \(\poly(\log N_E)\) time.
Let $g\in\mathbb R^{N_E}$ be the vector of auxiliary variables.
Let \(P_{\rm orig}:\mathbb R^{N_E}\to\mathbb R^N\) project onto the first \(N\) coordinates: \(P_{\rm orig}g=(g_1,\ldots,g_N)\). Algebraic label \(j\) has expander address \(j-1\), so the first \(N\) coordinates have addresses \([N]\). In the construction below, the added coordinates $g_{N+1},\ldots,g_{N_E}$ occur only in the homogeneous expander constraints. Only $P_{\rm orig}g$ enters the original equations and the recovered vector.
\end{remark}

\begin{definition}[Linear system with auxiliary variables]\label{def:auxiliary-variable-H}
Let \(G\in\mathbb Z^{N\times N}\) be invertible, with row and
column sparsity at most \(s\), coefficient magnitudes at most \(a\),
and \(0<\gamma\le\norm G\le\Gamma\), where
\(s,a,\gamma,\Gamma\) are fixed constants. Let
\(c\in\mathbb R^N\setminus\{0\}\). These assumptions hold for \(G_C=5K_C\); they need not hold for arbitrary sparse matrices. Set \(r_r:=\sum_jG_{rj}\), \(R_G:=\operatorname{diag}(r_1,\ldots,r_N)\), and use \(N_E,E_{N_E},P_{\rm orig}\) above. Introduce \(y\in\mathbb R^N\), \(g\in\mathbb R^{N_E}\). Let \(B_E^{(q_E)}\) repeat each incidence row indexed by a directed port
\(q_E\) times. Choose a fixed \(q_E\) with \(q_E\underline\sigma_E^2\ge4\max_r|r_r|^2+4\). Then \(q_E\underline\sigma_E^2-\norm{R_GP_{\rm orig}}^2\ge1\). Define
\[
H(y,g):=\begin{pmatrix}Gy-R_GP_{\rm orig}g\\B_E^{(q_E)}g\end{pmatrix},
\quad h:=\begin{pmatrix}c\\0\end{pmatrix},
\quad D_0(y,g):=y-P_{\rm orig}g.
\]
For \(N_E=N\), these become \(H(y,g)=(Gy-R_Gg,B_E^{(q_E)}g)\) and \(D_0(y,g)=y-g\).
\end{definition}

\begin{lemma}[Solutions after adding auxiliary variables]\label{lem:H-exact}
Writing \(u=(y,g)\), the exact solutions of \(Hu=h\) are precisely
\[
y=x+t\mathbf1_N,\qquad g=t\mathbf1_{N_E},
\qquad x=G^{-1}c,\quad t\in\mathbb R.
\]
Thus \(D_0(y,g)=G^{-1}c\) for every exact solution, and
\(\ker H=\Span\{(\mathbf1_N,\mathbf1_{N_E})\}\).
\end{lemma}

\begin{lemma}[Condition number of $H$]\label{lem:H-condition}
Under the fixed coefficient and sparsity bounds above and \(1\le\norm G\le\Gamma\),
\(\cond(H)\le C_H\cond(G)\) for an explicit rational \(C_H\) depending only on \(s,a,\Gamma\), the padding bound \(1\le N_E/N\le2\), and \(d_E,\underline\sigma_E,\overline\sigma_E\).
\end{lemma}

\begin{lemma}[Norm bounds for the auxiliary matrix]\label{lem:shifted-H-norm-window}
For \(H_C\) obtained from \(G_C=5K_C\), there are absolute constants, independent of circuit length, with \(0<\gamma_H\le\norm{H_C}\le\Gamma_H<\infty\).
\end{lemma}

We now replace each row whose coefficients sum to zero by difference and averaging equations. We add copy variables and averaging trees separately for each row. Using the same number of leaves in every tree rescales all nonzero right-hand-side entries by the same factor. Lemma~\ref{lem:lg-rhs-state} then prepares the new right-hand-side state by reversibly changing the row labels.

\begin{definition}[Replacing a zero-sum row]\label{def:da-row-template}
Consider \(a^\top v=\beta\), where \(v=(u_0,\ldots,u_p)\) lists
the shared variables supplied to the gadget. These variables, called shared terminals, may also occur in other constraints. Assume \(a\in\mathbb Z^{p+1}\setminus\{0\}\), \(p\le p_\ast\), \(|a_t|\le a_\ast\), and \(\sum_ta_t=0\), with fixed \(p_\ast,a_\ast\). Every Dirichlet propagation row remains nonzero after its auxiliary variable is added. Rows that become zero after coefficients are collected use Remark~\ref{rem:zero-columns-da} and need no gadget. The padding terminal \(u_0\) is \(g_r\) in the auxiliary-variable reduction. Set
\[
P(a):=\sum_{a_t>0}a_t=\sum_{a_t<0}(-a_t)>0,
\]
and choose a common power of two \(P_\ast\ge\max_aP(a)\).
\emph{Copy and pad.} For a positive or negative coefficient, respectively, create \(a_t\) or \(-a_t\) leaves and equate each leaf to \(u_t\) with a difference equation. Add \(P_\ast-P(a)\) \emph{additional} copies of \(u_0\) to each side. The two sides need not have the same number of original copies of \(u_0\).

\emph{Average.} Average each side's \(P_\ast\) leaves through a complete binary tree of equations \(x_u+x_v-2x_w=0\).

\emph{Impose the root difference.} For roots \(r_+,r_-\), impose \(r_+-r_-=\beta/P_\ast\). The additional copies cancel, so the root difference is \(a^\top v/P_\ast\). All newly introduced variables are private to this gadget.
\end{definition}

The residual of an equation is the difference between its left-hand and right-hand sides. The next lemma compares the residuals of a row and its difference--average replacement, including the effect of errors in the internal variables.

\begin{lemma}[Comparison of residuals for one row]\label{lem:template-stability}
For each fixed row construction \(T\), let \(v=(u_0,\ldots,u_p)\), and let \(w\) contain all copy leaves and internal averaging-tree vertices, including both roots. At \(P_\ast=1\), a root is its leaf and is counted once. Let \(T_T\) be the gadget constraint matrix on \((v,w)\). Let \(d_T(\beta)\) be zero except for the final difference-row entry \(\beta/P_\ast\). The linear map \(W_T\) fills in the leaves with their terminal values, then computes the averages up to the roots. Constants \(0<c_T\le C_T<\infty\), uniform over the finite family, satisfy, for all \(v,w,\beta\),
\[
\begin{aligned}
c_T\bigl(|a^\top v-\beta|^2+\norm{w-W_Tv}^2\bigr)
&\le\norm{T_T(v,w)-d_T(\beta)}^2\\
&\le C_T\bigl(|a^\top v-\beta|^2+\norm{w-W_Tv}^2\bigr).
\end{aligned}
\]
Also \(\norm{W_Tv}\le C_T\norm v\) and
\(\min_w\norm{T_T(v,w)-d_T(\beta)}\asymp_T|a^\top v-\beta|\), where \(\asymp_T\) denotes comparison by positive constants depending only on \(T\).
\end{lemma}

\begin{definition}[The DA reduction]\label{def:lg-da-compiler}
Form \((H,h)\) from \((G,c)\). Replace each original row
\(\sum_jG_{rj}y_j-r_rg_r=c_r\) by the DA gadget of Definition~\ref{def:da-row-template}, using common \(P_\ast\). Here \(g_r=(P_{\rm orig}g)_r\), and the row coefficients are bounded integers summing to zero. Each repeated row indexed by a directed port is already \(g_r-g_{r'}=0\). Padding vertices have no original-row gadget. Denote the resulting DA linear system by \(M_{\rm aux}z=d_{\rm aux}\).
For \(z=(u,w)\), where \(u=(y,g)\) and \(w\) contains the
internal gadget variables, define \(D_{\rm aux}(u,w):=D_0(u)\) on every resulting DA instance, so
\[
(D_{\rm aux}z)_i:=y_i-(P_{\rm orig}g)_i=y_i-g_i.
\]
It ignores internal gadget coordinates and padding variables, as in Remark~\ref{rem:dummy-gauges}.  In the Dirichlet family, source rows and columns are indexed by pairs consisting of a clock value and a computational-basis string.
\end{definition}

The surface construction needs a bound on how often each variable occurs and an efficient way to find those occurrences. Lemma~\ref{lem:lg-occurrence} gives the bound. It limits the number of attachment points needed for each variable in Definition~\ref{def:bc-construction-short}. Lemma~\ref{lem:explicit-occurrence-oracles} gives reversible access to the occurrences, which tells us which surface pieces to attach and lets us query the boundary matrix. Theorem~\ref{thm:bc-da-to-b2-wrapper-explicit} uses both results.

\begin{lemma}[Dirichlet coefficient bounds]\label{lem:shifted-clock-row-bounds}
For \(G_C=5K_C\), set
\[
M_{\rm abs}^{\rm clock}:=24,\qquad
M_{\rm sum}^{\rm clock}:=20,\qquad P_\ast^{\rm clock}:=32.
\]
They bound
\[
\max_r\sum_j|(G_C)_{rj}|=\max_j\sum_r|(G_C)_{rj}|\le24,
\qquad \max_r\left|\sum_j(G_C)_{rj}\right|\le20.
\]
Every zero-sum row \(\sum_j(G_C)_{rj}y_j-r_rg_r=c_r\) has sum of positive coefficients at most \(P_\ast^{\rm clock}\), with \(r_r=\sum_j(G_C)_{rj}\).
\end{lemma}

\begin{lemma}[DA column coefficient bound]\label{lem:lg-occurrence}
Use the fixed expander degree \(d_E\) and lower bound \(\underline\sigma_E\), and set
\[
q_E:=\left\lceil\frac{4(M_{\rm sum}^{\rm clock})^2+4}{\underline\sigma_E^2}\right\rceil
=\left\lceil\frac{1604}{\underline\sigma_E^2}\right\rceil.
\]
The Dirichlet reduction uses this value of \(q_E\) as the repetition count for each oriented-edge incidence row.  Each internal column of the binary-tree gadget has sum of absolute coefficients at most \(b_{\rm int}:=3\).  Define
\[
\begin{gathered}
b_y:=24,\qquad b_g:=2P_\ast^{\rm clock}+2q_Ed_E=64+2q_Ed_E,
\qquad b_{\rm pad}:=2q_Ed_E,\\
b_{\rm aux}:=\max\{b_y,b_g,b_{\rm pad},b_{\rm int}\}=64+2q_Ed_E.
\end{gathered}
\]
Every column of a Dirichlet-generated DA matrix has sum of absolute
coefficients at most \(b_{\rm aux}\).
\end{lemma}

\begin{lemma}[Reversible access to occurrences]\label{lem:explicit-occurrence-oracles}
The resulting DA instance has the two inverse coefficient-occurrence oracles of Definition~\ref{def:dabo}. Each call uses \(O(1)\) queries to sparse access for \(G\) and the expander rotation map, plus reversible arithmetic polynomial in the input address width. The time per call is therefore \(\poly(\nu)\).
\end{lemma}

To bound the condition number, we must also consider vectors that satisfy the new equations only approximately. Each shared variable occurs in a bounded number of rows, so the combined map that fills in the copy and averaging variables has bounded norm, independent of matrix size. The residual comparison then relates the error in the new equations to two quantities: the error in the zero-sum equations and the distance of the added variables from their prescribed values. Lemma~\ref{lem:lg-spectral} uses these estimates to bound the smallest nonzero singular value and the condition number required in Theorem~\ref{thm:shifted-to-da}.

\begin{lemma}[Bound on the internal variables]\label{lem:global-W-bounded}
For original row \(r\), let \(t(r)\) denote the type of difference--average equations used to replace it. The associated constraint matrix, map filling in internal variables, and right-hand-side map are \(T_{t(r)},W_{t(r)},d_{t(r)}\), respectively. For \(u=(y,g)\), define \(W(u):=\bigoplus_rW_{t(r)}v_r\), where \(v_r\) lists the shared terminals of row \(r\). This map fills in the internal variables for all original rows. Then \(\norm{W(u)}\le C_W\norm u\) for all \(u\), with \(C_W\) depending only on the fixed family of original-row gadgets and the occurrence bounds.
\end{lemma}

\begin{lemma}[Residuals of the DA linear system]\label{lem:lg-norm-equivalence}
Let \(u=(y,g)\) contain all original and auxiliary coordinates of \(H\), and let \(w\) be the direct sum of internal gadget coordinates.  With \(W\) above, every produced instance satisfies, for all \((u,w)\),
\[
\begin{aligned}
c_0\bigl(\norm{Hu-h}^2+\norm{w-W(u)}^2\bigr)
&\le\norm{M_{\rm aux}(u,w)-d_{\rm aux}}^2\\
&\le C_0\bigl(\norm{Hu-h}^2+\norm{w-W(u)}^2\bigr).
\end{aligned}
\]
Here \(0<c_0\le C_0<\infty\) depend only on the finite set of row constructions and expander constants.  The estimate also holds homogeneously with \(h=d_{\rm aux}=0\).
\end{lemma}

\begin{remark}[Role of the residual comparison]\label{rem:norm-equivalence-central}
The estimate forces \(w=0\) when \(u=0\) and the homogeneous residual vanishes. Thus no nonzero kernel vector is supported only on the internal variables. The estimate also gives the singular-value bounds needed to recover the normalized minimum-norm solution. Preserving exact solutions alone does not give these bounds.
\end{remark}

\begin{lemma}[Condition number of the DA matrix]\label{lem:lg-spectral}
For input families with fixed coefficient patterns and sparsity bounds, assume \(0<\gamma_H\le\norm H\le\Gamma_H<\infty\) with constants independent of size. There is an explicit rational \(C_{\rm aux}\) that depends only on the constants comparing local residuals, the norms of the extension maps, the coefficient and sparsity bounds, the expander constants, and the lower and upper bounds \(\gamma_H,\Gamma_H\) on \(\norm H\). It satisfies
\[
\cond(M_{\rm aux})\le C_{\rm aux}\cond(H)
\le C_{\rm aux}C_H\cond(G).
\]
This applies to \(G_C=5K_C\) by Lemma~\ref{lem:shifted-H-norm-window}, not to arbitrary sparse systems without these assumptions.
\end{lemma}

Recovery from a normalized DA solution state needs an additional norm bound. Every DA solution gives the same original solution under the recovery map, including the minimum-norm solution. We also construct a DA solution whose norm is at most a constant times the norm of the original solution. The minimum-norm DA solution cannot have larger norm. This bound gives a constant probability of successful recovery. The lemmas below bound this probability and control the effect of error in the prepared state. Lemma~\ref{lem:herald-bounds} uses this probability bound when recovering the circuit output from a solution state of a Hodge Laplacian linear system.

\begin{lemma}[Recovery from the minimum-norm solution]\label{lem:lg-mp-decoder}
For every feasible instance produced by the reduction, define the unnormalized vectors \(z_{\min}:=M_{\rm aux}^+d_{\rm aux}\) and \(x_{\min}:=D_{\rm aux}z_{\min}=G^{-1}c\).  For instances from the Dirichlet propagation construction, applying a block encoding of \(D_{\rm aux}\) to
\[
\ket{z_\star}:=\frac{z_{\min}}{\norm{z_{\min}}}
\]
and measuring its success ancilla reports successful recovery with probability at least \(\underline p_{\rm aux}>0\), independently of circuit length.  Conditional on success, it outputs the normalized state proportional to \(G_C^{-1}c_C=K_C^{-1}b_C\).
\end{lemma}

\begin{lemma}[Recovery from an approximate state]\label{lem:lg-decoder-stability}
Let \(\ket{z_\star}\) be the normalized DA solution above, \(\ket{x_\star}:=D_{\rm aux}\ket{z_\star}/\norm{D_{\rm aux}\ket{z_\star}}\), and \(A_D:=D_{\rm aux}/\alphabe(D_{\rm aux})\) the contraction implemented by the block encoding.  Let \(p_0>0\) be a lower bound on the ideal probability of successful recovery as in Lemma~\ref{lem:lg-mp-decoder}.  For normalized \(\ket{\widetilde z}\) with
\(\dist_{\rm ph}(\ket{\widetilde z},\ket{z_\star})\le\eta\le\sqrt{p_0}/2\), the probability of successful recovery is at least \(p_0/4\). Conditioned on success, the output has phase-invariant distance from \(\ket{x_\star}\) at most \(2\eta/\sqrt{p_0}\).  Thus error is \(O(\eta)\) when \(p_0=\Omega(1)\).
This applies Lemma~\ref{lem:standard-normalization} to the contraction
$A_D$ and the success bound of Lemma~\ref{lem:lg-mp-decoder}.
\end{lemma}

\begin{lemma}[Right-hand-side state preparation]\label{lem:lg-rhs-state}
With common \(P_\ast\), source row \(r\) contributes only \(c_r/P_\ast\), on its final difference row. All copy, averaging, and expander constraints have zero right-hand-side entries. Hence \(\norm{d_{\rm aux}}=\norm c/P_\ast\). Mapping each \(r\) to the index of its final difference row preserves normalized amplitudes. It therefore converts preparation of \(c/\norm c\) to preparation of \(d_{\rm aux}/\norm{d_{\rm aux}}\).
\end{lemma}

\begin{lemma}[Explicit parameter bounds]\label{lem:certified-lg-bounds}
For the Dirichlet family, use the fixed bounds on coefficients,
sparsity, operator norms from above and below, padding, and expanders.
Also use the singular-value and norm bounds for the residual matrices
and extension maps of the finite family of row gadgets. Together these
determine computable rational constants \(C_H,C_{\rm aux},\underline p_{\rm aux}>0\),
independent of the circuit. If \(p_{\rm aux}\) is the ideal success
probability of recovery from the auxiliary linear system, then
\[
\cond(H)\le C_H\cond(G),\quad
\cond(M_{\rm aux})\le C_{\rm aux}\cond(H),\quad
p_{\rm aux}\ge\underline p_{\rm aux}.
\]
The resulting explicit rational polynomial is
\(\kappa_{\rm DA}^{\max}(L)=C_{\rm aux}C_H\kappa_{\rm clock}^{\max}(L)\).
These bounds are Lemmas~\ref{lem:H-condition}, \ref{lem:lg-spectral},
and~\ref{lem:lg-mp-decoder}; Remark~\ref{rem:lg-certificate-algorithm}
gives the computation of the rational constants.
\end{lemma}

The next two theorems collect the solution identity, condition-number bound, sparse access, and right-hand-side preparation for this reduction. Corollary~\ref{cor:da-hard-uncond} proves hardness by repeatedly preparing a solution state, recovering the propagation state, and measuring the circuit output. Each trial returns a random bit if recovery fails. The boundary linear system construction therefore starts from DA instances that satisfy the bounds in Definition~\ref{def:certified-parameters} and have a BQP-hard decision problem.

\begin{theorem}[Bounds for the DA reduction]\label{thm:lg-certificate}
For the finite family of DA row gadgets produced from Dirichlet propagation systems, with constants independent of circuit length and matrix dimensions,
\[
\begin{gathered}
\operatorname{eocc}(M_{\rm aux})=\max_i\sum_r|(M_{\rm aux})_{ri}|\le b_{\rm aux},\\
\cond(M_{\rm aux})\le C_{\rm aux}\cond(H)\le C_{\rm aux}C_H\cond(G_C),
\qquad D_{\rm aux}M_{\rm aux}^+d_{\rm aux}=G_C^{-1}c_C.
\end{gathered}
\]
Lemma~\ref{lem:lg-norm-equivalence} gives the two-sided residual estimate. Recovery from the ideal normalized solution state succeeds with constant probability at least \(\underline p_{\rm aux}\).  Lemma~\ref{lem:lg-rhs-state} prepares \(d_{\rm aux}/\norm{d_{\rm aux}}\) from \(c_C/\norm{c_C}\) by deterministically relabeling the basis states.
\end{theorem}

\begin{theorem}[From Dirichlet propagation to DA equations]\label{thm:shifted-to-da}
For every circuit \(C\) over \(\mathcal G_{\mathbb Q}\), the
construction transforms \((G_C,c_C)\) into \((M_C,d_C)\) with
\[
\begin{gathered}
\operatorname{eocc}(M_C)\le b_{\rm aux},\qquad
\cond(M_C)\le\kappa_{\rm DA}^{\max}(L),\\
D_{\rm aux}M_C^+d_C=G_C^{-1}c_C=K_C^{-1}b_C.
\end{gathered}
\]
The explicit rational polynomial \(\kappa_{\rm DA}^{\max}(L)\) is recorded in Definition~\ref{def:certified-parameters}.  By Lemma~\ref{lem:explicit-occurrence-oracles}, each target-oracle
query uses \(O(1)\) source row, column, or right-hand-side preparation queries and reversible arithmetic polynomial in the input address width. Right-hand-side preparation deterministically maps each source row to its final difference row, as in Lemma~\ref{lem:lg-rhs-state}.
\end{theorem}

\begin{corollary}[Bounded-occurrence DA hardness]\label{cor:da-hard-uncond}
For fixed \(b_{\rm aux}\), the sparse-access decision problem for \(\BoundedOccurrenceDA_{b_{\rm aux},\cond,\varepsilon}\), using the prescribed recovery and measurement, is \(\BQP\)-hard for \(\cond=\poly(L)\) and constant \(\varepsilon\). The reduction constructs the oracle circuits uniformly in polynomial time. The corresponding decision problem with circuit-encoded inputs is \(\BQP\)-hard under classical polynomial-time many-one reductions. Repeated preparation of a sufficiently accurate solution state, followed by this recovery and measurement, solves the decision problem. Thus the state-preparation task is \(\BQP\)-hard as well.
\end{corollary}

\subsection{Proofs of the DA residual and spectral bounds}\label{app:auxiliary-variable-proofs}

\subsubsection{Local residual bounds}\label{app:lem-template-stability}
\begin{proof}[Proof of Lemma~\ref{lem:template-stability}]
We first specify the constants used in the parameter bounds.  The set of coordinates for the shared variables is
\[
V_T=\{u_0,u_1,\ldots,u_p\},
\]
with shared terminals listed in the ordered vector $v$. The set of internal coordinates is
\[
\mathcal W_T=\mathcal C_T^+\cup \mathcal C_T^-\cup \mathcal A_T^+\cup \mathcal A_T^-\cup\{r_+,r_-\},
\]
where $\mathcal C_T^\pm$ are copy leaves and $\mathcal A_T^\pm$ are internal averaging vertices other than the roots. Each root is listed separately. When $P_\ast=1$, its leaf and root labels denote the same coordinate, which is counted once. Coordinates explicitly identified with shared terminals are omitted. The vector $w$ lists the coordinates in $\mathcal W_T$ in order. Group the constraint residuals into blocks as
\[
L_T=(L_T^{\rm copy,+},L_T^{\rm copy,-},L_T^{\rm avg,+},L_T^{\rm avg,-},L_T^{\rm final}),
\]
corresponding respectively to copy constraints, averaging-tree constraints, and the final difference row $r_+-r_- = \beta/P_\ast$.

Fix the shared terminals $v$. The copy rows force every leaf to equal its assigned terminal. The averaging equations then determine each parent from its two children, up to the roots. Thus a unique assignment $W_Tv$ makes all copy and averaging residuals vanish. To obtain it, first assign each leaf its terminal value, then evaluate both binary trees from leaves to roots. With $w=W_Tv$, each root is the average of its side's terminal values, repeated according to the coefficient multiplicities. Both sides include the same additional copies of $u_0$. Therefore
\[
r_+-r_- = \frac{a^\top v}{P_\ast},
\]
and the only remaining residual is $(a^\top v-\beta)/P_\ast$. At $w=W_Tv$, the squared residual norm is $|a^\top v-\beta|^2/P_\ast^2$. This bounds the minimum over $w$ from above. We prove the lower bound next. The map $W_T$ is linear and bounded because each row construction has constant dimension.

For arbitrary $w$ and $\beta$, put
\[
e:=a^\top v-\beta,
\qquad z:=w-W_Tv.
\]
Let \(n_T\) and \(m_T\) be the numbers of internal coordinates and constraints, respectively, and let \(e_{\rm fin}\in\mathbb R^{m_T}\) be the unit vector of the final difference-constraint row. Define
\[
R_T:\mathbb R\oplus\mathbb R^{n_T}\to\mathbb R^{m_T},
\qquad R_T(e,z):=T_T(0,z)+\frac{e}{P_\ast}e_{\rm fin}.
\]
By the preceding computation and linearity,
\[
T_T(v,w)-d_T(\beta)=R_T(e,z).
\]
This identity holds for every right-hand side $\beta$. The map $R_T$ is injective. Indeed, if $R_T(e,z)=0$, then the copy and averaging constraints with zero shared variables force $z=0$. The final difference constraint then forces $e=0$. Consequently
\[
\sigma_{\min}(R_T)^2\left(|e|^2+\norm z^2\right)
\le
\norm{R_T(e,z)}^2
\le
\norm{R_T}^2\left(|e|^2+\norm z^2\right),
\qquad \sigma_{\min}(R_T)>0.
\]
The row constructions belong to a fixed finite family independent of the instance size.  Choose a positive rational lower bound on $\min_T\sigma_{\min}(R_T)^2$ and a rational upper bound on $\max_T\{\norm{R_T}^2,\norm{W_T}\}$.  These choices give the uniform constants $c_T,C_T$ in the statement, including the bound on the extension map $W_T$.
\end{proof}

\begin{remark}[Computing the constants]\label{rem:lg-certificate-algorithm}
Enumerate the finite Dirichlet coefficient patterns, including $Z$-gate rows, and the expander and padding types. For each DA row gadget, copy the terminal values to its leaves and compute the averages to obtain the rational extension map \(W_T\). Also form the residual map $R_T$ defined above, which the proof shows is injective.

Compute the associated rational Gram matrices. Use Sturm sequences to
isolate the roots of their characteristic polynomials, then round the
interval endpoints outward to rational numbers. This gives a positive rational lower bound on
$\sigma_{\min}(R_T)^2$ and rational upper bounds on $\norm{R_T}^2$
and $\norm{W_T}$. Take extrema over the finite family to choose uniform
$c_T,C_T$. Combining these with the coefficient, sparsity, padding,
expander, and operator-norm bounds in Lemmas~\ref{lem:H-condition},
\ref{lem:shifted-H-norm-window}, \ref{lem:global-W-bounded},
\ref{lem:lg-norm-equivalence}, and~\ref{lem:lg-spectral} computes
rational constants $C_H,C_{\rm aux}$. For the recovery bound in
Lemma~\ref{lem:lg-mp-decoder}, use conservative rational upper bounds
on the block-encoding normalization of $D_{\rm aux}$ and the
constant bounding the extension map's norm, computed from the same finite matrices and
occurrence bounds. This gives a positive rational
$\underline p_{\rm aux}$. All these constants are independent of
circuit length.
\end{remark}

\subsubsection{Global residual bounds}\label{app:lem-lg-norm-equivalence}
\begin{proof}[Proof of Lemma~\ref{lem:lg-norm-equivalence}]
For original row $r$, the residual in $H$ is $a_r^\top v_r-c_r$. Here $a_r$ is the bounded integer coefficient vector whose entries sum to zero and whose equation is replaced by the difference--average gadget.  Lemma~\ref{lem:template-stability} gives
\[
\begin{aligned}
&c_r'\left(|a_r^\top v_r-c_r|^2+\norm{w_r-W_{t(r)}v_r}^2\right)
 \le
 \norm{T_{t(r)}(v_r,w_r)-d_{t(r)}(c_r)}^2 \\
&\hspace{2cm}\le
 C_r'\left(|a_r^\top v_r-c_r|^2+\norm{w_r-W_{t(r)}v_r}^2\right).
\end{aligned}
\]
The row gadgets belong to a fixed finite family independent of the instance size. Taking the minimum and maximum over this family gives uniform constants. Sum the bounds over all original rows. The repeated expander constraints are difference constraints in both $H$ and $M_{\rm aux}$, with zero right-hand side. They therefore contribute exactly $\norm{B_E^{(q_E)}g}^2$ to both sides. Adding these contributions proves the displayed inequality.
\end{proof}

\subsubsection{Singular-value bounds}\label{app:lem-lg-spectral}
\begin{proof}[Proof of Lemma~\ref{lem:lg-spectral}]
It suffices to work with homogeneous linear systems.  Lemma~\ref{lem:lg-norm-equivalence} gives
\[
 c_0\left(\norm{Hu}^2+\norm{w-Wu}^2\right)
 \le
 \norm{M_{\rm aux}(u,w)}^2
 \le
 C_0\left(\norm{Hu}^2+\norm{w-Wu}^2\right).
\]
The homogeneous kernel is
\[
\ker M_{\rm aux}=\{(u,Wu):u\in\ker H\}.
\]
Indeed, zero residuals in the copy and averaging constraints force $w=Wu$. The remaining constraints on the shared variables are exactly $Hu=0$. Conversely, these two conditions make every residual vanish.

By Lemma~\ref{lem:global-W-bounded}, $W$ is a linear map whose operator norm is bounded independently of the instance size.  Therefore there are constants $0<a_1\le a_2<\infty$ such that
\[
 a_1\left(\dist(u,\ker H)^2+\norm{w-Wu}^2\right)
 \le
 \dist((u,w),\ker M_{\rm aux})^2
 \le
 a_2\left(\dist(u,\ker H)^2+\norm{w-Wu}^2\right).
\]
Using $\norm{Hu}\ge\sigma_{\min}^+(H)\dist(u,\ker H)$ in the lower residual bound gives
\[
\norm{M_{\rm aux}(u,w)}^2
\ge
c_2\min\{(\sigma_{\min}^+(H))^2,1\}\,
\dist((u,w),\ker M_{\rm aux})^2.
\]
Thus
\[
\sigma_{\min}^+(M_{\rm aux})
\ge
c_3\min\{\sigma_{\min}^+(H),1\}.
\]
The upper bound in Lemma~\ref{lem:lg-norm-equivalence} and the bound on $\|W\|$ give
\[
\norm{M_{\rm aux}}
\le
C_3\max\{\norm H,1\}.
\]
By the operator-norm hypothesis, $\norm H$ is bounded above by $\Gamma_H$ and below by $\gamma_H>0$.  Hence
\[
\frac{\max\{\norm H,1\}}{\min\{\sigma_{\min}^+(H),1\}}
\le
C_4(\gamma_H,\Gamma_H)\frac{\norm H}{\sigma_{\min}^+(H)}
=C_4(\gamma_H,\Gamma_H)\cond(H).
\]
This proves $\cond(M_{\rm aux})\le C_{\rm aux}\cond(H)$ using explicit rational bounds for the constants comparing local residuals, the norms of the extension maps, and \(\gamma_H,\Gamma_H\).  Lemma~\ref{lem:H-condition} gives the final explicit inequality.
\end{proof}

\subsection{Proofs for the expander construction and recovery}
\begin{proof}[Proof of Theorem~\ref{thm:certified-expander-family}]
Fix a constant degree \(d_0\ge3\) for Alon's strongly explicit construction. For every sufficiently large \(n\), it gives a \(d_0\)-regular graph on exactly \(n\) vertices. The nontrivial adjacency eigenvalues have absolute value at most a constant \(\lambda_0<d_0\)~\cite[Theorem~1.2]{AlonEverySize}. Apply this construction with \(n=2^{r-1}\). Strong explicitness means that we can compute the \(\ell\)-th neighbor in \(\poly(r)\) time. Adjacency lists count edge multiplicities. A base self-loop contributes one to the adjacency row sum, as in Alon's convention. For distinct endpoints $v,w$, pair the $j$-th port from $v$ to $w$ with the $j$-th port from $w$ to $v$ in the local orderings. Pair a self-loop port with itself. These rules give a rotation map that is its own inverse. Since the degree is constant, scanning the local lists takes time polynomial in $r$.  Put \(c_0:=d_0-\lambda_0>0\), and let \(A_r\) be its adjacency matrix.  Form a graph on \([2^{r-1}]\times\{0,1\}\) by replacing each base port \((v,\ell)\), whose reverse endpoint is \((w,\ell')\), with the cross-layer edge
\[
(v,b,\ell)\longleftrightarrow(w,1-b,\ell'),
\]
and add the perfect-matching port
\[
(v,b,d_0)\longleftrightarrow(v,1-b,d_0).
\]
This defines a rotation map of degree \(d_E=d_0+1\) on \(2^r\) vertices, and the map is its own inverse.  Every edge changes the layer bit, so the graph has no self-loops.  Port labels distinguish all directed incidences.

Ordering the vertices by layer gives the adjacency matrix
\[
\widetilde A_r=
\begin{pmatrix}0&A_r+I\\ A_r+I&0\end{pmatrix}.
\]
If \(A_rx=\lambda x\), the symmetric and antisymmetric layer vectors have Laplacian eigenvalues
\[
d_0-\lambda
\qquad\text{and}\qquad
d_0+2+\lambda,
\]
respectively.  The first family contains one zero eigenvalue and is otherwise at least \(c_0\); the second family is at least \(2\), because \(\lambda\ge-d_0\).  Thus, for all sufficiently large $r$, the new graph is connected and has Laplacian gap at least $\min\{c_0,2\}$.

Fix $r_0\ge2$ such that this construction applies for every $r\ge r_0$.  For each $1\le r<r_0$, write $n=2^r$ and use the cycle on $\mathbb Z/n\mathbb Z$, together with $d_E-2$ copies of the perfect matching $v\leftrightarrow v\mathbin{\mathrm{xor}}1$.  For $n=2$, the cycle consists of two parallel edges.  Parallel edges have distinct port labels.  The two cycle ports send $v$ to $v+1$ and $v-1$ modulo $n$. They are each other's reverse ports, and each matching port is its own reverse port. This gives an explicit rotation map that is its own inverse for a connected, loopless, $d_E$-regular multigraph.  Adding the matching Laplacians cannot decrease the spectral gap, so
\[
\lambda_2(L(E_n))\ge4\sin^2(\pi/n)\ge16/n^2.
\]
Put $n_{\rm small}=2^{r_0-1}\ge2$, and choose a positive rational constant
\[
c_E^{\rm gap}\le\min\{c_0,2,16/n_{\rm small}^2\}.
\]
This lower bound applies to every $r\ge1$, including all of the explicitly specified small graphs.

Define \((B_Eg)_{(v,\ell)}=g_v-g_w\).  Every undirected edge appears through its two opposite directed ports, so
\[
B_E^\top B_E=2L(E_{2^r}).
\]
Consequently, for \(\delta\perp\mathbf1\),
\[
\norm{B_E\delta}^2\ge2c_E^{\rm gap}\norm\delta^2,
\qquad
\norm{B_E}^2\le4d_E.
\]
Choose rational constants \(0<\underline\sigma_E\le\sqrt{2c_E^{\rm gap}}\) and \(\overline\sigma_E\ge2\sqrt{d_E}\).  For a row indexed by \((v,\ell)\), row access returns \((v,+1),(w,-1)\), where \(\mathsf{Rot}_E(r,v,\ell)=(w,\ell')\). Column access enumerates the \(d_E\) outgoing and \(d_E\) reverse directed-port rows using the rotation map. These procedures take \(\poly(r)\) time. The finitely many small orders use the fixed graphs above with the same degree and uniform bounds.  Finally, for an arbitrary requested number \(N\) of auxiliary variables, take \(N_E=2^{\max\{1,\lceil\log_2N\rceil\}}\), so \(N\le N_E\le2N\).
\end{proof}

\begin{proof}[Proof of Lemma~\ref{lem:H-exact}]
Because the expander is connected, its equality constraints force $g=t\mathbf1_{N_E}$. Since $R_G\mathbf1_N=G\mathbf1_N$, the original rows become $G(y-t\mathbf1_N)=c$. Invertibility gives $y-t\mathbf1_N=G^{-1}c$. Reversing this calculation proves the converse. Setting $c=0$ gives the kernel. Thus the padding variables add no independent kernel direction and do not change recovery of $G^{-1}c$.
\end{proof}

\begin{proof}[Proof of Lemma~\ref{lem:H-condition}]
Write
\[
g=\bar g\mathbf1_{N_E}+\delta,
\qquad
\delta\perp\mathbf1_{N_E},
\qquad
x:=y-\bar g\mathbf1_N.
\]
Then
\[
H(y,g)=
\begin{pmatrix}
Gx-R_G P_{\rm orig}\delta\\
B_E^{(q_E)}\delta
\end{pmatrix}.
\]
Here \(\delta\) includes all \(N_E\) auxiliary coordinates. Only \(P_{\rm orig}\delta\) enters the residual \(Gy-R_GP_{\rm orig}g-c\) of the original rows. We first bound the quotient norm of \((y,g)\), meaning its $\ell_2$-distance to \(\ker(H)=\Span\{(\mathbf1_N,\mathbf1_{N_E})\}\), in terms of \(x\) and \(\delta\). We have
\[
\begin{aligned}
\dist\bigl((y,g),\ker H\bigr)^2
&=\min_{t\in\mathbb R}\{\norm{y-t\mathbf1_N}^2+\norm{g-t\mathbf1_{N_E}}^2\} \\
&=\norm\delta^2+\min_{s\in\mathbb R}\{\norm{x-s\mathbf1_N}^2+N_Es^2\},
\end{aligned}
\]
where \(s=t-\bar g\) and \(\delta\perp\mathbf1_{N_E}\).  The choice \(s=0\) gives the upper bound
\[
\dist\bigl((y,g),\ker H\bigr)^2\le \norm x^2+\norm\delta^2.
\]
Conversely, for every \(s\),
\[
\norm x\le \norm{x-s\mathbf1_N}+|s|\sqrt N
\le \norm{x-s\mathbf1_N}+|s|\sqrt{N_E},
\]
so \(\norm x^2\le2\norm{x-s\mathbf1_N}^2+2N_Es^2\).  Taking the minimum over \(s\) gives
\[
\dist\bigl((y,g),\ker H\bigr)^2
\ge \frac12\norm x^2+\norm\delta^2.
\]
Thus the quotient norm is uniformly equivalent to \((\norm x^2+\norm\delta^2)^{1/2}\), with absolute constants under the padding convention \(N\le N_E\le2N\).

For the lower bound, use
\[
\norm{Gx-R_G P_{\rm orig}\delta}^2
\ge \frac12\norm{Gx}^2-\norm{R_G P_{\rm orig}\delta}^2,
\]
and the repeated expander block gives
\[
\norm{B_E^{(q_E)}\delta}^2\ge q_E\underline\sigma_E^2\norm\delta^2.
\]
Since \(\norm{R_G P_{\rm orig}\delta}\le \norm{R_G P_{\rm orig}}\norm\delta\),
\[
\norm{H(y,g)}^2
\ge
\frac12\norm{Gx}^2+
\bigl(q_E\underline\sigma_E^2-\norm{R_G P_{\rm orig}}^2\bigr)\norm\delta^2.
\]
By the definition of \(q_E\),
\[
q_E\underline\sigma_E^2-\norm{R_G P_{\rm orig}}^2\ge1,
\]
and hence
\[
\norm{H(y,g)}^2
\ge
\frac12(\sigma_{\min}(G))^2\norm x^2+
\norm\delta^2.
\]
Consequently
\[
\sigma_{\min}^+(H)
\ge c\min\{\sigma_{\min}(G),1\}
\]
for a constant \(c>0\), where the positive singular values are those of the restriction of $H$ to $(\ker H)^\perp$.

For the upper bound, use
\[
\norm{Gx-R_GP_{\rm orig}\delta}
\le \norm G\,\norm x+\norm{R_GP_{\rm orig}}\,\norm\delta,
\qquad
\norm{B_E^{(q_E)}\delta}
\le \sqrt{q_E}\,\overline\sigma_E\,\norm\delta .
\]
Hence
\[
\norm{H(y,g)}^2
\le C\bigl(\norm G^2\norm x^2+\norm\delta^2\bigr)
\]
for a constant \(C\) depending only on \(\max_r|r_r|\), \(q_E\), and \(\overline\sigma_E\).  To obtain an operator-norm bound on the original variables, we bound $x$ and $\delta$ by those variables:
\[
\norm x=\norm{y-\bar g\mathbf1_N}
\le \norm y+|\bar g|\sqrt N
\le \norm y+\sqrt{N/N_E}\,\norm g
\le \norm y+\norm g,
\qquad
\norm\delta\le\norm g.
\]
Therefore \((\norm x^2+\norm\delta^2)^{1/2}\le C_E\norm{(y,g)}\) with a constant depending only on the padding ratio, and consequently
\[
\norm H\le C'\max\{\norm G,1\}.
\]
For the family of integer propagation matrices with a fixed finite coefficient set, \(\norm G\ge1\) and \(\norm G\le\Gamma\).  If \(\sigma_{\min}(G)\le1\), then
\[
\cond(H)\le C''\frac{\norm G}{\sigma_{\min}(G)}=C''\cond(G).
\]
If \(\sigma_{\min}(G)>1\), then \(\cond(G)\ge1\), while \(\cond(H)\le C''\Gamma\). Increasing the explicit rational constant again gives \(\cond(H)\le C_H\cond(G)\). The explicit rational bounds listed in Lemma~\ref{lem:certified-lg-bounds} bound every constant used in the proof from above or below as needed. This proves the claim for the Dirichlet propagation matrices.
\end{proof}

\begin{proof}[Proof of Lemma~\ref{lem:shifted-H-norm-window}]
The upper bound follows from constant sparsity, fixed propagation coefficients, bounded entries of $R_G$, and constant $q_E,d_E$. Taking $g=0$ gives $\norm{H_C}\ge\norm{G_C}=5\lambda_{\max}(D_m)\ge10$, since every diagonal Rayleigh quotient of $D_m$ equals $2$. Thus $\gamma_H=10$ is valid.
\end{proof}

\begin{proof}[Proof of Lemma~\ref{lem:shifted-clock-row-bounds}]
Every row of \(G_C\) has diagonal entry \(10\).  Each row has at most two neighboring clock blocks.  A neighboring \(X\), \(Z\), Toffoli, or identity block contributes absolute row sum \(5\).  A neighboring \(R\)-block contributes a row of \(-5R\), hence absolute row sum \(7\).  Therefore the absolute row sum is at most
\[
10+7+7=24.
\]
The matrix is symmetric, so the same bound holds for absolute column sum.

The signed row sums of a neighboring \(-5R\) block are \(-7\) and \(-1\).  A neighboring permutation block contributes row sum \(-5\), while the signed-permutation block \(-5Z\) contributes either \(-5\) or \(+5\).  Hence an endpoint row has signed row sum of absolute value at most
\[
\max\{|10-7|,\ |10-1|,\ |10-5|,\ |10+5|\}=15,
\]
and an interior row has signed row sum of absolute value at most
\[
10+5+5=20.
\]
Thus \(M_{\rm sum}^{\rm clock}=20\) is valid.  This larger value is needed only because the chosen gate set includes the exact rational gate \(Z\).

For the auxiliary-variable row, the coefficient vector is
\[
((G_C)_{r1},\ldots,(G_C)_{rN},-r_r).
\]
The largest possible sum of the positive input coefficients is at most \(10+5+5=20\).  If the auxiliary variable has a positive coefficient, then the signed row sum of \(G_C\) is negative. The most negative row sum occurs for two \(-5R\) rows of sum \(-7\). The auxiliary variable then contributes at most \(4\), and the only positive input contribution is the diagonal term.  Therefore the sum of all positive coefficients is at most \(20\), and the uniform power-of-two choice \(P_\ast^{\rm clock}=32\) suffices.
\end{proof}

\begin{proof}[Proof of Lemma~\ref{lem:lg-occurrence}]
Every original variable and every auxiliary variable occurs in a constraint. This includes the variables added for individual rows, the padding variables, and the internal gadget variables. Thus the constructed DA matrix has no zero columns.

Occurrences are counted by absolute coefficient: an averaging coefficient $\pm2$ counts twice.
An original variable \(y_i\) appears in original-row gadgets with a total number of leaf-copy occurrences at most the sum of absolute values in column $i$ of \(G_C\), hence at most \(b_y=24\).

An auxiliary variable \(g_r\) for an original row appears in the gadget for row \(r\). This gadget has \(P_\ast^{\rm clock}\) positive leaves and \(P_\ast^{\rm clock}\) negative leaves. Thus each shared terminal appears in at most \(2P_\ast^{\rm clock}=64\) difference rows that copy values to leaves. The auxiliary variable also appears in at most \(2q_Ed_E\) repeated expander difference rows indexed by directed ports. Its column therefore has sum of absolute coefficients at most \(64+2q_Ed_E\).

A padding variable has no original-row gadget and appears only in repeated expander edges. The sum of absolute coefficients in its column is therefore at most \(2q_Ed_E\). For internal gadget variables, a leaf appears in one equality row and one averaging row. An internal averaging node appears with coefficient \(-2\) in its defining row and once in the next row. A root appears with coefficient \(-2\) in its defining averaging row and once in the final difference row. Hence \(b_{\rm int}=3\) is valid.

Taking the maximum gives
\[
\operatorname{eocc}(M_{\rm aux})\le b_{\rm aux}.
\]
\end{proof}

\begin{proof}[Proof of Lemma~\ref{lem:explicit-occurrence-oracles}]
For each gadget replacing an original equation, label its rows by the source row, the type of row construction, and the local row. Label expander rows by repetition and directed port. Fix an ordering of rows and coefficient occurrences within each DA row gadget. Collect repeated variable coefficients, expand \(\pm2\) into two same-sign occurrences, and order the at most four slots \(a\in[4]\) by this ordering.
For each variable, list the row and slot of each occurrence as a pair \((q,a)\), then sort these pairs lexicographically.  For \(y_j\), enumerate source column entries and their copy equations through the tables of the row constructions.  For \(g_r\), combine the gadget for original row \(r\) with the \(2q_Ed_E\) directed-port rows; padding variables have only the latter.  Internal variables use a fixed local table.  Each list has length at most \(b_{\rm aux}\) by Lemma~\ref{lem:lg-occurrence}.
The map \((q,a)\mapsto(i,\sigma,t)\) returns the variable, sign, and zero-based rank in this list. The inverse map \((i,t)\mapsto(q,a,\sigma)\) returns its \(t\)-th entry.  Unused positions return the dummy symbol.  Both maps use the same deterministic lists, so they are inverse on active occurrences. The lists have constant size, and their sorting and index arithmetic can be made reversible. Query counts depend only on the fixed parameters of the row constructions and expander, including repetitions.
\end{proof}

\begin{proof}[Proof of Lemma~\ref{lem:global-W-bounded}]
The bounds on the local extension maps give $\norm{Wu}^2\le C_T^2\sum_r\norm{v_r}^2$. By column sparsity, each $y_i$ belongs to at most a constant number of row gadgets. Each original $g_r$ belongs only to its own row gadget, and padding variables belong to none. Hence $\sum_r\norm{v_r}^2\le C_{\rm occ}\norm u^2$ and $\norm W\le C_T\sqrt{C_{\rm occ}}$.
\end{proof}

\begin{proof}[Proof of Lemma~\ref{lem:lg-mp-decoder}]
Every exact DA solution satisfies all internal row gadgets for the original constraints and all expander rows exactly.  By the residual bound in Lemma~\ref{lem:template-stability}, the shared variables therefore satisfy $Hu=h$.  Lemma~\ref{lem:H-exact} implies that $D_{\rm aux}z=G^{-1}c$ for every exact DA solution $z$.  Since the instance produced by the reduction is feasible, $d_{\rm aux}\in\im(M_{\rm aux})$, and $M_{\rm aux}^+d_{\rm aux}$ is the minimum-norm exact solution.  The same recovery identity applies to it.

Let
\[
\begin{gathered}
z_{\min}:=M_{\rm aux}^+d_{\rm aux},\qquad
x_{\min}:=D_{\rm aux}z_{\min}=G_C^{-1}c_C,\\
\ket{z_\star}:=\frac{z_{\min}}{\norm{z_{\min}}},\qquad
\ket{x_\star}:=\frac{x_{\min}}{\norm{x_{\min}}}.
\end{gathered}
\]
Take the exact solution with auxiliary variables $g=0$, shared variables $y=x_{\min}$, and internal gadget variables $w=W(y,g)$. Its norm is at most $C_{\rm ext}\norm{x_{\min}}$ by Lemma~\ref{lem:global-W-bounded}: the vector of shared variables with auxiliary values $0$ has norm $\norm{x_{\min}}$. The Moore--Penrose solution has minimum norm, so
\[
\norm{z_{\min}}\le C_{\rm ext}\norm{x_{\min}}.
\]
The recovery map $D_{\rm aux}$ has two nonzero entries per row and at most one per column, each equal to $+1$ or $-1$. With row and column sparsity bounds $2$ and $1$, its sparse block encoding has normalization $\alphabe(D_{\rm aux})=\sqrt2$.  Applying this block encoding to the normalized state $\ket{z_\star}$ and measuring its success ancillas gives ideal success probability
\[
p_D
=
\frac{\norm{D_{\rm aux}z_{\min}}^2}{\alphabe(D_{\rm aux})^2\norm{z_{\min}}^2}
=
\frac{\norm{x_{\min}}^2}{\alphabe(D_{\rm aux})^2\norm{z_{\min}}^2}
\ge
\frac{1}{\alphabe(D_{\rm aux})^2C_{\rm ext}^2}
=\Omega(1).
\]
Because $D_{\rm aux}$ has a fixed local sparsity pattern and the linear extension maps \(W_T\) belong to a fixed finite family, the constants $\alphabe(D_{\rm aux})$ and $C_{\rm ext}$ have fixed rational upper bounds.  The oracle reduction therefore fixes a concrete rational lower bound
\[
\underline p_{\rm aux}\le \frac{1}{\alphabe(D_{\rm aux})^2C_{\rm ext}^2},
\qquad
\underline p_{\rm aux}>0,
\]
and uses this computable lower bound to define the final promise gap.  Conditioned on success, the output state is exactly $\ket{x_\star}$.
\end{proof}

\begin{proof}[Proof of Lemma~\ref{lem:lg-rhs-state}]
Only the final difference row of gadget $r$ has right-hand side $c_r/P_\ast$; all copy, average, padding, and expander rows have zero right-hand side. The map from source row labels to final difference-row labels is injective and reversible, and it preserves normalized amplitudes. Compute the final-row label, use it to recover the original row index and erase the source register, then uncompute the workspace. This prepares the reduced right-hand side with all workspace registers cleared.
\end{proof}

\begin{proof}[Proof of Theorem~\ref{thm:lg-certificate}]
The occurrence and residual bounds follow from Lemmas~\ref{lem:lg-occurrence} and~\ref{lem:lg-norm-equivalence}. The condition-number bound follows from Lemma~\ref{lem:lg-spectral} with~\ref{lem:shifted-H-norm-window}. Lemma~\ref{lem:lg-mp-decoder} gives exact recovery and its success probability. Lemma~\ref{lem:lg-rhs-state} gives right-hand-side preparation.
\end{proof}

\begin{proof}[Proof of Theorem~\ref{thm:shifted-to-da}]
The integer matrix $G_C$ has entries in a fixed finite set and constant row and column sparsity. Its operator norm is bounded above and below by positive constants independent of $L$. It therefore meets Definition~\ref{def:auxiliary-variable-H}. Apply Definition~\ref{def:lg-da-compiler} and Theorem~\ref{thm:lg-certificate}. Together with the propagation condition-number bound, they give the rational inequality $\cond(M_C)\le\kappa_{\rm DA}^{\max}(L)$, exact recovery of the minimum-norm solution, and a constant probability of successful recovery.

Each type of row gadget has a constant-size lookup table. Row and column queries use these tables and the $G_C$ sparse oracle. Lemma~\ref{lem:explicit-occurrence-oracles} implements the mutually inverse coefficient-occurrence oracles with constant query overhead. Expander rows are explicit. Lemma~\ref{lem:lg-rhs-state} prepares the reduced right-hand side by relabeling final difference rows. Here $c_C=5\ket1\ket{0^{n_q}}$ is one-sparse, so its normalized state and the reduced state are prepared by deterministic reversible basis-state relabeling.
\end{proof}

\begin{proof}[Proof of Corollary~\ref{cor:da-hard-uncond}]
Compose Theorems~\ref{thm:shifted-clock-hard} and~\ref{thm:shifted-to-da}. Apply the $D_{\rm aux}$ block encoding to the DA solution state and measure its success ancillas. On success, measure the clock and, if its value lies in the final clock interval, the output qubit. Return an independently sampled fair random bit when the clock lies outside this interval or recovery fails. The decision circuit retains every outcome, including recovery failures.

The ideal recovery probability is $p_D\ge p_0=\Omega(1)$, and the successful state is the weighted history state. Its measurement of the clock and output qubit has bias $\delta\ge1/48$, so the complete circuit has bias at least $p_0\delta$. For an $\eta$-approximate DA solution, Lemma~\ref{lem:lg-decoder-stability} gives a constant probability of successful recovery and $O(\eta)$ error conditioned on success under its stated threshold. More directly, Lemma~\ref{lem:state-to-verifier-final} bounds the unconditional acceptance-probability error by $\eta$. Choosing $\eta\le p_0\delta/2$ leaves bias at least $p_0\delta/2$, which amplification turns into bounded-error acceptance probabilities.
\end{proof}
\section{From difference--average equations to boundary linear systems}\label{sec:da-to-b2}

Ding, Kyng, Probst Gutenberg, and Zhang use difference--average equations and triangulated surfaces in their classical reduction \cite{DingFull}. We choose the number of copies of each surface so that encoding a variable preserves its norm. We then prove that the encoding preserves the minimum-norm solution, bound the nonzero singular values, and give reversible sparse access to the unweighted boundary matrix.

Because the pseudoinverse selects the solution of minimum norm, preserving the set of solutions alone does not suffice. We must also show that the target minimum-norm solution encodes the source minimum-norm solution. Recovery must use only the target input and its solution state. Finally, the nonzero singular values of the rescaled boundary matrix need an inverse-polynomial lower bound. We establish these properties below.

\subsection{The surfaces and their constants}

Each gadget is a triangulated surface used to represent a variable. Its boundary cycles are separated so that joining the surfaces keeps each triangle's three vertices distinct, even when two occurrences of a variable belong to the same equation. In any solution of the boundary linear system, the equations on interior edges force the triangle coefficients on each surface to be equal. The orientations give each occurrence its required sign. Lemma~\ref{lem:bc-port-control} uses the spectral bounds for each surface to control the effect of coefficients that are not constant.

\begin{definition}[Triangulated surfaces with boundary ports]\label{def:bc-templates}
Fix a bound $b_\ast$ on the sum of absolute coefficients in each column of the difference--average matrix, and let $n_u:=\max\{3,u\}$ for $1\le u\le b_\ast$.  Let $K^{(u)}$ be the triangulated surface obtained by the explicit bipyramid refinement in the proof of Lemma~\ref{lem:bc-template-existence}, and let $t_u$ be its number of triangles.  Its $u$ designated triangular boundary cycles are called ports, with edges $\mathsf p_{u,a,r}$, $1\le a\le u$, $1\le r\le3$.

Let $R_u$ and $P_u$ be the submatrices of the oriented boundary matrix whose rows are indexed by interior edges outside the ports and by the $3u$ port edges, respectively. Orient the triangles coherently, so their signs on each shared interior edge are opposite. Give each port cycle the induced boundary orientation. Then
\[
R_u\mathbf1_{t_u}=0,\qquad
P_u\mathbf1_{t_u}=\sum_{a=1}^{u}\sum_{r=1}^{3}e_{a,r}.
\]
Here $e_{a,r}$ denotes the coordinate vector of port edge $(a,r)$. Attach a port with its orientation preserved for a positive occurrence and reversed for a negative occurrence.
\end{definition}

\paragraph{Surface parameters.}
For the surfaces in Definition~\ref{def:bc-templates}, use the following notation and constants:
\[
\begin{gathered}
t_u=50n_u-u,\qquad T_\ast:=\max_{1\le u\le b_\ast}t_u=\max\{149,49b_\ast\},\\
W_u:=\{g\in\mathbb R^{t_u}:\mathbf1^\top g=0\},\qquad
\mu_u^{\rm true}:=\min_{g\in W_u,\,\norm g=1}\norm{R_ug}^2,\\
\Gamma_u^{\rm true}:=\norm{P_u(R_u|_{W_u})^+}^2,\\
\underline\mu_\ast:=T_\ast^{-2},\quad
\overline\Gamma_\ast:=T_\ast^2,\quad
L_\ast:=4T_\ast,\quad
\Theta_\ast:=L_\ast\overline\Gamma_\ast=4T_\ast^3,\\
c_\ast:=\frac{3}{4(1+4T_\ast^3)},\qquad
\Lambda_\ast:=12T_\ast,\qquad
C_{b_\ast}^{\rm bc}:=9T_\ast^2.
\end{gathered}
\]
All these constants depend only on $b_\ast$.

\begin{remark}[Uniform bounds for the surfaces]\label{rem:fixed-gadget-obligation}
The bounds below are analytic and require no computer-generated certificate.  In the hardness reduction, set $b_\ast=b_{\rm aux}$ and use the constants listed after Definition~\ref{def:bc-templates}.
\end{remark}

\begin{lemma}[Geometry and spectra of the surfaces]\label{lem:bc-template-existence}
Each $K^{(u)}$ is a compact connected orientable triangulated surface with $u$ collared triangular boundary ports, and
\[
t_u\le T_\ast,\qquad \ker R_u=\Span\{\mathbf1\},\qquad
\mu_u^{\rm true}\ge T_\ast^{-2},\qquad \Gamma_u^{\rm true}\le T_\ast^2.
\]
The prescribed port identifications, including repeated ports for a coefficient $\pm2$, produce an abstract simplicial complex.
\end{lemma}

\begin{lemma}[Joining ports gives a simplicial complex]\label{lem:bc-simplicial-legality}
Joining the prescribed ports keeps distinct triangle labels distinct, does not repeat vertices within a triangle, and identifies only designated port edges. The resulting quotient is an abstract simplicial complex, and each triangle column has three distinct edge rows.
\end{lemma}

\begin{lemma}[Properties of the surface boundary matrices]\label{lem:surface-port-template-certificate}
The surface boundary matrices have the following properties:
\begin{enumerate}[label=(\roman*)]
\item $R_ug=0$ exactly when $g$ is constant;
\item $0<\underline\mu_\ast\le\min_u\mu_u^{\rm true}$ and $\max_u\Gamma_u^{\rm true}\le\overline\Gamma_\ast<\infty$;
\item a constant two-chain with coefficient $c$ contributes $+c$ to each row of a port cycle attached with its orientation preserved, and $-c$ when the orientation is reversed.
\end{enumerate}
Parts (i)--(ii) follow from Lemma~\ref{lem:bc-template-existence}.
For (iii), interior edges cancel in a coherently oriented constant chain;
reversing a port attachment reverses its three boundary signs.
\end{lemma}

\subsection{The simplicial complex and its indices}

With one surface per variable, variables assigned larger surfaces would contribute more to the squared norm of the triangle coefficients. We therefore take as many copies of each surface as it has triangles. This removes the unequal weights. Theorem~\ref{thm:bc-mp-main} proves that this gives an isometry from the minimum-norm solution of the difference--average linear system to the minimum-norm solution of the boundary linear system.

\begin{definition}[The boundary linear system]\label{def:bc-construction-short}
Let $(A,b)$ be a bounded-occurrence difference--average instance after the zero-column preprocessing of Remark~\ref{rem:zero-columns-da}.  For each variable $i$, put $u_i:=\sum_q|a_{q,i}|\in[1,b_\ast]$ and $t_i:=t_{u_i}$, and take $t_i$ disjoint copies of the $t_i$-triangle surface $K^{(u_i)}$.

For each difference--average row $q$, create a fresh oriented triangular constraint cycle with three new vertices and edge rows $\ell_{q,1},\ell_{q,2},\ell_{q,3}$.  Before gluing, the cycles share no vertices with one another or with the surfaces. If a row becomes zero after repeated variables are collected, it keeps its row label. Its constraint cycle has three edges, no incident triangle, and zero right-hand side. Thus zero rows need not be removed or the remaining rows renumbered.  Use the reversible occurrence oracles of Definition~\ref{def:dabo}: after collecting repeated variables, coefficient $a_{q,i}$ gives $|a_{q,i}|$ slots of its sign, so $\pm2$ gives two slots.  For every slot and every copy of variable $i$, attach its assigned port to the row-$q$ cycle, preserving or reversing orientation according to the sign.  All ports are attached.  Lemma~\ref{lem:bc-simplicial-legality} gives a 2-dimensional abstract simplicial complex $K^{\rm bc}$ with
\[
\partial_2^{\rm bc}:C_2(K^{\rm bc})\to C_1(K^{\rm bc}),\qquad
\partial_2^{\rm bc}=\begin{bmatrix}M^{\rm bc}\\B^{\rm bc}\end{bmatrix},\qquad
\gamma^{\rm bc}=\begin{bmatrix}0\\b\otimes\mathbf1_3\end{bmatrix}.
\]
Here $M^{\rm bc}$ contains the rows for private interior edges and $B^{\rm bc}$ contains the rows for constraint cycles, in that order.

\end{definition}

\paragraph{Isometric encoding of the variables.}
Write $f_i=(f_i^{(1)},\ldots,f_i^{(t_i)})$, with $f_i^{(s)}\in\mathbb R^{t_i}$, and set
\[
q_i:=t_i^{-1}\mathbf1_{t_i^2},\qquad
Q:=\operatorname{diag}(q_1,\ldots,q_n),\qquad
Q^\top Q=\Id,\qquad \mathsf{Dec}(f)_i:=q_i^\top f_i.
\]

If blocks have different sizes, computing an address may require summing the sizes of all preceding blocks. We add simplices to give the blocks fixed sizes, so that addresses can be computed directly. We also show how to identify the surface assigned to each variable using only the oracle for the boundary matrix. This lets recovery choose the required unitary. Lemma~\ref{lem:reverse-typed-oracle-simulation} uses these properties with the specified labels and sparse ordering.

\begin{definition}[Indexed boundary matrix]\label{def:bc-typed-chain-model}
Let $e_u^{\rm int}$ be the number of edges of $K^{(u)}$ outside its ports.  Use the following one-based mathematical labels:
\[
\begin{array}{lll}
\text{triangle columns}&(i,s,\tau)&1\le i\le n,\quad1\le s,\tau\le t_i,\\
\text{interior-edge rows}&(i,s,\eta)&1\le i\le n,\quad1\le s\le t_i,\quad1\le\eta\le e_{u_i}^{\rm int},\\
\text{constraint-edge rows}&(q,r)&1\le q\le m,\quad1\le r\le3.
\end{array}
\]
Place the block of interior rows $\bigoplus_{i=1}^n\bigoplus_{s=1}^{t_i}R_{u_i}$ before the block of constraint rows, whose entries record the signed port attachments. This gives $\partial_{2,\mathrm{idx}}^{\rm bc}$.  The right-hand side $\gamma_{\mathrm{idx}}^{\rm bc}$ equals $b_q$ on $(q,r)$ and zero on interior rows.  Write $Q_{\mathrm{idx}}$ for the previously defined isometry $Q$ in these coordinates.  Lemma~\ref{lem:bc-fixed-capacity-indexing} gives the conversion to zero-based oracle addresses.
\end{definition}

\paragraph{Indices for blocks of fixed size.}
Using the edge counts outside the ports defined above, set
\[
\begin{aligned}
N_2^\ast&:=T_\ast^2,\\
N_1^\ast&:=\max_{1\le u\le b_\ast}
\{t_u e_u^{\rm int}+3(N_2^\ast-t_u^2)\},\\
F_2(u)&:=N_2^\ast-t_u^2,\qquad
F_1(u):=N_1^\ast-t_u e_u^{\rm int}-3F_2(u).
\end{aligned}
\]
Adjoin $F_2(u_i)$ disjoint triangles and $F_1(u_i)$ isolated edges to variable block $i$.  Each added component uses its own new vertices.  With zero-based addresses, the global ranks are
\[
\begin{aligned}
\operatorname{rank}_2(i_{\rm addr},c)&=i_{\rm addr}N_2^\ast+c,\\
\operatorname{rank}_{1,\mathrm{priv}}(i_{\rm addr},e)&=i_{\rm addr}N_1^\ast+e,\\
\operatorname{rank}_{1,\mathrm{con}}(q_{\rm addr},r_{\rm addr})
&=nN_1^\ast+3q_{\rm addr}+r_{\rm addr}.
\end{aligned}
\]
Here $i_{\rm addr}\in[n]$, $c\in[N_2^\ast]$, $e\in[N_1^\ast]$, $q_{\rm addr}\in[m]$, and $r_{\rm addr}\in[3]$.

\begin{lemma}[Sparse access with fixed block sizes]\label{lem:bc-fixed-capacity-indexing}
The preceding construction gives exactly $N_2^\ast$ triangle columns and $N_1^\ast$ private edge rows per variable block.  Inverting the rank maps and answering sparse queries require $O_{b_\ast}(1)$ occurrence queries and $\poly(\log(m+n))$ reversible arithmetic.  Every column represents an actual triangle, so the added components introduce no zero columns.

Up to row and column permutations, the resulting boundary matrix is
\[
\partial_{2,\rm lit}^{\rm bc}\cong
\partial_{2,\rm active}^{\rm bc}
\oplus\bigoplus_{i=1}^n\bigoplus_{d=1}^{F_2(u_i)}\partial_2(\Delta^2)
\oplus0_{\rm iso},
\]
where $0_{\rm iso}$ consists only of isolated-edge rows.  Extend $Q$ and the right-hand side by zero on the added coordinates.  The identities for the minimum-norm solution and recovery remain valid. The added triangles contribute only positive singular values $\sqrt3$.
\end{lemma}

\begin{lemma}[Identifying a surface from the boundary matrix]\label{lem:bc-template-identification}
Using the indexing of Lemma~\ref{lem:bc-fixed-capacity-indexing}, the coefficient multiplicity $u_i$ can be computed reversibly on superpositions of variable addresses $i_{\rm addr}=i-1$ using only $O_{b_\ast}(1)$ row and column queries to the supplied unweighted boundary matrix $B=\partial_2^{\rm bc}$.  This requires neither an oracle for the difference--average linear system nor an additional oracle identifying the surface type.  Reversing the procedure uncomputes the register containing $u_i$ and restores every auxiliary register to zero.
\end{lemma}

\begin{remark}[Orientations of the simplicial complex]\label{rem:bc-orientation-convention}
Use the specified orientations of the simplices, and represent negative coefficients by reversing the orientation of the port attachment. Passing to the simplicial complex then requires only row and column permutations.  Alternatively, use orientations induced by a vertex ordering. In this case, let $S_1$ map indexed edge coordinates, including added edges, to vertex-order edge coordinates, and let $S_2^\top$ be the analogous signed permutation for triangle coordinates.  Then
\[
\begin{gathered}
\partial_{2,\rm lit}^{\rm bc}=S_1(\partial_{2,\mathrm{idx}}^{\rm bc}\oplus\partial_{2,\rm fill})S_2,\\
\gamma_{\rm lit}^{\rm bc}=S_1(\gamma_{\rm idx}^{\rm bc}\oplus0),\qquad
Q_{\rm lit}=S_2^\top(Q_{\rm idx}\oplus0).
\end{gathered}
\]
Here $\partial_{2,\rm fill}$ contains the boundary matrices of the added disjoint triangles and the zero rows for isolated edges.  Extending by zero places no amplitude on the added triangles.  Recovery includes the corresponding $\{+1,-1\}$ orientation phases.  We use the specified orientations throughout.
\end{remark}

\begin{proposition}[The indexed boundary matrix represents a simplicial complex]\label{prop:bc-realization}
Under the occurrence-ordering and zero-column conventions, the construction and added simplices produce an abstract simplicial complex.  With the specified orientations, permutation matrices $P_1,P_2$ satisfy
\[
\begin{gathered}
\partial_{2,\rm lit}^{\rm bc}=P_1(\partial_{2,\mathrm{idx}}^{\rm bc}\oplus\partial_{2,\rm fill})P_2,\\
\gamma_{\rm lit}^{\rm bc}=P_1(\gamma_{\rm idx}^{\rm bc}\oplus0),\qquad
Q_{\rm lit}=P_2^\top(Q_{\rm idx}\oplus0).
\end{gathered}
\]
Here $\partial_{2,\rm fill}$ contains boundary matrices of isolated triangles and rows for isolated edges.  The singular values of the active part, the Moore--Penrose identities, the normalized states, and the recovery probabilities are preserved. The added components contribute only singular values $\sqrt3$, and their right-hand side is zero.  Henceforth $\partial_2^{\rm bc},\gamma^{\rm bc},Q$ denote the matrix, right-hand side, and isometry after adding these components.
\end{proposition}

\subsection{Solutions, spectral bounds, and reductions}

We first show that the equations on the constraint cycles reproduce the difference--average equations. Counting nonzero entries determines the sparse-access normalization. We then describe all solutions of the boundary linear system and show that the minimum-norm solution is the isometric image of the difference--average minimum-norm solution. The bounds on the recovered state's error and the probability of successful recovery are used in Lemmas~\ref{lem:herald-bounds} and~\ref{lem:lower-full-recovery-gap}.

\begin{lemma}[Equations on the constraint cycles]\label{lem:bc-loop-identity}
For every $y$,
\[
B^{\rm bc}Q=A\otimes\mathbf1_3,\qquad \norm{B^{\rm bc}Qy}^2=3\norm{Ay}^2.
\]
\end{lemma}

\begin{lemma}[Sparsity of the boundary matrix]\label{lem:bc-port-sparsity}
Each port edge attaches to one constraint edge. Rows for constraint edges have at most $4T_\ast$ nonzero entries, rows for interior edges have at most two, and triangle columns have exactly three.  Thus $\norm{\partial_2^{\rm bc}}^2\le12T_\ast=\Lambda_\ast$, independently of orientation signs.
\end{lemma}

\paragraph{Orthogonal subspaces of the variable blocks.}
The orthogonal decomposition $\mathbb R^{t_i^2}=\Span\{q_i\}\oplus U_i\oplus W_i$ uses
\[
\begin{aligned}
U_i&:=\{(a_1\mathbf1_{t_i},\ldots,a_{t_i}\mathbf1_{t_i}):\textstyle\sum_s a_s=0\},\\
W_i&:=\{(g_1,\ldots,g_{t_i}):\mathbf1^\top g_s=0\ \forall s\}.
\end{aligned}
\]

\begin{lemma}[Solutions of the boundary linear system]\label{lem:bc-kernel-short}
The equation $\partial_2^{\rm bc}f=\gamma^{\rm bc}$ holds exactly when the coordinates on added triangles vanish, the coefficients on every active surface copy are constant with $f_i^{(s)}=a_{i,s}\mathbf1_{t_i}$, and $x_i:=\sum_sa_{i,s}$ satisfies $Ax=b$.  Moreover,
\[
\ker(\partial_2^{\rm bc})=Q\ker(A)\oplus U,\qquad U:=\bigoplus_iU_i.
\]
\end{lemma}

\begin{theorem}[The encoding preserves the Moore--Penrose solution]\label{thm:bc-mp-main}
For feasible bounded-occurrence difference--average instances satisfying the occurrence-ordering convention and the treatment of zero columns in Remark~\ref{rem:zero-columns-da},
\[
(\partial_2^{\rm bc})^+\gamma^{\rm bc}=QA^+b
\]
for the unnormalized right-hand sides of Remark~\ref{rem:demand-representatives}.  Thus $Q$ maps the normalized difference--average solution state to the normalized solution state of the boundary linear system.
\end{theorem}

\begin{lemma}[Recovery with $Q^\top$]\label{lem:bc-decoder-stability}
Let $\ket{x}=A^+b/\norm{A^+b}$ and $\ket{f_\star}=Q\ket{x}$.  Use a fixed finite family of unitaries on the triangle coordinates assigned to each variable and $O_{b_\ast}(1)$ queries to the supplied oracles for the boundary matrix. A quantum circuit with these resources implements $Q^\top$ and records whether recovery succeeds.  On $\ket{f_\star}$, success is certain and the output is $\ket{x}$.  For normalized $\ket{\widetilde f}$ with $\norm{\ket{\widetilde f}-\ket{f_\star}}\le\eta\le1/4$, success has probability at least $(1-\eta)^2\ge1/2$ and the output error conditioned on success is at most $4\eta$.
\end{lemma}

Exact recovery does not bound the nonzero singular values of the boundary matrix. For this, we must bound its action on every vector orthogonal to its kernel. The interior equations bound differences between triangle coefficients within each surface. The next lemma bounds the effect of these differences on the constraint edges. Together with the difference--average singular-value bound, this gives the required lower bound for the boundary matrix.

\begin{lemma}[Bounding variations on the constraint cycles]\label{lem:bc-port-control}
For $\varepsilon\in\bigoplus_iW_i$, whose coefficients have mean zero on every surface copy,
\[
\norm{B^{\rm bc}\varepsilon}^2\le\Theta_\ast\norm{M^{\rm bc}\varepsilon}^2.
\]
\end{lemma}

\begin{lemma}[Spectral bounds for the boundary matrix]\label{lem:bc-spectral-main}
The simplicial boundary matrix satisfies
\[
\begin{gathered}
\lambda_{\max}((\partial_2^{\rm bc})^\top\partial_2^{\rm bc})\le\Lambda_\ast,\\
\lambda_{\min}^+((\partial_2^{\rm bc})^\top\partial_2^{\rm bc})
\ge c_\ast\min\{\lambda_{\min}^+(A^\top A),1\},\qquad
\cond(\partial_2^{\rm bc})\le C_{b_\ast}^{\rm bc}\cond(A).
\end{gathered}
\]
\end{lemma}

\begin{theorem}[Recovery and spectral bounds with indexed coordinates]\label{thm:bc-typed-chain-theorem}
For bounded-occurrence difference--average instances satisfying the occurrence-ordering and zero-column access conventions, the indexed model satisfies the constraint-cycle identity of Lemma~\ref{lem:bc-loop-identity}, the Moore--Penrose identity of Theorem~\ref{thm:bc-mp-main}, and the recovery and spectral bounds of Lemmas~\ref{lem:bc-decoder-stability} and~\ref{lem:bc-spectral-main}, using the corresponding operators in indexed coordinates.  The constants are listed after Definition~\ref{def:bc-templates}.  Proposition~\ref{prop:bc-realization} transfers these statements to the actual simplicial complex.
\end{theorem}

The quantum reduction needs singular-value bounds after division by the block-encoding normalization. The following results give these bounds and show how an error in the solution state affects the recovery error and success probability. Theorem~\ref{thm:da-bo-to-quphls-main} combines the reductions to boundary linear systems and upper Hodge Laplacian linear systems with suspension. Lemma~\ref{lem:herald-bounds} uses these estimates to bound the probability of recovering the circuit output.

\begin{remark}[Rescaling and repeating recovery]\label{rem:bc-normalization-heralding}
For the unweighted boundary matrix $B=\partial_2^{\rm bc}$, use sparse-access capacities $s_{\rm row}=4T_\ast$ and $s_{\rm col}=3$. Its block-encoding normalization is
\[
\alphabe(B)=\sqrt{s_{\rm row}s_{\rm col}}=\sqrt{\Lambda_\ast}.
\]
Thus $\norm{B/\alphabe(B)}\le1$. The encoded operator $B/\alphabe(B)$ has pseudoinverse $\alphabe(B)B^+$, so it yields the same normalized solution state as the unweighted boundary linear system $Bf=\gamma^{\rm bc}$. If a recovery trial has ideal success probability $p_\ast>0$, at least one of $O(p_\ast^{-1}\log(1/\delta))$ independent trials returns a success flag with probability at least $1-\delta$.
\end{remark}

\begin{lemma}[Bounds after rescaling the boundary matrix]\label{lem:bc-lower-norm}
Every nonzero boundary matrix $B=\partial_2^{\rm bc}$ constructed above satisfies
\[
\norm B\ge\sqrt3,\qquad
\frac{\sigma_{\min}^+(B)}{\alphabe(B)}
=\frac{\norm B}{\alphabe(B)\cond(B)}
\ge\frac{\sqrt3}{\alphabe(B)\cond(B)}.
\]
\end{lemma}

\begin{theorem}[Reduction to boundary linear systems]\label{thm:bc-da-to-b2-wrapper-explicit}
For fixed $b_\ast$ and $0<\varepsilon\le1$, assume the occurrence-ordering convention of Definition~\ref{def:dabo} and either no zero columns or the nonzero-column indexing oracles of Remark~\ref{rem:zero-columns-da}.  There is a bounded-error sparse-access Turing reduction
\[
\BoundedOccurrenceDA_{b_\ast,\cond,\varepsilon}
\le_T\BoundaryQLSP_{(4T_\ast,3),\,C_{b_\ast}^{\rm bc}\cond,\,\varepsilon/4}.
\]
The target linear system uses the unweighted boundary matrix $B=\partial_2^{\rm bc}$, with $\alphabe(B)=\sqrt{\Lambda_\ast}$ and right-hand side $\ket{\overline\gamma^{\rm bc}}:=\ket{\gamma^{\rm bc}}/\norm{\gamma^{\rm bc}}$.  Each target query uses a constant number of source queries. The target promise is $\sigma_{\min}^+(B)/\alphabe(B)\ge1/(C_{b_\ast}^{\rm bc}\cond)$. The spectral and precision bounds depend on the promised condition parameter of the source and fixed $b_\ast$. They contain no additional $\nnz(A)$ factor.
\end{theorem}

\begin{lemma}[Reduction to upper Hodge Laplacian linear systems on edges]\label{lem:bc-b2-to-quphls1-explicit}
Let $B=\partial_2$ be an unweighted simplicial boundary matrix supplied with row and column sparse-access capacities $4T_\ast$ and $3$, respectively, so $\alphabe(B)=\sqrt{12T_\ast}$. Suppose $\sigma_{\min}^+(B)/\alphabe(B)\ge1/\kappa_B$. For $0<\varepsilon\le1$,
\[
\BoundaryQLSP_{(4T_\ast,3),\kappa_B,\varepsilon}
\le_T\UpperHodgeQLSP_{1,\kappa_B^2,\varepsilon'},
\qquad
\varepsilon'=\frac{\varepsilon}{12\kappa_B}.
\]
The target matrix is $BB^\top$, block-encoded with normalization $\alphabe(B)^2$. Choose the block-encoding and implementation errors as in the proof. Then $O(\kappa_B^2\log(1/\delta))$ trials give a success flag with probability at least $1-\delta$. There is no additional $\nnz(A)$ dependence.
\end{lemma}

\subsection{Constants for the difference--average and simplicial constructions}\label{app:fixed-gadget-certificates}

The constants for constructing the simplicial boundary matrix are given after Definition~\ref{def:bc-templates}; the hardness reduction uses $b_\ast=b_{\rm aux}$. The proofs below verify that the construction gives a simplicial complex, establish its spectral bounds, and provide sparse access. These proofs require no computer-generated certificate.

These constants give the spectral bounds supplied in Definition~\ref{def:certified-parameters} for the difference--average and boundary linear systems. They do not grow with the number of matrix entries or simplices.

The constants $c_T,C_T,C_W,C_H,C_{\rm aux},\underline p_{\rm aux}$ for the auxiliary variables are determined by the gadgets for difference--average rows and the norms of their solution extension maps. They also depend on the coefficient, sparsity, padding, expander, and operator-norm bounds in Section~\ref{sec:lg-da}. Together with the propagation coefficient bound $P_\ast^{\rm clock}=32$, these determine $b_{\rm aux}$. None depends on the circuit length, matrix dimensions, or number of nonzero entries.

\subsection{Proofs for the boundary linear system}\label{app:boundary-proofs}

\paragraph{Triangulated surfaces (Lemma~\ref{lem:bc-template-existence}).}
\begin{proof}
Start with the $n_u$-gonal bipyramid, with vertices $T,B,R_0,\ldots,R_{n_u-1}$ and faces
\[
(T,R_i,R_{i+1}),\qquad (B,R_{i+1},R_i),
\]
where indices are modulo $n_u$.  Replace each face $(A,B,C)$ by a disk with fresh vertices $p_0,p_1,p_2,\allowbreak c_{01},c_{12},c_{20},\allowbreak q_0,\ldots,q_5$.  Its 25 triangles are the central triangle $(p_0,p_1,p_2)$; the collars
\[
(p_0,p_1,c_{01}),\quad(p_1,p_2,c_{12}),\quad(p_2,p_0,c_{20});
\]
the twelve triangles $(h_j,h_{j+1},q_{j+1})$ and $(h_j,q_{j+1},q_j)$, with
$(h_0,\ldots,h_5)=(p_0,c_{01},p_1,c_{12},p_2,c_{20})$ and indices modulo $6$; and
\[
\begin{aligned}
&(A,B,q_2),\ (A,q_2,q_1),\ (A,q_1,q_0),\\
&(B,C,q_4),\ (B,q_4,q_3),\ (B,q_3,q_2),\\
&(C,A,q_0),\ (C,q_0,q_5),\ (C,q_5,q_4).
\end{aligned}
\]
Only $A,B,C$ are inherited from the bipyramid.  Delete the central triangle in the first $u$ top faces.  Each deleted triangle leaves a designated triangular boundary cycle, called a port, with edges $\mathsf p_{u,a,r}$, $1\le a\le u$, $1\le r\le3$.  Denote the resulting complex by $K^{(u)}$.  Each port edge belongs to one collar triangle. Its third vertex is private and is never identified with another vertex in the quotient.

Each refined bipyramid face is a triangulated disk with $25$ triangles.  The
bipyramid has $2n_u$ faces, and deleting the central triangle in exactly $u$
refined top faces creates $u$ disjoint triangular boundary components.  Hence
$K^{(u)}$ is a connected orientable sphere with $u$ triangular holes and
\[
t_u=25(2n_u)-u=50n_u-u\le \max\{149,49b_\ast\}=T_\ast.
\]
The construction gives the required collar around each port: each port edge
belongs to one collar triangle, whose third vertex is internal. Internal vertices
outside the ports are never identified during gluing.

The proof of Lemma~\ref{lem:bc-simplicial-legality} below checks that the quotient is simplicial, including when several ports attach to one constraint cycle.

Let $G_u^{\rm dual}$ be the graph whose vertices are the triangles of $K^{(u)}$
and whose edges are the non-port edges of the surface shared by two triangles.  The
refined bipyramid surface is connected by paths of triangles sharing non-port edges, so
$G_u^{\rm dual}$ is connected.  With coherent orientations, each interior edge
row of $R_u$ has one $+1$ and one $-1$, hence $R_u^\top R_u$ is the graph
Laplacian of $G_u^{\rm dual}$.  Therefore $\ker R_u=\Span\{\mathbf1\}$.

The connected dual graph has $t_u$ vertices. Its algebraic connectivity,
the smallest positive eigenvalue of its graph Laplacian, satisfies the standard
bound~\cite[Theorem~4.3]{Fiedler1973}. Thus
\[
\mu_u^{\rm true}\ge2\bigl(1-\cos(\pi/t_u)\bigr)\ge t_u^{-2}\ge T_\ast^{-2}.
\]
Finally, every port-edge row of $P_u$
has exactly one nonzero entry, and each triangle column has a nonzero entry
on at most one port row. Hence $\norm{P_u}\le1$.  Hence
\[
\Gamma_u^{\rm true}
=\norm{P_u(R_u|_{W_u})^+}^2
\le \norm{P_u}^2\norm{(R_u|_{W_u})^+}^2
\le T_\ast^2.
\]
\end{proof}

\paragraph{The quotient is simplicial (Lemma~\ref{lem:bc-simplicial-legality}).}
\begin{proof}
Only vertices and edges of designated ports are glued.  Every retained triangle contains a vertex internal to its refined face: collar triangles have a private third vertex, annulus and outer-disk triangles contain private refinement vertices, and retained central triangles lie entirely in a refined face.  These internal vertices are never identified across faces or copies.

The disk construction separates the ports in a further way: each vertex outside the ports has neighbours on at most one port.  Collar and inner-annulus vertices meet only the central port of their own refined face, and the original bipyramid vertices meet no port vertices.  Each gluing map is injective on the three vertices of an individual port. Thus distinct neighbours of any fixed vertex outside the ports remain distinct after gluing, even when several ports map to the same constraint cycle. Edges with both endpoints outside the ports are unchanged. Edges with one endpoint outside the ports cannot merge, by this separation property.  Thus only designated port edges are identified.

No triangle can acquire repeated vertices.  Moreover, each retained triangle has at most one port edge, so at least two of its edges remain distinct private interior edges.  Within one refined face, the gluing map is injective on its only possible port, so triangles with distinct vertex sets remain distinct.  Triangles in different refined faces or copies retain distinct internal vertices.  Therefore distinct indexed triangles remain distinct two-simplices.  The argument is independent of the orientation of each attachment and permits two ports for one coefficient $\pm2$ to meet the same constraint cycle.
\end{proof}

\paragraph{Indices for blocks of fixed size (Lemma~\ref{lem:bc-fixed-capacity-indexing}).}
\begin{proof}
By definition, $F_2(u)$ and $F_1(u)$ are nonnegative.  Each added triangle uses three fresh vertices and edges, and each isolated edge uses fresh endpoints.  Because these components are disjoint, they do not identify vertices of any active simplex or make it degenerate.  Their numbers give exactly $N_2^\ast$ triangle columns and $N_1^\ast$ private edge rows per variable block.

The column domain is $[nN_2^\ast]$; the row domain is $[nN_1^\ast+3m]$, with private rows first.  For global column address $j$, division gives
\[
i_{\rm addr}=\lfloor j/N_2^\ast\rfloor,\qquad c=j\bmod N_2^\ast.
\]
Set $i=i_{\rm addr}+1$.  If $c<t_i^2$, write
\[
s_{\rm addr}=\lfloor c/t_i\rfloor,\qquad
\tau_{\rm addr}=c\bmod t_i,
\]
and use mathematical copy and triangle labels $s=s_{\rm addr}+1$, $\tau=\tau_{\rm addr}+1$.  Otherwise $d=c-t_i^2$ is the zero-based label of an added triangle.  To invert the row ranks, likewise use division and remainder by $N_1^\ast$ on private rows. On constraint rows, divide by $3$ after subtracting $nN_1^\ast$.  In particular, mathematical constraint label $(q,r)$ has address $nN_1^\ast+3(q-1)+(r-1)$.

For the local private-row address $e$, define
\[
A_i:=t_i e_{u_i}^{\rm int},\qquad G_i:=3F_2(u_i).
\]
There are three disjoint cases for the local address:
\[
\begin{array}{ll}
0\le e<A_i:
&s_{\rm addr}=\lfloor e/e_{u_i}^{\rm int}\rfloor,\quad
\eta_{\rm addr}=e\bmod e_{u_i}^{\rm int},\\[2pt]
A_i\le e<A_i+G_i:
&d=\lfloor(e-A_i)/3\rfloor,\quad r_{\rm fill}=(e-A_i)\bmod3,\\[2pt]
A_i+G_i\le e<N_1^\ast:
&\text{isolated-edge label }e-A_i-G_i.
\end{array}
\]
In the first case, $s=s_{\rm addr}+1$ and $\eta=\eta_{\rm addr}+1$ are mathematical labels.  Labels for added triangles and isolated edges are zero-based.  These formulas and the displayed ranks are mutually inverse and require no prefix sums.

To answer a query for an active column, use the fixed table of edge--triangle incidences for the surface and the bidirectional occurrence oracle.  The column for added triangle $d$ has private-row labels $A_i+3d$, $A_i+3d+1$, and $A_i+3d+2$, with signs from
\[
\partial[v_0,v_1,v_2]=[v_1,v_2]-[v_0,v_2]+[v_0,v_1].
\]
Conversely, a row for an edge of an added triangle $A_i+3d+r_{\rm fill}$ returns column $t_i^2+d$; all slots for an isolated-edge row return the dummy index and value zero.  Testing address ranges and computing the rank maps use reversible arithmetic on $O(\log(m+n))$ bits.  Each sparse query uses $O_{b_\ast}(1)$ occurrence queries.  Each column for an added triangle has three nonzero entries, each row for an edge of an added triangle has one, and each row for an isolated edge has none. Thus the stated sparsity bounds still hold.

Because the added components are disjoint, the boundary matrix has the orthogonal direct sum in the statement.  The boundary matrix of an isolated triangle has one column $(1,-1,1)^\top$ up to orientation, hence its only singular value is $\sqrt3$.  The Moore--Penrose inverse respects orthogonal direct sums.  The extended right-hand side vanishes on rows for added components, so its minimum-norm solution has zero coordinates on the added triangles. The map $Q$ also vanishes there.  Thus the norm and all recovery identities are unchanged.  Relative to the block-encoding normalization $\alphabe(B)=\sqrt{\Lambda_\ast}$, the smallest positive singular value is at least the minimum of the active-block bound divided by $\alphabe(B)$ and $\sqrt3/\alphabe(B)$.  Isolated edges add zero rows, not zero columns.
\end{proof}

\paragraph{Identifying the surface assigned to a variable (Lemma~\ref{lem:bc-template-identification}).}
\begin{proof}
For a variable address $i_{\rm addr}$, set $i=i_{\rm addr}+1$ and enumerate the $N_2^\ast$ columns with global addresses $\operatorname{rank}_2(i_{\rm addr},c)$, $c\in[N_2^\ast]$.  For each column, list its three incident edges. For each private edge, also query its row list, including the padding slots.  Every active surface triangle has at most one port edge and hence has a private interior edge shared by exactly two triangles.  Each added triangle has three private edges incident to exactly one triangle.  Thus a triangle is active exactly when it has an incident private row whose sparse list contains two valid column indices. The dummy indices lie outside the logical column range, so the two entries are identified by their indices.  The nonzero entries of the boundary matrix are exactly $\pm1$. If approximate value access is used, error less than $1/4$ permits identification by the threshold $1/2$.

The number of active columns is exactly $t_{u_i}^2$.  Distinct coefficient multiplicities give distinct triangle counts: $t_1=149$, $t_2=148$, and $t_u=49u$ for $u\ge3$.  Thus the count determines $u_i$, including when the block has no added triangles.  The number of queries is at most a constant times $N_2^\ast\cdot3\cdot4T_\ast$, and the arithmetic acts on the index registers.  Count reversibly and look up the count in a fixed table to compute $u_i$ without changing $i_{\rm addr}$. Reversing these steps uncomputes the result.
\end{proof}

\paragraph{From indexed coordinates to simplices (Proposition~\ref{prop:bc-realization}).}
\begin{proof}
By Lemma~\ref{lem:bc-simplicial-legality}, the indexed coordinates for active and added simplices correspond bijectively to the edges and triangles of the quotient. Giving the quotient simplices their specified orientations preserves every signed incidence. Negative occurrences in the difference--average equations are already represented by reversed port attachments. Thus passing to the simplicial complex changes the matrix only by row and column permutations. These preserve singular values and satisfy $(P_1AP_2)^+=P_2^\top A^+P_1^\top$. Sparse access is preserved because the construction retains the indexed labels and an ordering of slots that can be computed. Invariance under permutations alone would not establish sparse access. The signed-permutation alternative of Remark~\ref{rem:bc-orientation-convention} is algebraically equivalent.
\end{proof}

\paragraph{Constraint identity and sparsity (Lemmas~\ref{lem:bc-loop-identity} and~\ref{lem:bc-port-sparsity}).}
\begin{proof}
Under $Qy$, each triangle in each of the $t_i$ copies for variable $i$ has coefficient $y_i/t_i$. By Lemma~\ref{lem:surface-port-template-certificate}(iii), an occurrence with sign $\pm1$ contributes $\pm y_i$ to each of the three constraint-cycle edges after summing the copies. Expanding $\pm2$ into two occurrences with the same sign therefore gives $B^{\rm bc}Q=A\otimes\mathbf1_3$.

A difference--average row has at most four signed slots; each slot contributes one port from at most $T_\ast$ copies. Hence rows for constraint edges have at most $4T_\ast$ nonzero entries. Rows for private interior edges have at most two, and each triangle column has exactly three entries in distinct rows by Lemma~\ref{lem:bc-simplicial-legality}. The norm bound is $\norm{\partial_2^{\rm bc}}^2\le\norm{\partial_2^{\rm bc}}_1\norm{\partial_2^{\rm bc}}_\infty\le12T_\ast=\Lambda_\ast$.
\end{proof}

\paragraph{Kernel and minimum-norm solution (Lemma~\ref{lem:bc-kernel-short} and Theorem~\ref{thm:bc-mp-main}).}
\begin{proof}
The boundary matrices of the added triangles are injective, and their right-hand sides are zero. Thus every coordinate of an exact solution on an added triangle is zero. Constraint cycles for zero rows have no incident triangles and impose only the identity $0=0$.  On the active block, $M^{\rm bc}f=0$ means $R_{u_i}f_i^{(s)}=0$ in every copy.  Lemma~\ref{lem:surface-port-template-certificate}(i) therefore forces $f_i^{(s)}=a_{i,s}\mathbf1_{t_i}$.  The equations on edges of the constraint cycles are then exactly the signed equations from Lemma~\ref{lem:bc-loop-identity}, with $x_i=\sum_sa_{i,s}$, and hence give $Ax=b$.  The converse follows by assigning constant values in each copy and using the same constraint-cycle identity.  Setting $b=0$ gives the kernel statement after decomposing each vector that is constant on every copy into its component $x_iq_i$ along $q_i$ and a vector of differences between copies in $U_i$.

By Lemma~\ref{lem:bc-kernel-short}, every exact solution of the boundary linear system is obtained from an exact difference--average solution $x$ and constants $a_{i,s}$ on the copies satisfying $\sum_sa_{i,s}=x_i$.  In block $i$, the vector $f_i$ consists of $t_i$ copies, each with $t_i$ triangle coefficients equal to one of the constants $a_{i,s}$.  Hence
\[
\norm{f_i}^2=t_i\sum_{s=1}^{t_i}a_{i,s}^2.
\]
For fixed $x_i$, Cauchy--Schwarz gives
\[
\sum_{s=1}^{t_i}a_{i,s}^2
\ge \frac1{t_i}\left(\sum_{s=1}^{t_i}a_{i,s}\right)^2
=\frac{x_i^2}{t_i},
\]
with equality if and only if all constants are equal, \(a_{i,1}=\cdots=a_{i,t_i}=x_i/t_i\).  Thus \(\norm{f_i}^2\ge x_i^2\), with equality precisely when
\[
f_i=x_iq_i,
\qquad
q_i=\frac1{t_i}\mathbf1_{t_i^2}.
\]
Because each \(q_i\) has unit norm and different variable blocks are orthogonal, \(Q^\top Q=\Id\) and \(\norm{Qx}=\norm x\).  Therefore the minimum-norm solution of the boundary linear system among all vectors representing a fixed exact difference--average solution $x$ is exactly $Qx$.  The remaining minimization is
\[
\min\{\norm{Qx}:Ax=b\}=\min\{\norm x:Ax=b\},
\]
whose unique minimizer is the Moore--Penrose solution \(A^+b\).  Hence the minimum-norm solution of the boundary linear system is \(QA^+b\), as claimed.
\end{proof}

\paragraph{Recovery map (Lemma~\ref{lem:bc-decoder-stability}).}
\begin{proof}
Use the column indexing map for blocks of fixed size to write a column address as $(i_{\rm addr},c)$, and set $i=i_{\rm addr}+1$.  Use Lemma~\ref{lem:bc-template-identification} to compute $u_i$ using the supplied oracle for the boundary matrix.  Using $u_i$ as the control, apply the fixed unitary that maps the normalized vector $t_i^{-1}\mathbf1_{t_i^2}$ on the active local columns to local label zero.  Its first row is the transpose of $t_i^{-1}\mathbf1_{t_i^2}$, extended by zero on the coordinates of added triangles. Complete the remaining rows to an orthonormal basis of the full block.  For each surface type, fix such an orthonormal completion with algebraic entries. The gate convention of Remark~\ref{rem:implementation-gate-set} includes these unitaries.

Uncompute $u_i$ by reversing the procedure that identifies the surface using queries to the boundary matrix; this is possible because the variable address $i_{\rm addr}$ has not changed.  Recovery succeeds exactly when the transformed local label is zero.  Conditioned on this outcome, the remaining variable register is $Q^\top f$.  Every query is to the supplied boundary matrix and is included in the query count. All ancillary registers used to identify the surface type return to $\ket0$.

Since $Q^\top Q=\Id$, recovery succeeds with certainty on the ideal encoded state.  Choose a global phase for the approximate input state that attains its phase-invariant distance from the ideal state.  The unnormalized successful vector then differs from $Q^\top\ket{f_\star}=\ket{x}$ by at most $\eta$.  Its norm is therefore at least $1-\eta$, which gives the lower bound on the probability of success.  Lemma~\ref{lem:standard-normalization} bounds the state error conditioned on successful recovery by $2\eta/(1-\eta)\le4\eta$.
\end{proof}

\paragraph{Variations on the ports (Lemma~\ref{lem:bc-port-control}).}
\begin{proof}
Subtract the mean coefficient from each surface copy to obtain its residual. For each such residual, $\norm{P_ug}^2\le\overline\Gamma_\ast\norm{R_ug}^2$. Summing over copies bounds the total squared norm on port edges by $\overline\Gamma_\ast\norm{M^{\rm bc}\varepsilon}^2$. Each port entry contributes to one row for a constraint edge, and each row sums at most $L_\ast=4T_\ast$ entries. Applying Cauchy--Schwarz to each row gives $\norm{B^{\rm bc}\varepsilon}^2\le L_\ast\overline\Gamma_\ast\norm{M^{\rm bc}\varepsilon}^2=\Theta_\ast\norm{M^{\rm bc}\varepsilon}^2$.
\end{proof}

\paragraph{Singular values of the boundary matrix (Lemma~\ref{lem:bc-spectral-main}).}
\begin{proof}
The upper bound is Lemma~\ref{lem:bc-port-sparsity}.  We first bound the smallest positive eigenvalue of the Gram matrix for the active part of the boundary matrix. This part excludes the disjoint added components.  In the following calculation, $\partial_2^{\rm bc}$ denotes that block and $f$ is a unit vector orthogonal to its kernel.  In each variable block use the orthogonal decomposition introduced before Lemma~\ref{lem:bc-kernel-short},
\[
\mathbb R^{t_i^2}=\Span\{q_i\}\oplus U_i\oplus W_i.
\]
Thus $f$ decomposes uniquely as $f=Qy+u+\varepsilon$ with $u\in U:=\oplus_iU_i$ and $\varepsilon\in\oplus_iW_i$.  Lemma~\ref{lem:bc-kernel-short} gives $U\subseteq\ker(\partial_2^{\rm bc})$, so $f\perp\ker(\partial_2^{\rm bc})$ forces $u=0$.  The same lemma gives $Q\ker(A)\subseteq\ker(\partial_2^{\rm bc})$; since $Q$ is an isometry, the orthogonality of $f$ to $Q\ker(A)$ forces $y\perp\ker(A)$.  Therefore
\[
f=Qy+\varepsilon,
\qquad y\perp\ker(A),
\qquad \varepsilon\in\bigoplus_i W_i,
\qquad \norm y^2+\norm\varepsilon^2=1.
\]
By Lemma~\ref{lem:bc-loop-identity},
\[
\norm{B^{\rm bc}Qy}^2=3\norm{Ay}^2\ge 3\lambda_{\min}^+(A^\top A)\norm y^2.
\]
By Lemma~\ref{lem:surface-port-template-certificate}(ii),
\[
\norm{M^{\rm bc}\varepsilon}^2\ge\underline\mu_\ast\norm\varepsilon^2.
\]
Moreover Lemma~\ref{lem:bc-port-control} implies
\[
\norm{B^{\rm bc}Qy}
\le
\norm{B^{\rm bc}(Qy+\varepsilon)}+\norm{B^{\rm bc}\varepsilon}
\le
\norm{\partial_2^{\rm bc}f}+\sqrt{\Theta_\ast}\norm{M^{\rm bc}\varepsilon}
\le
(1+\sqrt{\Theta_\ast})\norm{\partial_2^{\rm bc}f}.
\]
Thus
\[
\norm{\partial_2^{\rm bc}f}^2
\ge
\frac{1}{2(1+\Theta_\ast)}\norm{B^{\rm bc}Qy}^2.
\]
Also
\[
\norm{\partial_2^{\rm bc}f}^2\ge\norm{M^{\rm bc}\varepsilon}^2\ge\underline\mu_\ast\norm\varepsilon^2.
\]
If $\norm y^2\ge1/2$, the first lower bound gives
\[
\norm{\partial_2^{\rm bc}f}^2
\ge
\frac{3}{4(1+\Theta_\ast)}\lambda_{\min}^+(A^\top A).
\]
If $\norm\varepsilon^2\ge1/2$, the second gives $\norm{\partial_2^{\rm bc}f}^2\ge\underline\mu_\ast/2$.  Since $c_\ast=3/[4(1+4T_\ast^3)]\le1/(2T_\ast^2)=\underline\mu_\ast/2$, both cases imply the stated lower bound.

Restore the added components using the orthogonal direct sum in Lemma~\ref{lem:bc-fixed-capacity-indexing}.  For a general column vector $f=f_{\rm act}\oplus f_{\rm fill}$ orthogonal to the full kernel,
\[
\|\partial_2^{\rm bc}f\|^2
=\|\partial_{2,\rm act}^{\rm bc}f_{\rm act}\|^2+3\|f_{\rm fill}\|^2.
\]
The zero rows from isolated edges and constraint cycles with no incident triangles add neither column coordinates nor positive singular values.  Since $c_\ast\le3$, the active-block lower bound is also a lower bound for the boundary matrix including its added components.  Every vector in its kernel has zero coordinates on the added triangles. The same holds for every exact solution whose right-hand side is zero on the added components.

Restrict to nonzero columns as in Remark~\ref{rem:zero-columns-da}, and exclude a zero minimum-norm solution because it cannot be normalized. The matrix of difference or averaging constraints is then nonzero whenever the normalized solution state is defined.  A nonzero matrix of difference or averaging constraints contains an entry of magnitude at least one, so $\norm A\ge1$.  Therefore $\lambda_{\min}^+(A^\top A)\ge1/\cond(A)^2$.  As $\cond(A)\ge1$, the spectral bounds give
\[
\cond(\partial_2^{\rm bc})^2
\le\frac{\Lambda_\ast}{c_\ast}\cond(A)^2,
\qquad
\frac{\Lambda_\ast}{c_\ast}
=16T_\ast+64T_\ast^4
\le81T_\ast^4=(C_{b_\ast}^{\rm bc})^2.
\]
This proves the claimed condition-number bound.
\end{proof}

\paragraph{Indexed coordinates and the norm bound (Theorem~\ref{thm:bc-typed-chain-theorem} and Lemma~\ref{lem:bc-lower-norm}).}
\begin{proof}
The proofs of the constraint, kernel, Moore--Penrose, recovery, and spectral statements use only $R_u,P_u$, the signed attachments for each occurrence, and the orthogonal decomposition within each block of copies. They therefore apply directly in indexed coordinates. Proposition~\ref{prop:bc-realization} transfers them to the simplicial boundary matrix with its retained labels and orientation convention. Every triangle column of the unweighted boundary matrix $B=\partial_2^{\rm bc}$ has norm $\sqrt3$, so $\norm B\ge\sqrt3$. The identity $\sigma_{\min}^+(B)=\norm B/\cond(B)$, divided by the fixed normalization $\alphabe(B)$, gives the stated singular-value bound after rescaling.
\end{proof}

\paragraph{Reduction from difference--average equations to boundary linear systems (Theorem~\ref{thm:bc-da-to-b2-wrapper-explicit}).}
\begin{proof}
Let $(A,\ket b)$ be a bounded-occurrence difference--average instance with $\cond(A)\le\cond$ and with the supplied indexing assumptions for zero columns of Remark~\ref{rem:zero-columns-da}.  If zero columns are present, use the supplied indexing oracles to restrict to the nonzero columns. This preserves the nonzero singular values. Appending zeros on the deleted coordinates restores the original normalized Moore--Penrose state.  We therefore assume $u_i\ge1$ for every variable used below.  This oracle reduction does not require scanning the entire black-box matrix.

Construct the indexed boundary matrix of Definition~\ref{def:bc-typed-chain-model} and its simplicial realization using the surface gadgets and port identifications defining $K^{\rm bc}$ in Definition~\ref{def:bc-construction-short}.  With the specified orientations, Proposition~\ref{prop:bc-realization} identifies the indexed incidence matrix with the ordinary simplicial boundary matrix up to row and column permutations. Thus all Moore--Penrose and spectral statements in indexed coordinates apply to $\partial_2^{\rm bc}$.  By Theorem~\ref{thm:bc-mp-main},
\[
(\partial_2^{\rm bc})^+\gamma^{\rm bc}=QA^+b.
\]
Set $B=\partial_2^{\rm bc}$. The ideal $\BoundaryQLSP$ output for this unweighted boundary linear system is
\[
\ket f
=
\frac{B^+\ket{\overline\gamma^{\rm bc}}}{\norm{B^+\overline\gamma^{\rm bc}}}
=
Q\frac{A^+\ket b}{\norm{A^+b}}.
\]
Lemma~\ref{lem:bc-decoder-stability} implements the partial isometry $Q^\top$ on the block of triangle coordinates assigned to each difference--average variable.  Its measured success flag equals one with probability one on the ideal encoded state.  If the target output is within $\varepsilon/4$ of the ideal state and $0<\varepsilon\le1$, Lemma~\ref{lem:bc-decoder-stability} gives success probability at least $1/2$ and, conditioned on success, normalized output error at most $4(\varepsilon/4)=\varepsilon$ relative to $A^+\ket b/\norm{A^+b}$.  If zero columns were removed in the preprocessing step, the final reversible embedding appends zero amplitudes in those deleted coordinates, which is exactly the Moore--Penrose value on the original columns by Remark~\ref{rem:zero-columns-da}.

The target satisfies the singular-value promise relative to its block-encoding normalization. Lemma~\ref{lem:bc-spectral-main} and $\alphabe(B)^2=\Lambda_\ast$ give
\[
\frac{\norm B}{\alphabe(B)}\le1,
\qquad
\frac{\sigma_{\min}^+(B)}{\alphabe(B)}
\ge
\sqrt{c_\ast/\Lambda_\ast}\,\min\{\sigma_{\min}^+(A),1\}.
\]
Since the preprocessed matrix of difference or averaging constraints is nonzero, Remark~\ref{rem:zero-columns-da} gives $\norm A\ge1$, and with $\cond(A)\le\cond$ we have $\sigma_{\min}^+(A)\ge1/\cond$.  Hence
\[
\frac{\sigma_{\min}^+(B)}{\alphabe(B)}
\ge \frac{1}{C_{b_\ast}^{\rm bc}\cond},
\]
and every nonzero singular value of the encoded operator $B/\alphabe(B)$ lies in $[1/(C_{b_\ast}^{\rm bc}\cond),1]$.
To prepare the normalized right-hand-side state from $\ket b$, create a uniform superposition over the three edges of each row cycle.  No amplitude-entry oracle for $b$ is required.  Explicitly, map
\[
\ket b=\sum_{q=1}^m b_q\ket{q-1}
\quad\mapsto\quad
\frac1{\sqrt3}\sum_{q=1}^m\sum_{r=1}^3
b_q\ket{\operatorname{rank}_{1,\mathrm{con}}(q-1,r-1)},
\]
which is exactly $\ket{\overline\gamma^{\rm bc}}$, because $\gamma^{\rm bc}$ is the three-fold repetition of $b$ on the subspace spanned by the constraint-cycle edges.

Sparse access uses the row and column labels for blocks of fixed size in Lemma~\ref{lem:bc-fixed-capacity-indexing}. The local rank maps and their inverses use fixed capacities, so they require no sums over preceding blocks.  Within an active variable block, the one-based mathematical labels are $(i,s,\tau)$ for triangle columns, $(i,s,\eta)$ for interior-edge rows, and $(q,r)$ for constraint-cycle edges, with $r\in\{1,2,3\}$.  Subtracting one from each component gives the zero-based local addresses specified in Lemma~\ref{lem:bc-fixed-capacity-indexing}; the rank formulas then give the global oracle addresses.

Source coefficient-occurrence queries use row address $q-1$ and variable address $i-1$, and returned occurrence addresses are increased by one before the port lookup.  The same lemma gives the labels and sparse queries for the added disjoint triangles and isolated edges.  To answer queries on interior edges, use the fixed table that records which edges belong to each triangle of the relevant surface.  For rows indexed by constraint-cycle edges, use the oracle for coefficient occurrences from Definition~\ref{def:dabo}. For each of the at most four signed row slots of $q$, it returns the variable label, sign, and index of that occurrence within the column. These determine the port attached to the slot.  For triangle columns, use the occurrence ordering within each column and the table of edge--triangle incidences for $K^{(u_i)}$. These recover the unique row slot and constraint-cycle edge incident to a designated port triangle.  Thus each target query uses $O(1)$ input-oracle queries plus $\poly(\log(m+n))$ reversible arithmetic.  The number of target rows and columns is $O_{b_\ast}(m+n)$, so the target address width is $O(\log(m+n))$ for fixed $b_\ast$.
\end{proof}

\paragraph{Reduction to upper Hodge Laplacian linear systems on edges (Lemma~\ref{lem:bc-b2-to-quphls1-explicit}).}
\begin{proof}
Let $U:=BB^\top$ and let $\gamma$ be the prepared right-hand side, with $\|\Pi_B\gamma\|\ge1/\zeta_\gamma$.  Since $\im(U)=\im(B)$, the range-overlap promise is unchanged.  The target $\UpperHodgeQLSP_1$ solver prepares
\[
\ket z\approx \frac{U^+\ket\gamma}{\norm{U^+\gamma}}.
\]
The upper Hodge Laplacian $U$ has nonzero eigenvalues $\sigma_i(B)^2$ and block-encoding normalization $\alphabe(B)^2$. Thus the nonzero spectrum of $U/\alphabe(B)^2$ lies in $[1/\kappa_B^2,1]$. The unweighted $\UpperHodgeQLSP_1$ instance therefore satisfies the stated spectral-gap promise with parameter $\kappa_B^2$.

Let
\[
z_\star:=U^+\gamma,
\qquad
f_\star:=B^+\gamma.
\]
By Lemma~\ref{lem:standard-mp-identities},
\[
f_\star=B^\top(BB^\top)^+\gamma=B^\top z_\star.
\]
Apply a block-encoding of $B^\top$ to the normalized state $z_\star/\norm{z_\star}$. Retain the outcome only when all ancillary registers of the block encoding are in their designated success states.  The block encoding of $B^\top$ has the same normalization $\alphabe(B)$, and the ideal success probability is
\[
p_\ast
=
\frac{\norm{B^\top z_\star}^2}{\alphabe(B)^2\norm{z_\star}^2}
=
\frac{\norm{f_\star}^2}{\alphabe(B)^2\norm{z_\star}^2}.
\]
The vector $z_\star$ lies in $\im B$.  The singular-value bound in
Lemma~\ref{lem:standard-mp-identities} therefore gives
\[
p_\ast=\frac{\|B^\top z_\star\|^2}{\alphabe(B)^2\|z_\star\|^2}
\ge\frac{\sigma_{\min}^+(B)^2}{\alphabe(B)^2}
\ge\frac1{\kappa_B^2}.
\]
For the construction using surface copies, $\alphabe(B)=\sqrt{12T_\ast}=O_{b_\ast}(1)$ follows from the fixed access capacities and exact entries $\pm1$. Repeat independent trials, making a fresh call to the $\UpperHodgeQLSP_1$ solver each time. At least one trial reports successful recovery with probability at least $1-\delta$ after $O(\kappa_B^2\log(1/\delta))$ trials. The approximate implementation below has the same asymptotic trial bound because its success probability is at least $p_\ast/4$.

Let the normalized input to the recovery step have phase-invariant error $\varepsilon_z$.  By Definition~\ref{def:be-normalization}, the successful block of an $(\alphabe(B),a,\varepsilon_{\rm be})$ block-encoding differs from $B^\top/\alphabe(B)$ by at most $\varepsilon_{\rm be}/\alphabe(B)$.  Including an implementation error $\varepsilon_{\rm impl}$, the unnormalized vector conditioned on successful recovery therefore differs from its ideal value by at most
\[
\delta_B:=\varepsilon_z+\frac{\varepsilon_{\rm be}}{\alphabe(B)}+\varepsilon_{\rm impl}.
\]
Lemma~\ref{lem:standard-normalization} implies that, if $\delta_B\le\sqrt{p_\ast}/2$, the implemented success probability is at least $p_\ast/4$, and the normalized successful state differs from the ideal solution state of the boundary linear system by at most
\[
\frac{4\delta_B}{\sqrt{p_\ast}}.
\]
Since $\sqrt{p_\ast}\ge1/\kappa_B$, it is sufficient to choose
\[
\varepsilon_z+\frac{\varepsilon_{\rm be}}{\alphabe(B)}+\varepsilon_{\rm impl}
\le
\frac{\varepsilon}{4\kappa_B}.
\]
For example, take $\varepsilon_z,\varepsilon_{\rm impl}\le\varepsilon/(12\kappa_B)$ and $\varepsilon_{\rm be}\le\alphabe(B)\varepsilon/(12\kappa_B)$. Since $\varepsilon\le1$, these choices also imply $\delta_B\le\sqrt{p_\ast}/2$.
\end{proof}
\begin{remark}[Dependence of the numerical bounds]
For the reductions to difference--average and boundary linear systems above, the numerical bounds depend on local sparsities, fixed construction constants, promised spectral bounds, and requested accuracy. They also depend on the block-encoding normalization set by the sparse-access capacities. This applies to the condition-number bounds, the preparation and recovery precisions, the repetition or amplification costs, and the query costs. They contain no additional polynomial factor in the number of nonzero entries of the expanded source matrix. Arithmetic on addresses takes time polynomial in the number of address bits.

The construction of the boundary linear system uses $O(m+n)$ simplices for an $m\times n$ difference--average input with a fixed occurrence bound; this affects the address lengths without adding a factor depending on the full matrix dimensions to the numerical bounds.
\end{remark}
\section{Raising the dimension by suspension}\label{sec:dimension-raising}

Suspension raises the simplex dimension by one. It preserves the solution by assigning equal and opposite amplitudes to the two copies on the cones. Write $\spec^+(H)=\spec(H)\cap(0,\infty)$ for the positive spectrum. All boundary matrices in this section are unweighted incidence matrices. We rescale them only within their block encodings.

The previous section gives upper Hodge Laplacian instances on edges. We use suspension to obtain instances in every fixed higher dimension. We must preserve the solution and bound every positive eigenvalue, including the eigenvalues on the symmetric subspace. Lemma~\ref{lem:suspension-spectrum-explicit} therefore computes the complete spectrum. Lemma~\ref{lem:suspension-locality-explicit} shows how to query the new boundary matrix using the original oracles. Together, these results prove the reduction in Lemma~\ref{lem:dim-raise-explicit}.

\begin{remark}[Restriction to the $(i+1)$-skeleton]\label{rem:skeleton-convention}
An upper Hodge Laplacian instance on $i$-simplices uses only $B_{i+1}:C_{i+1}(X)\to C_i(X)$. The $(i+1)$-skeleton is the subcomplex containing all simplices of dimension at most $i+1$. Restricting to it leaves this boundary matrix and the corresponding upper Hodge Laplacian unchanged. We use this convention below, so $\partial_{i+2}^X=0$; the constructed complexes already have dimension $i+1$.
\end{remark}

\begin{definition}[Suspension and antisymmetric embedding]\label{def:suspension-explicit}\label{def:suspension-basis-explicit}\label{def:anti-embedding-explicit}
Fix $i\ge1$ and a simplicial complex $X$ on $[n_0]$ with $\dim X\le i+1$ and $C_{i+1}(X)\ne\{0\}$. Its suspension $\Sigma X$ has two new vertices ordered as $a<b<\overline v$ for every original vertex $v$. Its simplices are $\overline\sigma$, $a*\overline\sigma$, and $b*\overline\sigma$ for $\sigma\in X$, including the empty simplex; none contains both $a$ and $b$. Orient each simplex by this vertex order.

Write $B=B_{i+1}^X=\partial_{i+1}^X$. Order the chain bases as
\[
C_{i+2}(\Sigma X)=C_{i+1}(X)^{\oplus2},\qquad
C_{i+1}(\Sigma X)=C_{i+1}(X)\oplus C_i(X)^{\oplus2}.
\]
The two domain blocks contain the $a$- and $b$-cones. The three codomain blocks contain the original simplices, $a$-cones, and $b$-cones. Define
\[
J_i x=(0,x,-x),\qquad
\widehat J_i=J_i/\sqrt2:C_i(X)\longrightarrow C_{i+1}(\Sigma X).
\]
Then $\widehat J_i^\top\widehat J_i=\Id$.
\end{definition}

\begin{lemma}[Boundary matrix under suspension]\label{lem:suspension-block-explicit}
In these oriented bases, the unweighted boundary matrix after suspension is
\[
B'=B_{i+2}^{\Sigma X}=\partial_{i+2}^{\Sigma X}
=\begin{pmatrix}\Id&\Id\\-B&0\\0&-B\end{pmatrix}.
\]
This is the cone-boundary formula for this vertex order~\cite[Section~2.1, p.~121]{Hatcher2002}.
\end{lemma}

\begin{definition}[Block-encoding normalization under suspension]\label{def:susp-scalar}
Let $r,c\ge1$ be the supplied integer row and column sparsity bounds for $B$. Use the target bounds $r'=\max\{2,r\}$ and $c'=c+1$, and define
\[
\begin{aligned}
\alphabe(B)&=\sqrt{rc},&
\alphabe(B')&=\sqrt{r'(c+1)},\\
\alphabe(BB^\top)&=\alphabe(B)^2,&
\alphabe(B'(B')^\top)&=\alphabe(B')^2,\\
\chi(r,c)&:=\max\left\{\frac{r'(c+1)}{rc},\frac{r'(c+1)}2\right\}.&&
\end{aligned}
\]
These factors specify the block encodings of the unweighted matrices. The boundary matrices themselves remain unweighted. After each upper Hodge Laplacian is divided by its specified normalization, the bound on the reciprocal of its smallest positive eigenvalue increases by at most the rational factor $\chi(r,c)\ge1$.
\end{definition}

\begin{lemma}[Spectrum after suspension]\label{lem:suspension-spectrum-explicit}
Let $L_i^\uparrow=BB^\top$ and $L_{i+1}^{\uparrow,\Sigma}=B'(B')^\top$. Then:
\begin{enumerate}[label=(\roman*)]
\item $L_{i+1}^{\uparrow,\Sigma}J_i=J_iL_i^\uparrow$.
\item If $\lambda_1,\ldots,\lambda_h$ are the positive eigenvalues of $L_i^\uparrow$, listed with multiplicity, the positive spectrum of $L_{i+1}^{\uparrow,\Sigma}$ consists of
\[
\lambda_j,\quad\lambda_j+2\quad(1\le j\le h),
\qquad 2\quad\text{with multiplicity }\dim\ker B.
\]
Each singular-pair block in the symmetric subspace also contributes one zero eigenvalue. Vectors in $\ker B^\top$ give zero eigenvectors on both cones.
\item If $\spec^+(L_i^\uparrow/\alphabe(B)^2)\subseteq[1/\cond,1]$, then
\[
\begin{aligned}
\spec^+\!\left(\frac{L_{i+1}^{\uparrow,\Sigma}}{\alphabe(B')^2}\right)
&\subseteq
\left[\min\left\{\frac{\alphabe(B)^2}{\alphabe(B')^2\cond},\frac2{\alphabe(B')^2}\right\},1\right]\\
&\subseteq\left[\frac1{\chi(r,c)\cond},1\right].
\end{aligned}
\]
In particular, $\sigma_{\min}^+(B')\ge\min\{\sigma_{\min}^+(B),\sqrt2\}$ for the unweighted boundary matrices.
\end{enumerate}
\end{lemma}

\begin{lemma}[Oracle implementation for the boundary matrix after suspension]\label{lem:suspension-locality-explicit}
For the indexed complexes satisfying Remark~\ref{rem:skeleton-convention}, each row or column query to $B'$ uses $O(1)$ queries to $B$. The remaining computation is reversible, has cost polynomial in the address width, and returns its workspace to zero. The boundary matrix entries are exactly $0,\pm1$. The supplied bounds $r,c\ge1$ give target bounds $\max\{2,r\},c+1$ and the block-encoding normalizations of Definition~\ref{def:susp-scalar}, computable to any requested precision in polynomial time.
\end{lemma}

\begin{lemma}[Raising the dimension of constructed instances]\label{lem:dim-raise-explicit}
For each $i\ge1$, under the preceding assumptions on simplex indexing and sparsity, suspension gives the bounded-error sparse-access reduction
\[
\UpperHodgeQLSP_{i,\cond,\varepsilon}
\le_T\UpperHodgeQLSP_{i+1,\chi(r,c)\cond,O(\varepsilon)}.
\]
It sends $(X,r,c,\gamma)$ to $(\Sigma X,\max\{2,r\},c+1,\widehat J_i\gamma)$ and recovers the source solution by applying $\widehat J_i^\top$. For the unweighted upper Hodge Laplacians, the spectrum calculation gives the exact identities
\[
\Pi_{L_{i+1}^{\uparrow,\Sigma}}\widehat J_i=\widehat J_i\Pi_{L_i^\uparrow},\qquad
(L_{i+1}^{\uparrow,\Sigma})^+\widehat J_i\gamma
=\widehat J_i(L_i^\uparrow)^+\gamma.
\]
These identities preserve both the norm of the projection of the right-hand side onto the range and the normalized solution. The task parameter $\cond$ bounds the reciprocal of the smallest positive eigenvalue of the rescaled upper Hodge Laplacian $L/\alphabe(L)$. The reduction uses the simplex labels and decompositions provided by the construction. Sparse access to an arbitrary matrix alone does not provide this information.
\end{lemma}

\begin{corollary}[Iterated suspension]\label{cor:dim-raise-explicit}
For a constructed dimension-one instance with boundary matrix sparsity bounds $r_1,c_1\ge1$, set
\[
r_{j+1}=\max\{2,r_j\},\qquad c_j=c_1+j-1,\qquad
\Xi_k(r_1,c_1)=\prod_{j=1}^{k-1}\chi(r_j,c_j),
\]
with empty product $1$. Every fixed $k\ge1$ admits
\[
\UpperHodgeQLSP_{1,\cond,\varepsilon}
\le_T\UpperHodgeQLSP_{k,\Xi_k(r_1,c_1)\cond,O(\varepsilon)}.
\]
For the equal-copy construction, $r_1=4T_\ast$ and $c_1=3$, so $\Xi_k$ depends only on fixed $k$ and the occurrence bound $b_\ast$.
\end{corollary}

\begin{definition}[Fixed suspension multiplier]\label{def:bc-susp-multiplier}
For the equal-copy construction, set
\[
r_j^{\rm bc}=4T_\ast,\qquad c_j^{\rm bc}=j+2,\qquad
\Xi^{\rm susp}_{k,b_*}=\prod_{j=1}^{k-1}\chi(r_j^{\rm bc},c_j^{\rm bc}).
\]
The degree-$(j+1)$ boundary matrix has normalization $\sqrt{r_j^{\rm bc}c_j^{\rm bc}}$, and its upper Hodge Laplacian on $j$-simplices has normalization $r_j^{\rm bc}c_j^{\rm bc}$. For fixed $k,b_*$ this rational multiplier is independent of circuit length and matrix dimensions. It equals $1$ for $k=1$.
\end{definition}

\begin{theorem}[Reduction from bounded-occurrence difference--average systems to upper Hodge Laplacian linear systems]\label{thm:da-bo-to-quphls-main}
For fixed $b_*$ and $k\ge1$, assume the conditions of Definitions~\ref{def:literal-da}--\ref{def:dabo} and Remark~\ref{rem:zero-columns-da}. These require nondegeneracy, a fixed ordering of occurrences, and either no zero columns or supplied oracles that index the nonzero columns. Then
\[
\BoundedOccurrenceDA_{b_*,\cond,\varepsilon}
\le_T\UpperHodgeQLSP_{k,\,\Xi^{\rm susp}_{k,b_*}(C_{b_*}^{\rm bc}\cond)^2,\,\varepsilon/\poly(\cond)},
\qquad C_{b_*}^{\rm bc}=9T_*^2.
\]
Each target row or column query uses $O_{k,b_*}(1)$ source queries. The bound on the reciprocal of the smallest positive eigenvalue of $L_k^\uparrow/\alphabe(L_k^\uparrow)$ and the required solver error tolerance depend only on $\cond,k,b_*$, with no additional $\nnz(A)$ factor. The target supplies the oracles for the unweighted boundary matrix and its sparsity bounds that determine its block-encoding normalization, as in Definition~\ref{def:scaled-boundary-access}.
\end{theorem}

The complete reduction must give explicit spectral bounds. It must also recover the answer using only the final instance's oracles. Definition~\ref{def:certified-parameters} and Lemma~\ref{lem:constructed-gap-promise} provide the bounds used in Theorem~\ref{thm:hardness-final} to set the precision and prove an inverse-polynomial decision gap. Lemma~\ref{lem:reverse-typed-oracle-simulation} shows how recovery uses the simplex labels and final oracle to answer queries to each earlier boundary matrix. No additional oracle is needed.

\begin{definition}[Explicit spectral and condition-number bounds]\label{def:certified-parameters}
For a circuit of length $L$, the reduction supplies rational bounds
\[
\begin{aligned}
\kappa_{\rm clock}^{\max}(L)&=4(2L+1)^2,&
\kappa_{\rm DA}^{\max}(L)&=C_{\rm aux}C_H\kappa_{\rm clock}^{\max}(L),\\
\kappa_B^{\max}(L)&=C_{b_*}^{\rm bc}\kappa_{\rm DA}^{\max}(L),&
\kappa_{\rm final}^{\max}(L)&=\Xi^{\rm susp}_{k,b_*}(\kappa_B^{\max}(L))^2.
\end{aligned}
\]
The clock bound follows from Lemma~\ref{lem:path-green}, with $C_{\rm clock}=4$. The first two quantities bound condition-number ratios. The third bounds the reciprocal of the smallest positive singular value of $B_2/\alphabe(B_2)$; the last bounds the reciprocal of the smallest positive eigenvalue of $L_k^\uparrow/\alphabe(L_k^\uparrow)$. All displayed coefficients are fixed reduction constants. We use these supplied rational bounds to set the decision circuit's precision and acceptance gap.
\end{definition}

\begin{lemma}[Spectral bounds for the constructed upper Hodge Laplacians]\label{lem:constructed-gap-promise}
For a difference--average input with $\cond\le\kappa_{\rm DA}^{\max}$, Theorem~\ref{thm:da-bo-to-quphls-main} produces
\[
\spec^+\!\left(\frac{L_k^\uparrow}{\alphabe(L_k^\uparrow)}\right)\subseteq
\left[\frac1{\Xi^{\rm susp}_{k,b_*}(C_{b_*}^{\rm bc}\kappa_{\rm DA}^{\max})^2},1\right].
\]
For the propagation hardness instances the lower endpoint is $1/\kappa_{\rm final}^{\max}(L)$, and
\[
\alphabe(L_k^\uparrow)=4T_\ast(k+2)=O_{k,b_*}(1).
\]
\end{lemma}

\begin{lemma}[Intermediate oracle simulation from the final oracle]\label{lem:reverse-typed-oracle-simulation}
On the constructed upper, lower, and full Hodge Laplacian instances, recovery needs to query intermediate boundary matrices and identify the surface types used by $Q^\top$. The fixed recovery circuit answers each such query using $O_k(1)$ queries to the final instance's boundary matrix oracles. It uses the canonical simplex and component labels of Remark~\ref{rem:canonical-produced-ordering} and the supplied sparsity bounds. The index arithmetic is reversible, has cost polynomial in the address width, and returns its workspace to zero. The maps for right-hand-side indices come from the classical description. No extra oracle for an intermediate boundary matrix or surface type is supplied.
\end{lemma}

\subsection{Proofs for suspension and reduction composition}
\label{subsec:proofs-suspension-classification}

\paragraph{Spectrum after suspension (Lemma~\ref{lem:suspension-spectrum-explicit}).}
\begin{proof}
The unweighted boundary matrix after suspension in Lemma~\ref{lem:suspension-block-explicit} gives
\[
L_{i+1}^{\uparrow,\Sigma}
=\begin{pmatrix}
2\Id&-B^\top&-B^\top\\
-B&BB^\top&0\\
-B&0&BB^\top
\end{pmatrix}.
\]
Applied to $J_i x=(0,x,-x)$, its first block cancels and the remaining blocks are $BB^\top x$ and $-BB^\top x$, proving (i).

For (ii), use symmetric and antisymmetric cone coordinates
\[
u_+=(u_a+u_b)/\sqrt2,\qquad u_-=(u_a-u_b)/\sqrt2.
\]
The upper Hodge Laplacian is orthogonally equivalent to
\[
\begin{pmatrix}
2\Id&-\sqrt2 B^\top&0\\
-\sqrt2 B&BB^\top&0\\
0&0&BB^\top
\end{pmatrix}.
\]
For a nonzero singular pair $Bv=\sigma u$, $B^\top u=\sigma v$, the symmetric block on $\Span\{(v,0),(0,u)\}$ is
\[
\begin{pmatrix}2&-\sqrt2\sigma\\-\sqrt2\sigma&\sigma^2\end{pmatrix},
\]
with eigenvalues $0$ and $2+\sigma^2$. The antisymmetric block contributes $\sigma^2=\lambda_j$. Vectors in $\ker B$ contribute eigenvalue $2$ on the subspace supported on the original simplices. Vectors in $\ker B^\top$ give zero eigenvectors on both cones. Together with the zero eigenvector from each singular-pair block in the symmetric subspace, these vectors account for the remaining subspaces.

For (iii), the source promise gives $\lambda_j\ge \alphabe(B)^2/\cond$. Part (ii) therefore bounds every positive target eigenvalue below by $\min\{\alphabe(B)^2/\cond,2\}$. The target boundary matrix $B'$ has entries $0,\pm1$, row sparsity at most $r'=\max\{2,r\}$ and column sparsity at most $c+1$, so its upper Hodge Laplacian has norm at most $\alphabe(B')^2$. Dividing by this block-encoding normalization gives the first interval in (iii). Since $\cond\ge1$,
\[
\min\left\{\frac{\alphabe(B)^2}{\alphabe(B')^2\cond},\frac2{\alphabe(B')^2}\right\}
\ge\frac1{\chi(r,c)\cond}.
\]
Taking square roots of the unweighted positive spectra gives the stated bound on the singular values of the boundary matrix.
\end{proof}

\paragraph{Suspension access (Lemma~\ref{lem:suspension-locality-explicit}).}
\begin{proof}
A cone column contains one base-face entry $1$ and the negatives of the entries in one source column. Base rows have two cofaces; cone rows contain a copy of a source row with its signs reversed. Thus $r'=\max\{2,r\}$ and $c'=c+1$ bound the target row and column sparsities. The block positions are fixed. Each target query therefore uses a constant number of source queries, together with reversible index computations and sign changes. The normalizations are square roots of supplied positive integers and can be computed to the precision required for block encoding.
\end{proof}

\paragraph{Raising the dimension (Lemma~\ref{lem:dim-raise-explicit} and Corollary~\ref{cor:dim-raise-explicit}).}
\begin{proof}
Let $L_i=L_i^\uparrow(X)$ and $L_{i+1}=L_{i+1}^{\uparrow,\Sigma}$. Since $\ker L_i=\ker B^\top$ and $\im L_i=\im B$, the antisymmetric part of the kernel is exactly the image under $\widehat J_i$ of $\ker L_i$. The orthogonal complement of this part within the kernel lies in the symmetric subspace. Consequently
\[
\Pi_{L_{i+1}}\widehat J_i=\widehat J_i\Pi_{L_i},\qquad
\|\Pi_{L_{i+1}}\widehat J_i\gamma\|=\|\Pi_{L_i}\gamma\|.
\]
Thus the projection onto the range has the same norm for the source and target right-hand sides. This holds for every actual prepared right-hand side, including one with a kernel component.

By Lemma~\ref{lem:suspension-spectrum-explicit}(i), the restriction of $L_{i+1}$ to $\im\widehat J_i$ is $\widehat J_iL_i\widehat J_i^\top$. If $L_i e_j=\lambda_j e_j$ with $\lambda_j>0$, then $\widehat J_i e_j$ has the same eigenvalue for $L_{i+1}$. Both pseudoinverses annihilate the corresponding kernel components, so
\[
L_{i+1}^+\widehat J_i\gamma=\widehat J_iL_i^+\gamma.
\]
For the rescaled operators represented by the block encodings, the same identity becomes
\[
\left(\frac{L_{i+1}}{\alphabe(B')^2}\right)^+\widehat J_i\gamma
=\frac{\alphabe(B')^2}{\alphabe(B)^2}\widehat J_i\left(\frac{L_i}{\alphabe(B)^2}\right)^+\gamma.
\]
The positive scalar disappears when the solution vector is normalized. Thus the target solution state is exactly $\widehat J_i$ applied to the source solution state. Part (iii) of the spectrum lemma shows that, after division by the target block-encoding normalization, the reciprocal of the smallest positive eigenvalue is at most $\chi(r,c)\cond$. The solver receives access to the unweighted boundary matrix and the target sparsity bounds.

To recover, apply $\widehat J_i^\top$. A Hadamard on the register that distinguishes the two cones identifies the antisymmetric branch. The ideal encoded state lies entirely in that branch, so recovery succeeds with probability one. Lemma~\ref{lem:standard-normalization} gives phase-invariant recovery error $O(\varepsilon)$ for an $\varepsilon$-approximate encoded state. The circuit acts efficiently on the supplied simplex indices.

Iterating $k-1$ times gives $r_{j+1}=\max\{2,r_j\}$, $c_j=c_1+j-1$ and the multiplier $\Xi_k(r_1,c_1)$. The number of source queries needed to answer a final query and the factor by which recovery can increase the error depend only on the fixed number of suspensions. For the hard family the normalizations and $\Xi_k$ depend only on $k,b_\ast$.
\end{proof}

\paragraph{Composing the reductions and bounding the spectrum (Theorem~\ref{thm:da-bo-to-quphls-main} and Lemma~\ref{lem:constructed-gap-promise}).}
\begin{proof}
Compose Theorem~\ref{thm:bc-da-to-b2-wrapper-explicit}, Lemma~\ref{lem:bc-b2-to-quphls1-explicit}, and Corollary~\ref{cor:dim-raise-explicit}. The base boundary matrix is $B_2=\partial_2^{\rm bc}$, with sparsity bounds $4T_\ast,3$ and normalization $\alphabe(B_2)=\sqrt{\Lambda_\ast}$. The reduction to a boundary linear system gives
\[
\sigma_{\min}^+\!\left(\frac{B_2}{\alphabe(B_2)}\right)
\ge\frac1{C_{b_\ast}^{\rm bc}\cond}
\]
and embeds the difference--average solution state by applying $Q$. The upper Hodge Laplacian after rescaling is the Gram matrix of $B_2/\alphabe(B_2)$, so the reciprocal of its smallest positive eigenvalue is at most $(C_{b_\ast}^{\rm bc}\cond)^2$. Suspension multiplies that bound by $\Xi^{\rm susp}_{k,b_\ast}$. This proves the stated spectral interval and, since $r_k^{\rm bc}c_k^{\rm bc}=4T_\ast(k+2)$, the normalization formula.

Recovery of the solution to the base boundary linear system succeeds with probability at least $1/\kappa_B^2$, where $\kappa_B=C_{b_\ast}^{\rm bc}\cond$. The required inverse precision and number of trials are polynomial in this bound. Each suspension preserves the ideal solution state under $\widehat J_j$. It increases the error and query costs by at most fixed factors. Sparsities, entry magnitudes, and block-encoding normalizations are $O_{k,b_\ast}(1)$. All numerical bounds depend on $\cond(A),k,b_\ast$, without an additional $\nnz(A)$ factor.
\end{proof}

\paragraph{Answering earlier queries from the final oracle (Lemma~\ref{lem:reverse-typed-oracle-simulation}).}
\begin{proof}
At one suspension, the target boundary matrix is
\[
B'=\begin{pmatrix}\Id&\Id\\-B&0\\0&-B\end{pmatrix}.
\]
Every entry of $B$ therefore appears in either cone block with opposite sign. To answer a source column query, query the corresponding cone column, discard the single entry from the identity block on the original simplices, retain the cone rows, and reverse the signs. For a source row query, query the corresponding cone row and reverse its signs. The canonical ordering of sparse entries puts copied entries, identity entries, and dummy slots in computable positions. Reversible arithmetic processes their labels and signs in time polynomial in the address width and returns its workspace to zero. Each simulated query uses a constant number of target queries.

After $k-1$ inverse suspension steps for an upper Hodge Laplacian instance, the recovered boundary matrix is exactly $B_2=\partial_2^{\rm bc}$. A block encoding of its transpose uses the normalization determined by its sparsity bounds. It requires queries only to the final instance. Lemma~\ref{lem:bc-template-identification} identifies the surface types required by $Q^\top$ from the simulated oracles for the base boundary matrix.

For lower Hodge Laplacian instances, the same procedure reverses the $k-2$ suspensions of the final degree-$k$ boundary matrix. For full Hodge Laplacian instances, first select the indexed hard component; the optional disjoint $(k+1)$-simplex has a fixed identifying label. The optional fixed simplex has a full Hodge Laplacian with a constant positive smallest eigenvalue. After division by the block-encoding normalization, its eigenvalues remain bounded below by a fixed positive constant: the supplied sparsity bounds are fixed for fixed $k,b_\ast$. Thus every query to an intermediate boundary matrix, and every query used to identify a surface type, is answered using the final instance's boundary matrix oracles. All of these queries are counted. No oracle that directly queries a Hodge Laplacian is required.
\end{proof}
\section{Hardness and quantum algorithms in the oracle and circuit input models}\label{sec:hodge-classification}

We now use the reductions between solution states to prove decision hardness. The acceptance probability must include runs in which recovery fails. Lemmas~\ref{lem:herald-bounds} and~\ref{lem:composed-heralded-accounting} bound the success probabilities and approximation errors, including the uniformly random output on failure. Their bounds specify how accurately the solution state must be prepared in Theorem~\ref{thm:hardness-final}. Together with the spectral bounds from the previous section, they show that recovery succeeds with probability at least inverse polynomial in the source circuit length and that inverse-polynomial error in the prepared solution state is sufficient. The decision hardness proved this way also establishes state-preparation hardness: repeated preparation of a sufficiently accurate solution state, followed by the prescribed recovery and measurement, lets us decide the source problem with bounded error.

\paragraph{Recovery operations.}
On the constructed instances, recovery reverses the suspension embeddings, applies a block-encoded boundary matrix or its transpose when required, applies $Q^\top$ to recover the difference--average coordinates, and applies $D_{\rm aux}$ to remove the auxiliary offsets. It then tests the output qubit when the clock lies in the interval for the appended identity gates. The right-hand sides are supported on $O_k(1)$ computational-basis states, and the matrix-entry oracles are classical reversible circuits. All operations use the supplied instance oracles; no independent verification circuit is provided.

The prescribed recovery records a success flag for each stage. When all flags equal one, it runs the Feynman--Kitaev output test. Otherwise it outputs an independently sampled uniformly random bit. Let $\kappa_B$ bound $\alphabe(B)/\sigma_{\min}^+(B)$ for the base boundary matrix $B=B_2$, and let $\underline p_{\rm aux}>0$ be the fixed rational bound from Lemma~\ref{lem:lg-mp-decoder}. For the product of lower bounds on success probabilities below, set $\underline p_{\rm dim}=1$ and $\underline p_{B^\top}=1/(\kappa_B^{\max}(L))^2$. Equal-copy recovery succeeds with probability one on the ideal input. For approximate inputs covered by Lemma~\ref{lem:bc-decoder-stability}, we use the lower bound $\underline p_{\rm bc}=1/2$.

\begin{lemma}[Probabilities of successful recovery]\label{lem:herald-bounds}
The measured flags and ideal conditional success probabilities of the recovery stages satisfy
\[
\begin{array}{lll}
\text{Stage}&\text{Flag}&\text{Ideal conditional probability}\\
\UpperHodgeQLSP_k\to\UpperHodgeQLSP_1&F_{\rm dim}&p_{\rm dim}=1\\
\UpperHodgeQLSP_1\to\BoundaryQLSP&F_{B^\top}&p_{B^\top}\ge1/\kappa_B^2\\
\BoundaryQLSP\to\BoundedOccurrenceDA&F_{\rm bc}&p_{\rm bc}=1\\
\BoundedOccurrenceDA\to\FKClockLS&F_{\rm aux}&p_{\rm aux}\ge\underline p_{\rm aux}.
\end{array}
\]
On the constructed instances, these supplied bounds give
\[
\underline p_{\rm chain}
=\underline p_{\rm dim}\underline p_{B^\top}\underline p_{\rm bc}\underline p_{\rm aux}
=\frac{\underline p_{\rm aux}}{2(\kappa_B^{\max}(L))^2}
\ge1/\poly(L).
\]
\end{lemma}

\paragraph{Errors in recovery.}
Let $\eta_{\rm in}$ bound the error in the input state, measured in $\dist_{\rm ph}$. Let $\eta_{\rm be}$ bound the sum of operator-norm errors in the blocks that implement successful recovery, after applying their matrix transformations. Let $\eta_{\rm dec}$ bound the accumulated operator-norm error from synthesizing the remaining recovery unitaries. Write $E=\eta_{\rm in}+\eta_{\rm be}+\eta_{\rm dec}$ and $C_{\rm cond}=4$. The parameter $\varepsilon_{\rm be}$ measures the error in the unnormalized matrix represented by the block encoding, before any matrix transformation. Thus it differs from the error in the normalized encoded block and from the error after QSVT. Using the block that approximates the rescaled transpose $B^\top/\alphabe(B)$ of the boundary matrix contributes at most $\varepsilon_B/\alphabe(B)$; Theorem~\ref{thm:qsvt-final} bounds errors after QSVT. The complete ideal and implemented channels $\mathcal E_x,\widetilde{\mathcal E}_x$ output the classical decision bit. They include flag measurements and random output on failure. Bound their distance by
\[
\tfrac12\|\widetilde{\mathcal E}_x-\mathcal E_x\|_\diamond\le\eta_{\rm impl}.
\]

\begin{lemma}[Error bounds for composed recovery operations]\label{lem:composed-heralded-accounting}
Consider a fixed finite sequence of recovery operations with measured success flags, and let $p_{\rm chain}>0$ be a lower bound on the probability that all ideal recovery operations succeed.
\begin{enumerate}[label=(\roman*)]
\item \emph{Conditional state.} If $E\le\sqrt{p_{\rm chain}}/2$, the recovered state, conditioned on success, has phase-invariant error at most $C_{\rm cond}E/\sqrt{p_{\rm chain}}$.
\item \emph{Unconditional decision.} Suppose the ideal final test, conditional on recovery success, accepts yes instances with probability at least $1/2+\delta$ and no instances with probability at most $1/2-\delta$, where $0<\delta\le1/2$. Uniformly random output on failure gives an unconditional acceptance bias of magnitude at least $p_{\rm chain}\delta$ with the same sign. The complete decision circuits satisfy
\[
\left|\Pr[\widetilde W_x^{\mathcal O_x}(\ket{\widetilde\psi})\text{ accepts}]
-\Pr[W_x^{\mathcal O_x}(\ket\psi)\text{ accepts}]\right|
\le\eta_{\rm in}+\eta_{\rm impl}.
\]
If this sum is at most $p_{\rm chain}\delta/2$, the implemented procedure retains an acceptance bias of magnitude at least $p_{\rm chain}\delta/2$, without conditioning on success.
\end{enumerate}
\end{lemma}

\begin{lemma}[Spectral bounds for the constructed instances]\label{lem:constructed-quphls-spectral-promise}
For a length-$L$ circuit, Theorems~\ref{thm:shifted-to-da} and~\ref{thm:da-bo-to-quphls-main} give
\[
\spec^+\!\left(\frac{L_k^\uparrow}{\alphabe(L_k^\uparrow)}\right)
\subseteq[1/\kappa_{\rm final}^{\max}(L),1].
\]
This is the spectral promise of Definition~\ref{def:quphls}. The specified block-encoding normalization is computed from fixed sparsity bounds and is $O_{k,b_\ast}(1)$.
\end{lemma}

\begin{theorem}[Hardness of upper Hodge Laplacian linear systems]\label{thm:hardness-final}
For every fixed $k\ge1$, the decision problem
$\VerifyProblemSA\text{-}\UpperHodgeQLSP_k$ with the prescribed recovery
is $\BQP$-hard: a classical polynomial-time reduction constructs
$1^\nu$, the target sparse-access description, and all boundary matrix
and right-hand-side oracle circuits. Its circuit-encoded version
$\VerifyProblemCirc\text{-}\UpperHodgeQLSP_k$ is $\BQP$-hard under
classical polynomial-time many-one reductions. Consequently, preparing
the normalized minimum-norm solution state to the prescribed
inverse-polynomial accuracy is $\BQP$-hard on these instances.

The oracle circuits are generated uniformly, and recovery receives no
arbitrary verifier. The right-hand side has at most $c_k$ nonzero
computational-basis amplitudes, with $c_k$ depending only on $k$. The
functions that answer sparse queries have classical reversible circuits of size
$\poly(n_C)$.
\end{theorem}

To prove the quantum upper bounds, we need a solver for every input that satisfies the promises. This includes right-hand sides with a component in the kernel. Theorem~\ref{thm:qsvt-final} gives such a solver and bounds the required precision and number of queries using the spectral promise. Theorems~\ref{thm:containment-final} and~\ref{thm:lower-full-containment} apply it to the upper, lower, and full Hodge Laplacians. Lemma~\ref{lem:rhs-robustness} bounds the change in the solution when the prepared right-hand side has a small error. The proof of hardness for lower Hodge Laplacians uses this bound to compare the solution for the actual prepared state with the solution for the reference state.

\begin{theorem}[Quantum algorithm for positive semidefinite linear systems]\label{thm:qsvt-final}
Let $H\succeq0$ have a block encoding with known normalization $\alphabe(H)>0$ and matrix error $\varepsilon_{\rm be}$. Suppose
\[
\spec^+\!\left(H/\alphabe(H)\right)\subseteq[1/\cond,1],
\qquad \|\Pi_Hb\|\ge1/\zeta>0
\]
for the normalized prepared state $\ket b$.

For $0<\varepsilon\le1/2$, choose a QSVT polynomial of degree
\[
d=O\!\left(\cond\log\frac{\cond\zeta}{\varepsilon}\right)
\quad\text{and require}\quad
\varepsilon_{\rm be}\le\alphabe(H)\left(\frac{\varepsilon}{512\cond\zeta d}\right)^2.
\]
With the arithmetic and synthesis precisions specified in the proof, one amplified preparation attempt has failure probability at most $\varepsilon^2/16$. Conditioned on the algorithm reporting success, the output is within phase-invariant distance $\varepsilon$ of $H^+b/\|H^+b\|$.
The attempt calls the block encoding, its inverse, and their controlled versions
\[
Q_{\rm BE}=O\!\left(\cond^2\zeta
\log\frac{\cond\zeta}{\varepsilon}\log\frac1\varepsilon\right)
\]
times, and makes polynomially many calls in these parameters to $\mathcal O_b,\mathcal O_b^\dagger$. If each call to the block encoding uses $q_{\rm be}$ sparse queries, the block-encoding calls use $q_{\rm be}Q_{\rm BE}$ sparse queries in total. Repeating on failure increases the expected cost by a constant factor. Decision algorithms use one attempt per sample and include the failure outcome, so each sample has a fixed worst-case cost.
\end{theorem}

\begin{lemma}[Effect of error in the right-hand side]\label{lem:rhs-robustness}
Suppose $H\succeq0$ and $\spec^+(H/\alphabe(H))\subseteq[1/\kappa,1]$. For normalized $b_{\rm ref}\in\im(H)$ and prepared $b_x$ with $\dist_{\rm ph}(b_x,b_{\rm ref})\le\delta_{\rm rhs}$, choose the minimizing phase. Then
\[
\|\Pi_Hb_x\|\ge1-\delta_{\rm rhs},\qquad
\|H^+(b_x-b_{\rm ref})\|\le\frac{\kappa}{\alphabe(H)}\delta_{\rm rhs}.
\]
If $\kappa\delta_{\rm rhs}\le1/2$, their normalized pseudoinverse states have phase-invariant distance at most $4\kappa\delta_{\rm rhs}$.  The prepared right-hand side may have a component in $\ker H$. The pseudoinverse sets that component to zero.  In particular $\delta_{\rm rhs}\le1/2$ permits $\zeta=2$, and $\delta_{\rm rhs}\le\varepsilon/(8\kappa)$ gives an error of at most $\varepsilon/2$.
These bounds follow from Lemmas~\ref{lem:standard-mp-identities}
and~\ref{lem:standard-normalization}, using
$\|H^+b_{\rm ref}\|\ge1/\alphabe(H)$.
\end{lemma}

\begin{theorem}[Upper Hodge Laplacian solution-state preparation and its oracle cost]\label{thm:containment-final}
For fixed $k\ge1$, the task of Definition~\ref{def:quphls}, with actual prepared right-hand side $\gamma_x=\mathcal O_\gamma\ket0$, polynomially bounded sparsities, and supplied bounds $\kappa,\zeta_\gamma,\varepsilon^{-1}\le\poly(\nu)$, admits a polynomial-query quantum algorithm. The specified normalization $\alphabe(L_k^\uparrow)$ is determined by the supplied sparsity bounds. The algorithm uses $q_{\rm be}Q_{\rm BE}$ queries to the boundary matrix, with $Q_{\rm BE}$ as in Theorem~\ref{thm:qsvt-final}, and polynomially many calls to $\mathcal O_\gamma,\mathcal O_\gamma^\dagger$, including those used for reflections. For the XOR lookup interface, $q_{\rm be}=O(s_{\rm row}+s_{\rm col})$, and $O_k(1)$ on the hard family. The designated circuit encodings give gate complexity polynomial in the circuit-encoded input length.
\end{theorem}

\begin{theorem}[Quantum algorithms for lower and full Hodge Laplacian solution states]\label{thm:lower-full-containment}
Consider the access model for unweighted boundary matrices of Definition~\ref{def:lower-full-hodge}, with the specified normalizations $\alphabe(B_j)=\sqrt{s_{{\rm row},j}s_{{\rm col},j}}$. Omit boundary matrices known to be absent and their Gram-matrix terms. For $M=L_k^\downarrow$ or $M=\Delta_k$, suppose $\spec^+(M/\alphabe(M))\subseteq[1/\kappa,1]$ and the actual prepared right-hand side $\eta_x=\mathcal O_\eta\ket0$ has a projection onto the range of norm at least $1/\zeta_\eta$.

If the sparsities and $\kappa,\zeta_\eta,\varepsilon^{-1}$ are polynomially bounded in $\nu$, solution-state preparation has polynomial query complexity, with computed block-encoding normalizations
\[
\alphabe(L_k^\downarrow)=\alphabe(B_k)^2,
\qquad
\alphabe(\Delta_k)=\sum_{\substack{j\in\{k,k+1\}\\B_j\text{ present}}}\alphabe(B_j)^2.
\]
Theorem~\ref{thm:qsvt-final} gives the cost for the corresponding normalized spectral promise. All queries to the boundary matrices and to $\mathcal O_\eta,\mathcal O_\eta^\dagger$ are counted. The designated circuit encodings give polynomial gate complexity by Lemma~\ref{lem:circuit-input-simulation}.
\end{theorem}
To solve the decision problem, we apply the fixed recovery circuit to the prepared solution. Lemma~\ref{lem:state-to-verifier-final} bounds the resulting acceptance-probability error. Corollary~\ref{cor:complete-final} then shows that repeated sampling distinguishes the promised cases. Lemmas~\ref{lem:circuit-input-simulation} and~\ref{lem:finite-oracle-encoding-size} give the corresponding results for finite circuit inputs: executing the supplied oracle circuits gives a quantum algorithm, and constructing these circuits uniformly gives a classical many-one hardness reduction. The cost includes every oracle call used in recovery.

\begin{lemma}[Effect of state error on acceptance probability]\label{lem:state-to-verifier-final}
For normalized pure states $\psi,\phi$ and a fixed verification circuit with access to the instance oracles $W^{\mathcal O_x}$,
\[
\left|\Pr[W^{\mathcal O_x}(\ket\psi)\text{ accepts}]-\Pr[W^{\mathcal O_x}(\ket\phi)\text{ accepts}]\right|\le\dist_{\rm ph}(\ket\psi,\ket\phi).
\]
This is the standard measurement-stability bound
\cite[Theorem~3.4, Corollary~3.40, and Eq.~(1.186)]{WatrousTQI}.
\end{lemma}

\begin{corollary}[Hardness and quantum algorithms for upper Hodge Laplacian linear systems]\label{cor:complete-final}
For fixed $k\ge1$, $\VerifyProblemSA\text{-}\UpperHodgeQLSP_k$ is
$\BQP$-hard under the reductions of Theorem~\ref{thm:hardness-final}
and has a polynomial-query quantum algorithm. On the encodings of
Definition~\ref{def:circuit-encoded-input},
$\VerifyProblemCirc\text{-}\UpperHodgeQLSP_k$ is $\BQP$-complete under
classical polynomial-time many-one reductions. Preparing the solution
state is therefore $\BQP$-hard: repeated preparation, followed by the
prescribed recovery and measurement, solves the same decision problem.
Theorem~\ref{thm:containment-final} gives a polynomial-query algorithm for that task.

These claims require a prescribed error tolerance for the solution state
with polynomially bounded inverse, inverse-polynomial lower bounds on
the acceptance-probability gap and smallest positive eigenvalue of
$L_k^\uparrow/\alphabe(L_k^\uparrow)$, polynomial sparsity bounds, and
$\zeta_\gamma\le\poly(\nu)$.
\end{corollary}

For the lower Hodge Laplacian linear system, we choose a right-hand side proportional to $B_k^\top\gamma$ so that its normalized solution is the desired solution state for the boundary linear system. Lemma~\ref{lem:lower-rhs-preparation} constructs the preparation circuit required by the input model. Lemma~\ref{lem:explicit-boundary-desuspension} recovers the solution of the original boundary linear system by reversing the suspensions. The bounds below show how accurately the right-hand side and recovery operations must be implemented to preserve the decision gap. Together, these results prove Theorem~\ref{cor:lower-full-hodge-hardness}. In the top dimension of the complex, the full and lower Hodge Laplacians coincide. Corollary~\ref{cor:full-both-nonzero} extends hardness to full Hodge Laplacians with both terms nonzero: it adds a fixed disjoint simplex and checks the spectral promise after dividing the full Hodge Laplacian by its specified normalization.

\begin{lemma}[Preparing the right-hand side for the lower Hodge Laplacian linear system]\label{lem:lower-rhs-preparation}
For a pair $(B_k,\gamma)$ in the fixed-$k$ hardness family, with $\gamma$ normalized, let $\eta_{\rm ref}=B_k^\top\gamma/\|B_k^\top\gamma\|$.  For every $0<\delta_{\rm rhs}<1$, the classical reduction outputs a unitary circuit $C_{\rm rhs}$ of size $\poly(n_C,\log(1/\delta_{\rm rhs}))$ whose actual prepared state $\eta_x=C_{\rm rhs}\ket0$ satisfies $\dist_{\rm ph}(\eta_x,\eta_{\rm ref})\le\delta_{\rm rhs}$.  All workspace qubits return exactly to $\ket0$; reversing the circuit supplies the inverse oracle.
\end{lemma}

\begin{lemma}[Recovery of the solution to a boundary linear system after suspension]\label{lem:explicit-boundary-desuspension}
For one suspension of a complex of dimension $i\ge2$, let $B_i=\partial_i$ and $B_{i+1}^{\Sigma}=\partial_{i+1}^{\Sigma X}$. Set $Kx=(x,-x)/\sqrt2$ and $Jy=(0,y,-y)/\sqrt2$. For $\gamma\in\im(B_i)$,
\[
(B_{i+1}^{\Sigma})^+J\gamma=-K(B_i^+\gamma).
\]
The normalized solution lies in $\im K$.  A fixed unitary that makes no oracle queries sends $K\ket x$ to $\ket0_{\rm cone}\ket x$, where the cone register identifies the $a$- or $b$-cone.  Retaining this register preserves the phase-invariant input error $\epsilon$ in the joint output. Discarding it gives a reduced state with trace distance at most $\epsilon$, where trace distance is $\frac12\|\rho-\sigma\|_1$.  These bounds persist through any fixed number of suspensions.
\end{lemma}

\begin{lemma}[Acceptance bias after recovery for lower and full Hodge Laplacian instances]\label{lem:lower-full-recovery-gap}
On the reference solution to the boundary linear system, the prescribed recovery for the lower Hodge Laplacian linear system, including uniformly random output on failure, has ideal acceptance bias of magnitude at least
\[
\Delta_{\rm base}^{\downarrow,\mathrm{full}}=\frac{\underline p_{\rm aux}}{96},
\qquad
\Delta_{\rm rec}=\frac{\underline p_{\rm aux}}{192}>0.
\]
The bias is positive on yes instances and negative on no instances. Recovery and arithmetic implementation errors that change acceptance by at most $\Delta_{\rm base}^{\downarrow,\mathrm{full}}/2$ leave an acceptance bias of magnitude at least $\Delta_{\rm rec}$. We account for the error in preparing the right-hand side separately. These bounds are independent of circuit length for fixed $k$ and fixed choices in the difference--average reduction and its recovery map. The component that encodes the hard full Hodge Laplacian instance uses the same circuit.
\end{lemma}

\paragraph{Accuracy of the right-hand side in the hardness construction.}
For the reference right-hand side $\eta_{\rm ref}\in\im(L_k^\downarrow)$ of Lemma~\ref{lem:lower-rhs-preparation}, let $\kappa_\downarrow$ bound the reciprocal of the smallest positive eigenvalue of $L_k^\downarrow/\alphabe(L_k^\downarrow)$ and take $\Delta_{\rm rec}=\underline p_{\rm aux}/192$ from Lemma~\ref{lem:lower-full-recovery-gap}. Choose the circuit so that it returns its workspace to zero and prepares an actual state within phase-invariant distance
\begin{equation}\label{eq:lower-rhs-hardness-tolerance}
\delta_{\rm rhs}\le\min\left\{\frac{\varepsilon}{16\kappa_\downarrow},\frac{\Delta_{\rm rec}}{16\kappa_\downarrow}\right\}
\end{equation}
of $\eta_{\rm ref}$. Lemma~\ref{lem:rhs-robustness} then bounds the resulting errors in the solution state and acceptance probability.

\begin{theorem}[Hardness and decision-problem completeness for lower and full Hodge Laplacian linear systems]\label{cor:lower-full-hodge-hardness}
For every fixed $k\ge2$, the decision problem for lower Hodge Laplacian
linear systems with the prescribed recovery is $\BQP$-hard: a classical
polynomial-time reduction constructs the target sparse-access description
and all oracle circuits. It has a polynomial-query quantum algorithm,
and its designated circuit-encoded version is $\BQP$-complete under
classical polynomial-time many-one reductions. The same conclusions
hold for the unweighted full Hodge Laplacian of
Definition~\ref{def:lower-full-hodge}. Consequently, preparing the
corresponding normalized Moore--Penrose solution states is $\BQP$-hard
on the constructed instances. Repeated preparation, followed by the same
prescribed recovery and measurement, decides the source computation
with bounded error. The constructed instances have right-hand-side support bounded
by a constant depending only on $k$, and the circuits that answer sparse queries
are classical and reversible.

The target uses the actual state prepared by the designated
right-hand-side circuit, which returns its workspace to zero, with tolerance~\eqref{eq:lower-rhs-hardness-tolerance}.
All claims require a prescribed error tolerance for the solution state
with polynomially bounded inverse, inverse-polynomial lower bounds on
the acceptance-probability gap, smallest positive eigenvalue after division
by the specified block-encoding normalization, and norm of the right-hand
side's projection onto the range, and polynomial sparsity bounds.
Theorem~\ref{thm:lower-full-containment} supplies the quantum algorithms.
\end{theorem}

\begin{corollary}[Hardness for full Hodge Laplacians with both summands nonzero]\label{cor:full-both-nonzero}
The conclusions for full Hodge Laplacian linear systems in Theorem~\ref{cor:lower-full-hodge-hardness} hold with both $B_k^\top B_k$ and $B_{k+1}B_{k+1}^\top$ nonzero. For fixed $k\ge2$, adjoin a disjoint $(k+1)$-simplex with all its faces and give the added component zero right-hand-side amplitude. The reference right-hand side remains in $\im(B_k^\top)$ on the hard component, and its normalized solution extends by zero. The spectral bound after rescaling changes only by a fixed factor, accounting for both the simplex eigenvalue $k+2$ and the normalization computed for the full Hodge Laplacian. For an approximately prepared right-hand side, use the corresponding bound on the reciprocal of the smallest positive eigenvalue of the full Hodge Laplacian after rescaling in~\eqref{eq:lower-rhs-hardness-tolerance}; the guarantees for the projection onto the range and for the error of Theorem~\ref{cor:lower-full-hodge-hardness} then apply.
\end{corollary}

\begin{remark}[Dimension range of the lower and full Hodge Laplacian hardness results]\label{rem:lower-full-scope}
The claims for lower and full Hodge Laplacian linear systems start at $k\ge2$ because they begin with $B_2$ and suspend.  They do not establish hardness for graph incidence matrices at $k=1$ for $B_1^\top B_1$, or hardness for arbitrary explicitly listed complexes. Corollary~\ref{cor:full-both-nonzero} uses a disjoint component and makes no claim about connected hard instances. Independently of connectedness, the standard Hodge decomposition, Lemma~\ref{lem:standard-hodge-decomposition}, gives $L_k^\downarrow L_k^\uparrow=L_k^\uparrow L_k^\downarrow=0$ and orthogonal ranges.
\end{remark}

\subsection{Proofs for recovery, hardness, and quantum algorithms}\label{subsec:proofs-classification}
\paragraph{Recovery probabilities (Lemma~\ref{lem:herald-bounds}).}
\begin{proof}
Inverse suspension and projection onto the equal-copy subspace each succeed ideally with probability one, by Corollary~\ref{cor:dim-raise-explicit} and Theorem~\ref{thm:bc-da-to-b2-wrapper-explicit}. The success probability for the step applying $B^\top$ has the lower bound of Lemma~\ref{lem:bc-b2-to-quphls1-explicit}; removing the auxiliary variables succeeds with constant probability by Lemma~\ref{lem:lg-mp-decoder}. For equal-copy recovery, the product uses the conservative bound for approximate inputs stated above, even though exact recovery succeeds with probability one. Multiplying the four lower bounds on conditional success probabilities gives the stated product. The normalization is constant and the inverse spectral bound is polynomial. Hence this success probability is at least inverse polynomial in $L$.
\end{proof}

\paragraph{Composed recovery and measurement (Lemma~\ref{lem:composed-heralded-accounting}).}
\begin{proof}
Let $p$ be the probability that all ideal recovery operations succeed, and let $p_{\rm test}$ be the acceptance probability of the fixed final test conditional on success. Uniformly random output on every failure gives unconditional acceptance probability $(1-p)/2+pp_{\rm test}$, so conditional bias at least $\delta$ becomes bias at least $p_{\rm chain}\delta$ when $p\ge p_{\rm chain}$. Each successful recovery block is a contraction, as is their product. Comparing the ideal and implemented factors one at a time bounds the error in the unnormalized vector on success by $\eta_{\rm in}+\eta_{\rm be}+\eta_{\rm dec}$.  Its ideal norm is at least $\sqrt{p_{\rm chain}}$, and the stated error threshold makes the implemented norm at least $\sqrt{p_{\rm chain}}/2$.  Lemma~\ref{lem:standard-normalization} therefore gives the bound on the state conditioned on success with $C_{\rm cond}=4$.

For the decision bound, the complete procedure, including flag measurements and random failure bits, is a quantum channel (a completely positive, trace-preserving map) followed by a two-outcome measurement. The trace distance between the input pure states is at most their $\ell_2$ distance minimized over global phase.  Trace distance cannot increase under this map. This fact and the stated bound on half the diamond-norm distance~\cite[Theorem~3.4, Corollary~3.40, and Definition~3.43]{WatrousTQI} bound the acceptance changes from input and implementation errors by $\eta_{\rm in}$ and $\eta_{\rm impl}$, respectively. These are bounds for the full experiment, without conditioning on success. To bound $\eta_{\rm impl}$ from the implemented recovery blocks, let $K$ and $\widetilde K$ be the ideal and implemented products of successful blocks, and let $M$ and $\widetilde M$ be the acceptance effects of the final test on the recovered register.  All four operators satisfy $\|K\|,\|\widetilde K\|\le1$ and $0\le M,\widetilde M\le I$.  Write $\eta_{\rm test}=\|\widetilde M-M\|$.  Uniformly random output on failure gives the acceptance effects
\[
F=\tfrac12I+K^\dagger(M-\tfrac12I)K,\qquad
\widetilde F=\tfrac12I+\widetilde K^\dagger(\widetilde M-\tfrac12I)\widetilde K.
\]
Since $\|M-\tfrac12I\|\le1/2$, comparing the factors one at a time gives
\[
\|\widetilde F-F\|\le\|\widetilde K-K\|+\eta_{\rm test}
\le\eta_{\rm be}+\eta_{\rm dec}+\eta_{\rm test}.
\]
For the channels that output one classical bit,
$\tfrac12\|\widetilde{\mathcal E}_x-\mathcal E_x\|_\diamond
=\|\widetilde F-F\|$: the upper bound follows by writing the Hermitian
difference between the acceptance effects as its positive and negative parts, and an eigenvector
attaining its operator norm gives equality.  Thus the error tolerances for the transformed blocks, the synthesis of
recovery operations, and the final test give a computable bound on
$\eta_{\rm impl}$ with the same polynomial overhead.
The total acceptance error is at most $\eta_{\rm in}+\eta_{\rm impl}$,
and the stated precision leaves half the ideal bias.  This argument applies
to the fixed final output measurement.

When a proof needs the recovered state conditioned on success, it suffices to make each of $\eta_{\rm in},\eta_{\rm be},\eta_{\rm dec}$ at most $p_{\rm chain}^{3/2}\delta/(100C_{\rm cond})$. This satisfies the required threshold for the error in the state conditioned on success and requires only inverse-polynomial precision when $p_{\rm chain},\delta\ge1/\poly(\nu)$. For a complete decision procedure, the bound $\eta_{\rm in}+\eta_{\rm impl}$ applies directly to the unconditional acceptance error.
\end{proof}

\paragraph{Spectral bounds for the constructed instances (Lemma~\ref{lem:constructed-quphls-spectral-promise}).}
\begin{proof}
The path spectrum gives $\cond(G_C)\le\kappa_{\rm clock}^{\max}(L)$; the construction with one offset per row and expander constraints gives $\cond(M_C)\le\kappa_{\rm DA}^{\max}(L)$. The base boundary matrix has normalization $\alphabe(B_2)=\sqrt{\Lambda_\ast}$. Its singular-value bounds, together with the direct sum for the added isolated simplices, give $\|B_2\|\le\alphabe(B_2)$ and
\[
\lambda_{\min}^+\!\left(\frac{B_2^\top B_2}{\alphabe(B_2)^2}\right)
\ge\min\left\{[C_{b_\ast}^{\rm bc}\kappa_{\rm DA}^{\max}(L)]^{-2},3/\Lambda_\ast\right\}
=[C_{b_\ast}^{\rm bc}\kappa_{\rm DA}^{\max}(L)]^{-2}.
\]
The equality uses $(C_{b_\ast}^{\rm bc})^2\ge\Lambda_\ast/3$ and $\kappa_{\rm DA}^{\max}(L)\ge1$. The added isolated triangles supply the second term; isolated edges only add zero rows. The normalized upper Hodge Laplacian $B_2B_2^\top/\alphabe(B_2)^2$ has the same positive spectrum. Lemma~\ref{lem:constructed-gap-promise} bounds the final reciprocal of the smallest positive eigenvalue by $\kappa_{\rm final}^{\max}(L)$ using the sparse-access bounds after each suspension. The boundary matrix entries remain $0,\pm1$, the embeddings preserve normalized solution states, and the specified block-encoding normalizations are $O_{k,b_\ast}(1)$.
\end{proof}

\paragraph{Hardness (Theorem~\ref{thm:hardness-final}).}
\begin{proof}
Start with the BQP-hardness of the Dirichlet propagation linear system problem \((G_C,c_C)\) of Theorem~\ref{thm:shifted-clock-hard}.  The input circuit generates the uniform circuits for the sparse-entry and right-hand-side oracles; it is not supplied separately to recovery, consistent with Remark~\ref{rem:circuit-dependent-oracles}.  Theorem~\ref{thm:shifted-to-da} maps it to a bounded-occurrence difference--average instance \((M_C,d_C)\) with the supplied rational bound
\[
\cond(M_C)\le\kappa_{\rm DA}^{\max}(L),
\]
and with a recovery map that succeeds with constant probability and recovers the weighted Feynman--Kitaev history state obtained by matrix inversion.  Apply Theorem~\ref{thm:da-bo-to-quphls-main} to map \((M_C,d_C)\) to a \(\UpperHodgeQLSP_k\) instance.  By Lemma~\ref{lem:constructed-gap-promise}, the reciprocal of the smallest positive eigenvalue of $L_k^\uparrow/\alphabe(L_k^\uparrow)$ is bounded above by
\[
\kappa_{\rm final}^{\max}(L)=\Xi^{\rm susp}_{k,b_\ast}\bigl(\kappa_B^{\max}(L)\bigr)^2.
\]
In particular, the constructed instance satisfies the stated promise
\[
\norm{L_k^\uparrow}\le\alphabe(L_k^\uparrow),
\qquad
\lambda_{\min}^+(L_k^\uparrow)\ge \alphabe(L_k^\uparrow)/\kappa_{\rm final}^{\max}(L).
\]

The fixed decision circuit starts with a supplied \(\UpperHodgeQLSP_k\) solution state and measures a success flag after each recovery step. First, it applies the adjoints of the suspension embeddings from Corollary~\ref{cor:dim-raise-explicit}. It then applies the \(\UpperHodgeQLSP_1\to\BoundaryQLSP\) recovery map by using a block encoding of \(B^\top\), as in Lemma~\ref{lem:bc-b2-to-quphls1-explicit}. Next, it projects onto the equal-copy subspace, as in Theorem~\ref{thm:bc-da-to-b2-wrapper-explicit}. Finally, it removes the auxiliary variables by applying \(D_{\rm aux}\), as in Lemma~\ref{lem:lg-mp-decoder}.

If every measured success flag equals one, the circuit performs the Feynman--Kitaev output test centered at acceptance probability $1/2$ from Theorem~\ref{thm:shifted-clock-hard}; otherwise it outputs a uniformly random bit.  Each block encoding is implemented using the final instance oracles, and all calls to those oracles, including state-preparation calls, are counted.  Lemma~\ref{lem:reverse-typed-oracle-simulation} shows how the required oracles for intermediate boundary matrices are simulated from the final indexed \(\UpperHodgeQLSP_k\) boundary matrix oracle with the normalization factors computed from the sparse-access bounds.

The reduction uses the following explicit lower bound from Lemma~\ref{lem:herald-bounds} on the probability that all recovery steps succeed:
\[
\underline p_{\rm chain}=\underline p_{\rm dim}\,\underline p_{B^\top}\,\underline p_{\rm bc}\,\underline p_{\rm aux}
\ge 1/\poly(L).
\]
For the exact recovery circuit, let $p_{\rm succ}$ be the probability that all its recovery steps succeed and let $p_{\rm tail}$ be the probability, conditional on successful recovery, that the clock lies in the final interval of appended identity gates.  Its ideal acceptance probability satisfies
\[
p_\star^{\rm id}=\frac{1-p_{\rm succ}}2+p_{\rm succ}p_{\rm FK},
\qquad
p_\star^{\rm id}-\frac12=p_{\rm succ}p_{\rm tail}\left(p_{\rm acc}(C)-\frac12\right).
\]
Thus the sign of the original acceptance bias is preserved, and the lower bounds on $p_{\rm succ}$ and $p_{\rm tail}$ give an inverse-polynomial decision gap.  Each recovery success event is recorded by a specified measured flag. This unconditional probability includes the failed branches.

Let $\Delta_{\rm FK}$ be a lower bound on the magnitude of the unconditional acceptance bias of the Feynman--Kitaev clock-output test, which outputs a uniformly random bit when the clock is outside the final interval of appended identity gates. This bias therefore includes the factor $p_{\rm tail}$.

Let $\eta_{\rm in}$ be the phase-invariant error of the supplied solution state, and let $\eta_{\rm impl}$ bound half the diamond-norm distance between the ideal and implemented decision channels, as in Lemma~\ref{lem:composed-heralded-accounting}.  These errors refer to the complete decision circuit, including the random output on failure.  By Lemma~\ref{lem:composed-heralded-accounting}, it suffices to choose the implementation precisions so that
\[
\eta_{\rm in}+\eta_{\rm impl}
\le
\frac12\underline p_{\rm chain}\Delta_{\rm FK}.
\]
Define the final decision promise by the thresholds $1/2\pm\Delta_{\rm final}$, where
\[
\Delta_{\rm final}:=\frac12\underline p_{\rm chain}\Delta_{\rm FK}.
\]
Writing $p_\star$ for the implemented acceptance probability, the magnitude of its acceptance bias obeys
\[
\left|p_\star-\frac12\right|
\ge
\underline p_{\rm chain}\Delta_{\rm FK}
-(\eta_{\rm in}+\eta_{\rm impl})
\ge
\Delta_{\rm final}.
\]
The acceptance bias remains positive on yes instances and negative on no instances.
The reciprocals of the smallest positive singular values and eigenvalues used in recovery are polynomially bounded in $L$. The block-encoding normalizations depend only on the fixed simplex dimension $k$ and occurrence bound $b_\ast$. Thus all recovery operations succeed with probability at least inverse polynomial, and the required precision has polynomial overhead. The reduction constructs the target description and all oracle circuits in classical polynomial time, proving $\BQP$-hardness of the sparse-oracle decision problem. Lemma~\ref{lem:finite-oracle-encoding-size} gives $\BQP$-hardness under classical polynomial-time many-one reductions for the circuit-encoded formulation. A procedure that prepares a sufficiently accurate solution state also solves this decision problem through the prescribed recovery and repetition. State-preparation hardness is therefore a consequence of the same reduction.
\end{proof}

\paragraph{Positive semidefinite linear systems (Theorem~\ref{thm:qsvt-final}).}
\begin{proof}
Put $F=H/\alphabe(H)$. Then $\spec^+(F)\subseteq[1/\cond,1]$, the supplied unitary encodes $F$ with normalization one and error $\varepsilon_{\rm be}/\alphabe(H)$, and $F^+=\alphabe(H)H^+$. We use this rescaled matrix in the analysis. The input matrix and target state remain $H$ and $H^+b/\|H^+b\|$. Let $\rho=\|\Pi_Hb\|\ge1/\zeta>0$. The algorithm needs only the supplied lower bound specified by $\zeta$.

For $0<\varepsilon\le1/2$, there is an odd polynomial $p$, bounded by $1$ on $[-1,1]$, with $p(0)=0$ and
\[
\begin{gathered}
d=O\!\left(\cond\log\frac{\cond\zeta}{\varepsilon}\right),\\
\left|p(x)-\frac1{2\cond x}\right|\le\delta_{\rm poly}
\quad\text{for }\frac1\cond\le|x|\le1,\\
v_{\rm id}=\frac{F^+\ket b}{2\cond}
=\frac{\alphabe(H)H^+\ket b}{2\cond},\qquad
\|v_{\rm id}\|\ge\frac{\rho}{2\cond}\ge\frac1{2\cond\zeta}.
\end{gathered}
\]
Apply the cited pseudoinverse theorem to $F/2$ with cutoff $1/(2\cond)\le1/2$ and approximation precision $\min\{\delta_{\rm poly},1/(2\cond)\}$, and set $p(x)=P(x/2)$ for its polynomial $P$. The polynomial approximates the operator giving the displayed $v_{\rm id}$. Its value $p(0)=0$ removes the component in the kernel. Let $\delta_{\rm QSVT}$ denote the operator-norm error of the successful QSVT block relative to $\alphabe(H)H^+/(2\cond)$. The cited robustness bound gives
\[
\delta_{\rm QSVT}\le\delta_{\rm poly}
+C_{\rm rob}d\sqrt{\varepsilon_{\rm be}/\alphabe(H)},
\qquad C_{\rm rob}=4.
\]
It suffices to choose
\[
\delta_{\rm poly}\le\frac{\varepsilon}{128\cond\zeta},\qquad
\varepsilon_{\rm be}\le\alphabe(H)
\left(\frac{\varepsilon}{128C_{\rm rob}\cond\zeta d}\right)^2.
\]
Choose the arithmetic and phase-synthesis precisions so that they contribute at most $\varepsilon/(128\cond\zeta)$ to the error in the transformed block. Over the amplified circuit, their accumulated operator-norm error is at most $\varepsilon/16$. The implemented successful branch then has norm at least $1/(4\cond\zeta)$ and, after normalization, phase-invariant error at most $\varepsilon/2$ relative to $H^+b/\|H^+b\|$.

Fixed-point amplitude amplification uses threshold $1/(8\cond\zeta)$, strictly below the proved lower bound on the branch amplitude, and failure probability at most $(\varepsilon/8)^2$, leaving residual failure amplitude at most $\varepsilon/8$. Measure the success flag, which equals one precisely when all designated auxiliary qubits for block encoding and matrix transformation are in their success states. Including the accumulated error in the amplified circuit, one implemented attempt fails with probability at most $(\varepsilon/8+\varepsilon/16)^2\le\varepsilon^2/16$. Conditional on success, its state has phase-invariant error at most $\varepsilon$. Repeating on failure increases the expected cost by a constant factor. The decision circuits below use one attempt with fixed worst-case gate and query bounds and include its failure outcome. Per amplified attempt, the block encoding, inverse, and controlled versions are used
\[
Q_{\rm BE}=O\!\left(\cond^2\zeta
\log\frac{\cond\zeta}{\varepsilon}\log\frac1\varepsilon\right)
\]
times, with polynomially many calls in the same parameters to $\mathcal O_b,\mathcal O_b^\dagger$. If a block-encoding use costs $q_{\rm be}$ sparse queries, the cost is $q_{\rm be}Q_{\rm BE}$, polynomial in $\nu$ under the numerical promises. The specified normalization determines the allowed matrix error and required arithmetic precision. Its effect on the rescaled spectral interval is already included in $\cond$.
This applies the QSVT pseudoinverse algorithm and fixed-point amplitude
amplification with explicit parameter bounds~\cite[full version, Theorem~41,
Lemma~22, and Theorem~27]{GSLW2019}. A formulation for possibly singular
Hermitian matrices with a promise on the norm of the projection onto the range is also given in
~\cite[full version, Corollary~31]{Chakraborty2018}.
\end{proof}

\paragraph{Block encodings of lower and full Hodge Laplacians (Theorem~\ref{thm:lower-full-containment}).}
\begin{proof}
For each boundary matrix $B_j=\partial_j$ that is present, sparse access supplies a block encoding with specified normalization $\alphabe(B_j)=\sqrt{s_{{\rm row},j}s_{{\rm col},j}}\ge1$, and $\|B_j\|\le\alphabe(B_j)$. The product and linear-combination-of-unitaries (LCU) constructions~\cite[full version, Lemmas~53 and~52]{GSLW2019} give
\[
L_k^\downarrow=B_k^\top B_k,\qquad
L_k^\uparrow=B_{k+1}B_{k+1}^\top,\qquad
\Delta_k=L_k^\downarrow+L_k^\uparrow.
\]
Their normalizations are $\alphabe(L_k^\downarrow)=\alphabe(B_k)^2$ and $\alphabe(L_k^\uparrow)=\alphabe(B_{k+1})^2$. Summing the squared boundary matrix normalizations over the boundary matrices that are present gives $\alphabe(\Delta_k)$. To obtain the unweighted sum, each normalized Gram block is weighted by its own normalization. The control state used by the LCU construction therefore has amplitude $\alphabe(B_j)/\sqrt{\alphabe(\Delta_k)}$ for the term from $B_j$. Consequently the successful block is
\[
\sum_{\substack{j\in\{k,k+1\}\\B_j\text{ present}}}
\frac{\alphabe(B_j)^2}{\alphabe(\Delta_k)}
\frac{L_k^{(j)}}{\alphabe(B_j)^2}
=\frac{\Delta_k}{\alphabe(\Delta_k)},
\quad
L_k^{(k)}=L_k^\downarrow,\quad L_k^{(k+1)}=L_k^\uparrow.
\]
Thus these weights preserve the unweighted full Hodge Laplacian even when the two boundary matrix normalizations differ. Boundary matrices known to be absent are omitted throughout.

For a required error $0<\varepsilon_{\rm be}\le1$ in the unnormalized matrix, choose the error for each boundary matrix to satisfy
\[
\varepsilon_{B_j}\le\frac{\varepsilon_{\rm be}}{8\alphabe(B_j)}.
\]
The resulting error in its Gram matrix is at most
\[
2\|B_j\|\varepsilon_{B_j}+\varepsilon_{B_j}^2
\le\frac{\varepsilon_{\rm be}}4+\frac{\varepsilon_{\rm be}^2}{64}.
\]
The sum of at most two such errors is at most
$\varepsilon_{\rm be}/2+\varepsilon_{\rm be}^2/32
\le17\varepsilon_{\rm be}/32$.
Allocate the remaining $15\varepsilon_{\rm be}/32$ to the error in the unnormalized matrix from synthesizing the LCU circuit. It suffices to synthesize the circuit to operator-norm error at most $15\varepsilon_{\rm be}/(32\alphabe(\Delta_k))$. All required arithmetic has polynomial cost. The lower Hodge Laplacian linear system uses a single Gram matrix and needs no LCU step.

The promised spectral interval applies to the relevant $M/\alphabe(M)$, so Theorem~\ref{thm:qsvt-final} applies to the actual prepared right-hand side. One use of the block encoding of $L_k^\downarrow$ costs $O(s_{{\rm row},k}+s_{{\rm col},k})$ sparse queries. One use for $\Delta_k$ costs $O(s_{{\rm row},k}+s_{{\rm col},k}+s_{{\rm row},k+1}+s_{{\rm col},k+1})$, omitting absent factors. The reflections used in amplitude amplification call both the right-hand-side preparation circuit and its inverse. Multiplying the cost of one block encoding by $Q_{\rm BE}$ gives polynomial query cost. On the hard family, each block encoding uses $O_k(1)$ sparse queries.
\end{proof}

\paragraph{Quantum algorithm (Theorem~\ref{thm:containment-final}).}
\begin{proof}
Lemma~\ref{lem:sparse-access-block-encoding} gives a block encoding of $B_{k+1}=\partial_{k+1}$ with the specified normalization
\[
\alphabe(B_{k+1})=\sqrt{s_{\rm row}s_{\rm col}},
\qquad \alphabe(L_k^\uparrow)=s_{\rm row}s_{\rm col}.
\]
The second identity follows by composing the boundary matrix block encoding with its adjoint, using the product construction~\cite[full version, Lemma~53]{GSLW2019} with the required projections. Choose the boundary matrix error so that
$2\|B_{k+1}\|\varepsilon_B+\varepsilon_B^2$ meets the tolerance of
Theorem~\ref{thm:qsvt-final}; the known bound
$\|B_{k+1}\|\le\alphabe(B_{k+1})$ suffices for this allocation.
The theorem applies to $L_k^\uparrow$ with its promised interval after division by $\alphabe(L_k^\uparrow)$ and to the actual prepared state $\gamma_x$ with the supplied lower bound on the norm of its projection onto the range. The success flag records that all designated auxiliary registers are in their success states. The prepared right-hand side may contain a kernel component.

For the promise decision problem with fixed recovery and measurement, Definition~\ref{def:measure-wrapper} gives the fixed recovery and measurement circuit $W_{x,\tau}$ access to precisely the instance oracles $\mathcal O_x$. Therefore any sparse-access or state-preparation queries used by $W_{x,\tau}$ are included in the query count: if a single call to $W_{x,\tau}^{\mathcal O_x}$ uses $q_W$ oracle queries, then the decision algorithm adds $q_W$ to the cost of each acceptance-probability sample.

For the hard family, the sparse-access bounds depend only on the fixed simplex dimension $k$ and occurrence bound $b_\ast$, including after suspension. Thus $\alphabe(B_{k+1})=O(1)$ and $\alphabe(L_k^\uparrow)=O(1)$. In general these normalizations are polynomially bounded by the supplied sparsities, and the normalized spectral parameter $\kappa$ accounts for their effect on the QSVT query cost. All queries used in recovery and measurement are also counted.
\end{proof}

\paragraph{Completeness (Corollary~\ref{cor:complete-final}).}
\begin{proof}
Run one amplified preparation attempt at accuracy $\varepsilon=\Delta/10$ by Theorem~\ref{thm:containment-final}.  On success execute the fixed recovery and measurement; on failure output an independently sampled uniformly random bit.  The implemented preparation fails with probability at most $(\Delta/40)^2$, which adds at most that amount to acceptance error.

For a chosen bound on the sampling failure probability $0<\delta_{\rm fail}<1$, estimate unconditional acceptance to additive error $\Delta/5$.  Lemma~\ref{lem:state-to-verifier-final} bounds the acceptance error due to the state produced on success by $\Delta/10$. If preparation and the subsequent recovery and measurement use $q_{\rm sol}$ and $q_W$ instance queries, sampling uses $O(\Delta^{-2}\log(1/\delta_{\rm fail})(q_{\rm sol}+q_W))$ queries in the worst case.  Taking $\delta_{\rm fail}=1/10$ gives polynomial query complexity in $\nu$, and the bound on the state error preserves the promised gap.

Together with Theorem~\ref{thm:hardness-final}, this proves sparse-oracle hardness and a polynomial-query quantum algorithm. Lemmas~\ref{lem:circuit-input-simulation} and~\ref{lem:finite-oracle-encoding-size} give membership for the finite-string problem and a polynomial-time many-one hardness reduction, respectively. $\BQP$-completeness under classical polynomial-time many-one reductions therefore holds for the designated circuit-encoded problem; in the sparse-oracle formulation, the reduction constructs the target description and all oracle circuits in classical polynomial time.
\end{proof}

\Needspace{9\baselineskip}
\subsubsection{Lower and full Hodge Laplacians}

\paragraph{Preparing the right-hand side for the lower Hodge Laplacian linear system (Lemma~\ref{lem:lower-rhs-preparation}).}
\begin{proof}
On the hard instances, $\gamma\in\im(B_k)$ has equal-magnitude amplitudes, with signs, on a support of size $s_\gamma=3\cdot2^{k-2}$. The nonzero entries of $B_k=\partial_k$ are $\pm1$. Let $\kappa_B$ be a constructed bound on $\alphabe(B_k)/\sigma_{\min}^+(B_k)$.  Use the final $B_k$ oracle to enumerate the rows in the support of the right-hand side. Group repeated column labels and sum the signed incidences exactly.  This gives
\[
\begin{gathered}
(B_k^\top\gamma)_j=\frac1{\sqrt{s_\gamma}}m_j,\quad m_j\in\mathbb Z,\quad |m_j|\le s_\gamma,\\
\|B_k^\top\gamma\|\ge\alphabe(B_k)/\kappa_B>0,
\end{gathered}
\]
with $O_k(1)$ nonzero $m_j$. Cancellation can be tested exactly.  After normalization the amplitudes are $m_j/\sqrt{\sum_\ell m_\ell^2}$.  These amplitudes belong to a finite family of algebraic numbers independent of circuit length.  Under Remark~\ref{rem:implementation-gate-set}, fix an exact preparation unitary for each pattern in a constant-size support register.  Reversible lookup writes the full simplex labels. Reversible search among the distinct labels then erases the support index and returns the workspace to zero.  Thus preparation is exact under the stated gate convention and has cost polynomial in the address width.  All final-oracle queries are counted; the reduction for circuit-encoded inputs generates the lookup circuits.  Compilation to another universal gate set contributes an error in simulating the complete circuit as specified in that remark.  Lemma~\ref{lem:rhs-robustness} also permits approximate prepared inputs.
\end{proof}

\paragraph{Recovery of solution states for boundary linear systems after suspension (Lemma~\ref{lem:explicit-boundary-desuspension}).}
\begin{proof}
The block decomposition of the suspended boundary matrix in Lemma~\ref{lem:suspension-block-explicit} gives
\[
B_{i+1}^{\Sigma}K=-JB_i,
\qquad
(B_{i+1}^{\Sigma})^\top J=-KB_i^\top.
\]
Each positive singular pair of $B_i$ therefore becomes an antisymmetric singular pair with the same singular value. Expanding $\gamma\in\im(B_i)$ in these pairs proves the pseudoinverse identity.
Since $K\ket x=\ket-\ket x$, a Hadamard followed by a bit flip resets the cone register. The unitary preserves the error in the joint state. Later recovery maps act as the identity on any cone registers that are retained. If the registers are discarded, trace distance cannot increase. Bounding it by the phase-invariant distance of the joint pure states gives the claim.
\end{proof}

\paragraph{Acceptance bias after recovery for lower and full Hodge Laplacian instances (Lemma~\ref{lem:lower-full-recovery-gap}).}
\begin{proof}
For the reference solution of the boundary linear system, inverse suspension and equal-copy recovery succeed with probability one by Lemma~\ref{lem:explicit-boundary-desuspension} and Theorem~\ref{thm:bc-mp-main}. Lemma~\ref{lem:bc-decoder-stability} gives the conservative equal-copy success bound $1/2$ for input error below $1/4$, and Lemma~\ref{lem:lg-mp-decoder} gives success probability at least $\underline p_{\rm aux}$. The Feynman--Kitaev test of Theorem~\ref{thm:shifted-clock-hard} has bias at least $1/48$, including its independently sampled uniformly random output when the clock lies outside the appended-identity interval. The same random output on recovery failure and Lemma~\ref{lem:composed-heralded-accounting} therefore give the following lower bound on the ideal acceptance bias
\[
(1/2)\underline p_{\rm aux}(1/48)=\Delta_{\rm base}^{\downarrow,\mathrm{full}}.
\]
Allowing recovery and arithmetic implementation errors to use half of this bias leaves $\Delta_{\rm rec}$; the right-hand-side preparation error is budgeted separately in~\eqref{eq:lower-rhs-hardness-tolerance}. On the component that encodes the hard full Hodge Laplacian instance, $\Delta_k=L_k^\downarrow$, so the same recovery circuit applies.
\end{proof}

\paragraph{Hardness for lower and full Hodge Laplacian linear systems (Theorem~\ref{cor:lower-full-hodge-hardness}).}
\begin{proof}
Starting with the hard $B_2$, apply $k-2$ suspensions to obtain the unweighted degree-$k$ boundary matrix, for $k\ge2$. Lemma~\ref{lem:explicit-boundary-desuspension} identifies the normalized solution of the boundary linear system at each step, up to global phase. Lemma~\ref{lem:suspension-spectrum-explicit} bounds all positive singular values, including those outside the antisymmetric subspace, and accounts for the sparse-access normalization at each step. Each query to the final boundary matrix uses a constant number of source queries. For the lower Hodge Laplacian one may take
\[
\kappa_\downarrow
=\Xi^{\rm susp}_{k-1,b_\ast}\bigl(\kappa_B^{\max}(L)\bigr)^2,
\qquad
\spec^+\!\left(\frac{L_k^\downarrow}{\alphabe(L_k^\downarrow)}\right)
\subseteq[1/\kappa_\downarrow,1].
\]
Here the empty suspension product is one for $k=2$, and
$\alphabe(L_k^\downarrow)=\alphabe(B_k)^2$.

For normalized $\gamma\in\im(B_k)$, put $L_k^\downarrow=B_k^\top B_k$ and $\eta=B_k^\top\gamma$. Then $\eta\in\im(L_k^\downarrow)$ and
\[
\|\eta\|\ge\sigma_{\min}^+(B_k)
\ge\frac{\alphabe(B_k)}{\sqrt{\kappa_\downarrow}}>0.
\]
Lemma~\ref{lem:standard-mp-identities} gives
\[
(L_k^\downarrow)^+B_k^\top\gamma=B_k^+\gamma.
\]
Thus the normalized reference solution for the lower Hodge Laplacian linear system is exactly the normalized solution of the hard boundary linear system. More generally, a bound $\kappa_B\ge\alphabe(B_k)/\sigma_{\min}^+(B_k)$ gives $\kappa_B^2$ as the parameter for the lower Hodge Laplacian linear system.

Lemma~\ref{lem:lower-rhs-preparation} constructs the circuit for $\eta_x$ and ensures that its workspace returns to zero. Thus the reduction implements preparation of an approximation to the normalized vector proportional to $B_k^\top\gamma$. The task uses this actual prepared state. Lemma~\ref{lem:rhs-robustness} bounds the distance between its normalized solution and the normalized reference solution by $4\kappa_\downarrow\delta_{\rm rhs}$ and gives $\zeta_\eta\le2$. By~\eqref{eq:lower-rhs-hardness-tolerance}, the right-hand-side approximation contributes at most $\varepsilon/4$ to the error in the solution state and at most $\Delta_{\rm rec}/4$ to acceptance-probability error. Recovery first reverses the suspension encodings and retains the cone registers. It then applies the equal-copy recovery map and the map that removes the auxiliary variables. Lemmas~\ref{lem:lower-full-recovery-gap} and~\ref{lem:state-to-verifier-final} leave an acceptance bias of magnitude at least inverse polynomial, with the correct sign; Lemma~\ref{lem:reverse-typed-oracle-simulation} counts recovery queries using only the final oracle.

For the full Hodge Laplacian linear system, choose the hard complex of dimension $k$, so $C_{k+1}=0$ and $\Delta_k=L_k^\downarrow$. The boundary matrix $B_{k+1}$, known to be absent, is omitted from the encoding, hence $\alphabe(\Delta_k)=\alphabe(L_k^\downarrow)$ and the same normalized spectral promise applies.

For Corollary~\ref{cor:full-both-nonzero}, adjoin a disjoint $(k+1)$-simplex with all faces and zero right-hand-side support. Its full Hodge Laplacian on $k$-simplices is $(k+2)I$. Each diagonal receives $k+1$ lower incidences and one upper incidence, and the off-diagonal contributions cancel. Thus the unweighted full Hodge Laplacian is the orthogonal direct sum of the hard lower Hodge Laplacian and $(k+2)I$. The fixed sparse-access bounds give
\[
\alphabe(L_{k,\rm hard}^\downarrow)=4T_\ast(k+1).
\]
The simplex does not increase the row or column sparsity bound for $B_k$, and its boundary matrix $B_{k+1}$ has row sparsity bound one and column sparsity bound $k+2$, so
\[
\alphabe(\Delta_{k,\rm aug})
=\alphabe(L_{k,\rm hard}^\downarrow)+k+2.
\]
Consequently,
\[
\begin{aligned}
\lambda_{\min}^+\!\left(
\frac{\Delta_{k,\rm aug}}{\alphabe(\Delta_{k,\rm aug})}
\right)
&\ge\min\left\{
\frac{\alphabe(L_{k,\rm hard}^\downarrow)}
{\alphabe(\Delta_{k,\rm aug})\kappa_\downarrow},
\frac{k+2}{\alphabe(\Delta_{k,\rm aug})}
\right\}\\
&\ge\frac1{\kappa_{\rm full}},
\end{aligned}
\]
where
\[
\kappa_{\rm full}:=\max\left\{1,
\frac{\alphabe(\Delta_{k,\rm aug})}
{\alphabe(L_{k,\rm hard}^\downarrow)},
\frac{\alphabe(\Delta_{k,\rm aug})}{k+2}
\right\}\kappa_\downarrow.
\]
Both normalizations depend only on $k,b_\ast$. Thus the task parameter changes by only a fixed factor. The norm bound from sparse block encoding gives the upper endpoint one. The normalized reference solution extends by zero on the simplex. Choose the actual prepared right-hand side using~\eqref{eq:lower-rhs-hardness-tolerance} with $\kappa_{\rm full}$ in place of $\kappa_\downarrow$; Lemma~\ref{lem:rhs-robustness} then gives the same bounds on the error in the solution state and acceptance probability for the full Hodge Laplacian. The quantum algorithms, including all queries to prepare the right-hand side, are supplied by Theorem~\ref{thm:lower-full-containment}.

The construction produces the target description and all oracle circuits
in classical polynomial time, as in Lemma~\ref{lem:finite-oracle-encoding-size}.
The preserved acceptance bias proves hardness of the prescribed decision
problem. The quantum algorithm and repeated sampling give membership for
the circuit-encoded problem by Lemma~\ref{lem:circuit-input-simulation},
so this problem is $\BQP$-complete. Repeated preparation of accurate solution states, followed by the same
recovery and measurement, would solve it with bounded error. This
establishes the claimed state-preparation hardness.
\end{proof}

\bibliographystyle{unsrt}
\bibliography{main}

\begin{thebibliography}{10}

\bibitem{BattistonReview}
Federico Battiston, Giulia Cencetti, Iacopo Iacopini, Vito Latora, Maxime Lucas, Alice Patania, Jean-Gabriel Young, and Giovanni Petri.
\newblock Networks beyond pairwise interactions: Structure and dynamics.
\newblock {\em Physics Reports}, 874:1--92, 2020.

\bibitem{Battiston2021}
Federico Battiston, Enrico Amico, Alain Barrat, Ginestra Bianconi, Guilherme Ferraz~de Arruda, Benedetta Franceschiello, Iacopo Iacopini, Sonia K{\'e}fi, Vito Latora, Yamir Moreno, Micah~M. Murray, Tiago~P. Peixoto, Francesco Vaccarino, and Giovanni Petri.
\newblock The physics of higher-order interactions in complex systems.
\newblock {\em Nature Physics}, 17:1093--1098, 2021.

\bibitem{IacopiniKarsaiBarrat2024}
Iacopo Iacopini, M{\'a}rton Karsai, and Alain Barrat.
\newblock The temporal dynamics of group interactions in higher-order social networks.
\newblock {\em Nature Communications}, 15:7391, 2024.

\bibitem{Grilli2017}
Jacopo Grilli, Gy{\"o}rgy Barab{\'a}s, Matthew~J. Michalska-Smith, and Stefano Allesina.
\newblock Higher-order interactions stabilize dynamics in competitive network models.
\newblock {\em Nature}, 548:210--213, 2017.

\bibitem{SkardalArenas2019}
Per~Sebastian Skardal and Alex Arenas.
\newblock Abrupt desynchronization and extensive multistability in globally coupled oscillator simplexes.
\newblock {\em Physical Review Letters}, 122:248301, 2019.

\bibitem{Torres2021}
Leo Torres, Ann~S. Blevins, Danielle Bassett, and Tina Eliassi-Rad.
\newblock The why, how, and when of representations for complex systems.
\newblock {\em SIAM Review}, 63(3):435--485, 2021.

\bibitem{Bianconi2021}
Ginestra Bianconi.
\newblock {\em Higher-Order Networks: An Introduction to Simplicial Complexes}.
\newblock Cambridge University Press, 2021.

\bibitem{Bick2023}
Christian Bick, Elizabeth Gross, Heather~A. Harrington, and Michael~T. Schaub.
\newblock What are higher-order networks?
\newblock {\em SIAM Review}, 65(3):686--731, 2023.

\bibitem{JiangHodgeRank}
Xiaoye Jiang, Lek-Heng Lim, Yuan Yao, and Yinyu Ye.
\newblock Statistical ranking and combinatorial {Hodge} theory.
\newblock {\em Mathematical Programming}, 127:203--244, 2011.

\bibitem{DesbrunDEC2005}
Mathieu Desbrun, Anil~N. Hirani, Melvin Leok, and Jerrold~E. Marsden.
\newblock Discrete exterior calculus.
\newblock arXiv:math/0508341, 2005.

\bibitem{ArnaudonKuramoto}
Alexis Arnaudon, Robert~L. Peach, Giovanni Petri, and Paul Expert.
\newblock Connecting {Hodge} and {Sakaguchi--Kuramoto} through a mathematical framework for coupled oscillators on simplicial complexes.
\newblock {\em Communications Physics}, 5:211, 2022.

\bibitem{Nurisso2024}
Marco Nurisso, Alexis Arnaudon, Maxime Lucas, Robert~L. Peach, Paul Expert, Francesco Vaccarino, and Giovanni Petri.
\newblock A unified framework for simplicial {Kuramoto} models.
\newblock {\em Chaos}, 34(5):053118, 2024.

\bibitem{HorakJost2013}
Danijela Horak and J{\"u}rgen Jost.
\newblock Spectra of combinatorial {Laplace} operators on simplicial complexes.
\newblock {\em Advances in Mathematics}, 244:303--336, 2013.

\bibitem{HHL2009}
Aram~W. Harrow, Avinatan Hassidim, and Seth Lloyd.
\newblock Quantum algorithm for linear systems of equations.
\newblock {\em Physical Review Letters}, 103:150502, 2009.

\bibitem{CKS2017}
Andrew~M. Childs, Robin Kothari, and Rolando~D. Somma.
\newblock Quantum algorithm for systems of linear equations with exponentially improved dependence on precision.
\newblock {\em SIAM Journal on Computing}, 46:1920--1950, 2017.

\bibitem{GSLW2019}
Andr{\'a}s Gily{\'e}n, Yuan Su, Guang~Hao Low, and Nathan Wiebe.
\newblock Quantum singular value transformation and beyond: Exponential improvements for quantum matrix arithmetics.
\newblock In {\em Proceedings of the 51st Annual ACM SIGACT Symposium on Theory of Computing (STOC 2019)}, pages 193--204, 2019.
\newblock Full version: arXiv:1806.01838.

\bibitem{LedittoHodgeRank}
Caesnan M.~G. Leditto, Angus Southwell, Behnam Tonekaboni, Muhammad Usman, and Kavan Modi.
\newblock Quantum {HodgeRank}: Topology-based rank aggregation on quantum computers.
\newblock {\em Physical Review Applied}, 23:L061005, 2025.

\bibitem{LedittoOscillators}
Caesnan M.~G. Leditto, Angus Southwell, Muhammad Usman, and Kavan Modi.
\newblock Efficient quantum algorithms for higher-order coupled oscillators.
\newblock arXiv:2604.20108, 2026.

\bibitem{WatrousSurvey}
John Watrous.
\newblock Quantum computational complexity.
\newblock In Robert~A. Meyers, editor, {\em Encyclopedia of Complexity and Systems Science}, pages 7174--7201. Springer, 2009.
\newblock Preprint: arXiv:0804.3401.

\bibitem{OrsucciDunjko2021}
Davide Orsucci and Vedran Dunjko.
\newblock On solving classes of positive-definite quantum linear systems with quadratically improved runtime in the condition number.
\newblock {\em Quantum}, 5:573, 2021.

\bibitem{DingFull}
Ming Ding, Rasmus Kyng, Maximilian {Probst Gutenberg}, and Peng Zhang.
\newblock Hardness results for laplacians of simplicial complexes via sparse-linear equation complete gadgets.
\newblock In {\em 49th International Colloquium on Automata, Languages, and Programming (ICALP 2022)}, volume 229 of {\em Leibniz International Proceedings in Informatics}, pages 53:1--53:19, 2022.
\newblock Full version: arXiv:2202.05011.

\bibitem{KitaevShenVyalyi2002}
Alexei~Yu. Kitaev, Alexander Shen, and Mikhail~N. Vyalyi.
\newblock {\em Classical and Quantum Computation}, volume~47 of {\em Graduate Studies in Mathematics}.
\newblock American Mathematical Society, 2002.

\bibitem{AharonovEtAl2007}
Dorit Aharonov, Wim van Dam, Julia Kempe, Zeph Landau, Seth Lloyd, and Oded Regev.
\newblock Adiabatic quantum computation is equivalent to standard quantum computation.
\newblock {\em SIAM Journal on Computing}, 37(1):166--194, 2007.

\bibitem{BabbushOscillators}
Ryan Babbush, Dominic~W. Berry, Robin Kothari, Rolando~D. Somma, and Nathan Wiebe.
\newblock Exponential quantum speedup in simulating coupled classical oscillators.
\newblock {\em Physical Review X}, 13:041041, 2023.

\bibitem{BCK2015}
Dominic~W. Berry, Andrew~M. Childs, and Robin Kothari.
\newblock Hamiltonian simulation with nearly optimal dependence on all parameters.
\newblock In {\em Proceedings of the 56th IEEE Symposium on Foundations of Computer Science (FOCS 2015)}, pages 792--809, 2015.

\bibitem{Bennett1973}
Charles~H. Bennett.
\newblock Logical reversibility of computation.
\newblock {\em IBM Journal of Research and Development}, 17(6):525--532, 1973.

\bibitem{BeraEtAl2010}
Debajyoti Bera, Stephen Fenner, Frederic Green, and Steve Homer.
\newblock Efficient universal quantum circuits.
\newblock {\em Quantum Information and Computation}, 10(1--2):16--27, 2010.

\bibitem{Penrose1955}
Roger Penrose.
\newblock A generalized inverse for matrices.
\newblock {\em Proceedings of the Cambridge Philosophical Society}, 51(3):406--413, 1955.

\bibitem{WatrousTQI}
John Watrous.
\newblock {\em The Theory of Quantum Information}.
\newblock Cambridge University Press, 2018.

\bibitem{Dragomir2007Norm}
Sever~S. Dragomir.
\newblock On some inequalities in normed linear spaces.
\newblock {\em RGMIA Research Report Collection}, 10(Supplement), 2007.

\bibitem{Belenchia2013}
Alessio Belenchia, Giacomo~Mauro D'Ariano, and Paolo Perinotti.
\newblock Universality of computation in real quantum theory.
\newblock {\em EPL}, 104:20006, 2013.

\bibitem{KuperbergSK}
Greg Kuperberg.
\newblock Breaking the cubic barrier in the {Solovay--Kitaev} algorithm.
\newblock arXiv:2306.13158v2, 2025.

\bibitem{McKague2011}
Matthew McKague.
\newblock On the power quantum computation over real hilbert spaces.
\newblock arXiv:1109.0795, 2011.

\bibitem{YuehInverse2006}
Wen-Chyuan Yueh.
\newblock Explicit inverses of several tridiagonal matrices.
\newblock {\em Applied Mathematics E-Notes}, 6:74--83, 2006.

\bibitem{YuehSpectrum2005}
Wen-Chyuan Yueh.
\newblock Eigenvalues of several tridiagonal matrices.
\newblock {\em Applied Mathematics E-Notes}, 5:66--74, 2005.

\bibitem{AlonEverySize}
Noga Alon.
\newblock Explicit expanders of every degree and size.
\newblock {\em Combinatorica}, 41(4):447--463, 2021.

\bibitem{Fiedler1973}
Miroslav Fiedler.
\newblock Algebraic connectivity of graphs.
\newblock {\em Czechoslovak Mathematical Journal}, 23(2):298--305, 1973.

\bibitem{Hatcher2002}
Allen Hatcher.
\newblock {\em Algebraic Topology}.
\newblock Cambridge University Press, 2002.

\bibitem{Chakraborty2018}
Shantanav Chakraborty, Andr{\'a}s Gily{\'e}n, and Stacey Jeffery.
\newblock The power of block-encoded matrix powers: Improved regression techniques via faster {Hamiltonian} simulation.
\newblock In {\em 46th International Colloquium on Automata, Languages, and Programming (ICALP 2019)}, pages 33:1--33:14, 2019.
\newblock Full version: arXiv:1804.01973v2.

\end{thebibliography}
\end{document}